\documentclass[11pt]{article}%
\usepackage{amsfonts}
\usepackage{amsmath,amssymb,amsthm,enumerate,epsfig,graphicx,natbib}
\usepackage{ifthen,latexsym,syntonly}
\usepackage{amsmath,amssymb,booktabs,multirow,makecell,graphicx}
\usepackage{natbib}
\usepackage{rotating}
\usepackage{multirow}
\usepackage{booktabs}
\usepackage{tabularx}
\usepackage{lscape}
\usepackage{ragged2e}
\usepackage{threeparttable}
\usepackage{amsmath}
\usepackage{caption}
\usepackage{array}
\usepackage{float}
\usepackage{amssymb}
\usepackage[normalem]{ulem}
\usepackage{graphicx}
\usepackage{graphicx}
\usepackage{subcaption}
\usepackage{color}
\usepackage{tikz}
\usepackage{placeins}
\usepackage{xr}
\usepackage[bottom]{footmisc}%
\providecommand{\U}[1]{\protect\rule{.1in}{.1in}}

\usetikzlibrary{arrows.meta,positioning,fit,backgrounds}
\useunder{\uline}{\ul}{}
\newtheorem{theorem}{Theorem}
\newtheorem{corollary}{Corollary}

\newtheorem{lemma}{Lemma}
\newtheorem{proposition}{Proposition}
\newtheorem{assumption}{Assumption}
\newtheorem{algorithm}{Algorithm}
\theoremstyle{definition}

\newtheorem{example}{Example}

\allowdisplaybreaks
\begin{document}
	
	\title{Estimation and Inference for Latent Markov Models \\ by Fourier Recursions}
    \author{
\normalsize Yanqi Huang$^*$ \qquad
Chenxu Li$^*$ \qquad Qiwei Yao$^\dag$\\
$^*$Guanghua School of Management, Peking University\\
$^\dag$Department of Statistics, London School of Economics
}
	\date{\today }
	\maketitle
	
	\begin{abstract}
	
This paper proposes a new theoretically exact Fourier recursion framework for a broad class of latent Markov models (LMMs), encompassing models widely used across a broad range of fields in economics.
It can be viewed as a counterpart of the celebrated Kalman filter for non-Gaussian and nonlinear LMMs.
Closed-form recursive updates of Fourier coefficients jointly deliver filtering, likelihood evaluation, and simultaneously accumulate the score and Hessian online. We introduce a unified truncated implementation that ensures uniform error control and numerical stability, preventing approximation errors from accumulating through the recursion. We establish asymptotic properties of the feasible maximum likelihood estimator for LMMs, provide a recursion-based consistent estimator for the Fisher information and discuss the possible model misspecification. These results provide the first general asymptotic theory for feasible approximate maximum likelihood estimation in LMMs. Simulations demonstrate its accuracy
		and stability, while
		an application to U.S. bankruptcy data recovers a persistent latent
		bankruptcy-pressure process.

		%
		
		\medskip.
		
	\end{abstract}
	
	
	\section{Introduction}
	
	Latent Markov models (LMMs), defined by a latent variable and an observed
	variable that jointly form a Markov process, provide a general framework for
	dynamic systems in which persistent unobserved states generate observable
	signals. For example, state-space representations of dynamic stochastic
	general equilibrium (DSGE) and heterogeneous-agent models use latent states to
	summarize structural shocks and equilibrium dynamics; stochastic
	volatility models use latent processes to describe the evolution of financial
	market volatility. LMMs are also frequently considered when the underlying structural
	mechanism is too complex or insufficiently observed to be modeled directly, where the latent state usually serves as a reduced-form summary of persistent
	economic conditions. Examples include regime-switching
	models, dynamic factor models, and recurrent neural
	networks (RNNs) in machine learning models, which use hidden states to retain
	information from past observations and capture temporal dependence. 
	
	A canonical solution to estimation and inference for
     linear-Gaussin LMMs is provided by the Kalman filter, which exploits the first two conditional moments to deliver an exact and simple recursion.
    A long-standing objective has therefore been to extend this recursive paradigm beyond Gaussian  
    setting. This task is fundamentally difficult.\footnote{A direct extension would be to track higher-order moments. Yet, moments provide only a limited characterization of a distribution: even the full moment sequence need not uniquely determine the underlying distribution, let alone any finite collection of moments (see, e.g., \cite{schmudgen2017moment}).} As LMMs become increasingly complex, existing extensions tend either to rely on restrictive model structures or approximations, limiting their generality and accuracy, or to sacrifice the numerical robustness and analytical transparency that make the Kalman filter attractive in the first place. Can one nevertheless retain the recursive logic and \textit{exactness} of the Kalman filter for a broad class of LMMs? Can this exact representation be converted into a \textit{feasible} procedure without sacrificing rigorous likelihood-based estimation and inference? We provide answers to those questions.
	
	\subsection{Fourier recursion for estimation and inference in LMMs}
	
	The objective of this paper is to develop a unified, deterministic framework for
	online filtering, likelihood evaluation, as well as parameter estimation and
	inference for LMMs that is both accurate and computationally stable. To this
	end, we propose a theoretically exact closed-form Fourier coefficients recursion for a broad class of LMMs, where both
	filtering and likelihood evaluation can be directly recovered from the
	associated Fourier coefficients without additional treatment. Conceptually,
	this approach shares a close connection with sieve estimation in nonparametric
	econometrics (see, e.g., \cite{chen2012estimation, chen2014sieve}), by extending it to a parametric setting.
	
	Our methodology applies broadly to LMMs with a well-defined one-step latent-observation transition: for continuous latent states, the transition admits a sufficiently smooth density, while for discrete latent states, the latent state space is finite. This class includes state-space models (SSMs), hidden
	Markov models (HMMs), as well as models whose transition dynamics are
	specified through transition characteristic functions, e.g., the affine models
	(see, e.g., \cite{duffie2000transform,bates2006maximum}). The key innovation lies in representing the evolving filtering density via
	its associated Fourier series, allowing the corresponding coefficients to be
	updated recursively as new observations arrive. This framework fundamentally
	transforms the mechanics of filtering by replacing the intractable multiple
	integrals of classical recursions with a sequence of algebraic
	multiplications. Moreover, the closed-form Fourier coefficients recursion can itself be
	differentiated recursively w.r.t. the model parameters. This yields
	deterministic online updates for the gradient and the Hessian matrix of the
	sample log-likelihood, thereby supporting gradient-based maximum marginal likelihood
	estimation (MMLE) and its subsequent statistical inference. The implementation is
	designed so that even after truncation, both the estimated filtered density
	and the evaluated likelihood remain strictly positive, which prevents invalid
	evaluations and stabilizes the recursive updates. By
	exploiting the analytical properties of the Fourier basis, our approach
	ensures a rapid decay of truncation errors, thereby guaranteeing favorable asymptotic properties of the resulting approximate estimators. 
	
	\subsection{Why estimation and inference for LMMs are essential and
		challenging\label{sec:related work}}
	
	Statistical analysis of these models centers on two closely related
	objectives: latent-state inference, typically through filtering to recover latent variables, parameter estimation
	and inference, which characterize the underlying dynamic mechanisms, quantify
	uncertainty, enable the testing of economically relevant hypotheses and
	model restrictions. The growing scope and complexity of LMMs make these tasks
	both increasingly essential and increasingly challenging. The associated literature has developed along two related directions: approximation methods and asymptotic inference.
	
	These approximation approaches can be grouped into three main categories. The
	first category builds on the Kalman filter. A variety of extensions have been
	developed for nonlinear systems. The extended Kalman filter (EKF) and the
	unscented Kalman filter (UKF) are two extensions (see, e.g., Chapter 7 of
	\cite{TsayChen2019NonlinearTimeSeries}). Relatedly, \cite{bates2006maximum}  uses the first two conditional moments to parameterize a Gamma approximation to the filtering distribution. Their applicability, however,
	depends strongly on the specific structure of the model. In nonlinear or
	non-Gaussian settings, local linearization or Gaussian-type approximations may
	lead to non-negligible filtering and likelihood approximation errors, which may
	lead to biased parameters estimations. The second strand is based
	on sequential Monte Carlo methods, commonly referred to as particle filters.
	Particle filters represent the filtering distribution by a weighted collection
	of simulated particles that are propagated and updated sequentially as new
	observations arrive. The bootstrap particle filter (BPF) and the auxiliary
	particle filter (APF) are two commonly used variants (see, e.g., Chapter 8 of
	\cite{TsayChen2019NonlinearTimeSeries}). However, applying particle
	filters comes with Monte Carlo variability, particle degeneracy, and
	additional model-specific choices concerning the number of particles, proposal
	distributions, and resampling schemes. These issues become particularly
	important for parameter estimation. The resulting simulated likelihood may be
	noisy or insufficiently smooth, making numerical optimization and the
	computation of likelihood derivatives or Fisher
	information\footnote{Kalman-filter-type approximations do not generally
		deliver these quantities with comparable accuracy, whereas the Monte Carlo
		variability of particle-based likelihoods may produce a noisy or
		insufficiently smooth objective function and make optimization or asymptotic
		covariance matrix estimation substantially more difficult (see, e.g., Chapter
		8 of \cite{TsayChen2019NonlinearTimeSeries}).} more
	difficult (see e.g.,
	\cite{MalikPitt2011ParticleFiltersContinuous,PoyiadjisDoucetSingh2011}). The
	third strand combines latent-state simulation with Markov chain Monte Carlo
	(MCMC) methods for Bayesian parameter inference. Two representative approaches
	are pseudo-marginal MCMC and particle MCMC (see, e.g.,
	\cite{AndrieuDoucetHolenstein2010}). Their implementation,
	however, typically involves additional model-specific choices, including
	proposal design, chain length, and convergence diagnostics. Moreover, because
	these methods are primarily simulation-based and offline, they are less
	naturally suited to recursive likelihood updating and to the direct online
	computation of the Fisher information.
	
	The theoretical literature for SSMs and HMMs, which typically exploits their
	factorization structures, studies the propagation of filtering errors and the
	asymptotic properties of likelihood-based estimators. Existing results provide
	conditions for filtering stability and approximation-error control, as well as
	the consistency and asymptotic normality of MLEs (see, e.g.,
	\cite{gland2004,cappe2005inference}). Related likelihood theory has also been
	developed for general Markov regime-switching models (see, e.g.,
	\cite{Pouzo2022}). However, these results rely on the special factorization
	structures of SSMs and HMMs and therefore do not readily extend to general
	LMMs. Moreover, the exact MMLE is typically unavailable in practice, and
	existing methods instead deliver an approximated estimator. In contrast to the
	extensive theory for exact estimators, the asymptotic properties of such
	feasible parameter estimators remain relatively underdeveloped.
	
	\subsection{Contributions and positioning}
	
	The paper makes three contributions. First, it provides a theoretically exact closed-form Fourier
	recursion for a wide class of LMMs and shows how the filter, likelihood
	function together with its score and Hessian matrix, can be computed
	recursively, facilitating the associated estimation and inference for LMMs.
	The implementation is designed for numerical resilience and reliability: it
	provides a global approximation to the filtering density and preserves the
	positivity of both the approximated filtering density and the evaluated likelihood.
	
	Second, the paper develops theoretical guarantees beyond the SSM and HMM
	settings, with particular emphasis on feasible approximate estimators. The
	analysis does not rely on the special structure of SSMs or HMMs, and
	establishes uniform error control for the filters obtained from
	the truncated Fourier recursion in a general latent Markov environment. It
	then derives the asymptotic properties of the approximated MMLEs, including
	consistency and asymptotic normality. These results clarify when the
	approximation error induced by truncating the Fourier expansion is
	asymptotically negligible and when the feasible MMLE inherits the large-sample
	behavior of the infeasible exact MMLE. The proposed Fourier
	recursion provides a consistent estimator for the MMLE's asymptotic covariance matrix. Moreover, the asymptotic theory is extended to the misspecified case, where the estimator converges to the pseudo-true parameter.
	
	Third, this paper evaluates the proposed method through extensive simulations
	and an empirical application. In nonlinear, non-Gaussian environments,
	including the filtering of RNNs, it remains highly accurate while eliminating
	degeneracy problems of particle filters. Also, our framework successfully
	achieves MMLEs for the Barndorff-Nielsen and Shephard stochastic volatility
	model (see \cite{barndorff2001non}) with discrete-time observations, which is
	demanding due to the L\'{e}vy-driven setting. Moving to empirical data, we
	apply the method to a latent Poisson Markov model to investigate bankruptcy
	and financial systemic stress. The application demonstrates that our proposed
	method can be used to estimate an economically meaningful time series of
	latent financial fragility. This also provides a feasible reduced-form latent
	modeling scheme for studying financial systemic risk.\begin{figure}[tbh]
		\centering
		\par
		\begin{minipage}[b]{0.5\textwidth}
			\centering
			\includegraphics[width=\textwidth]{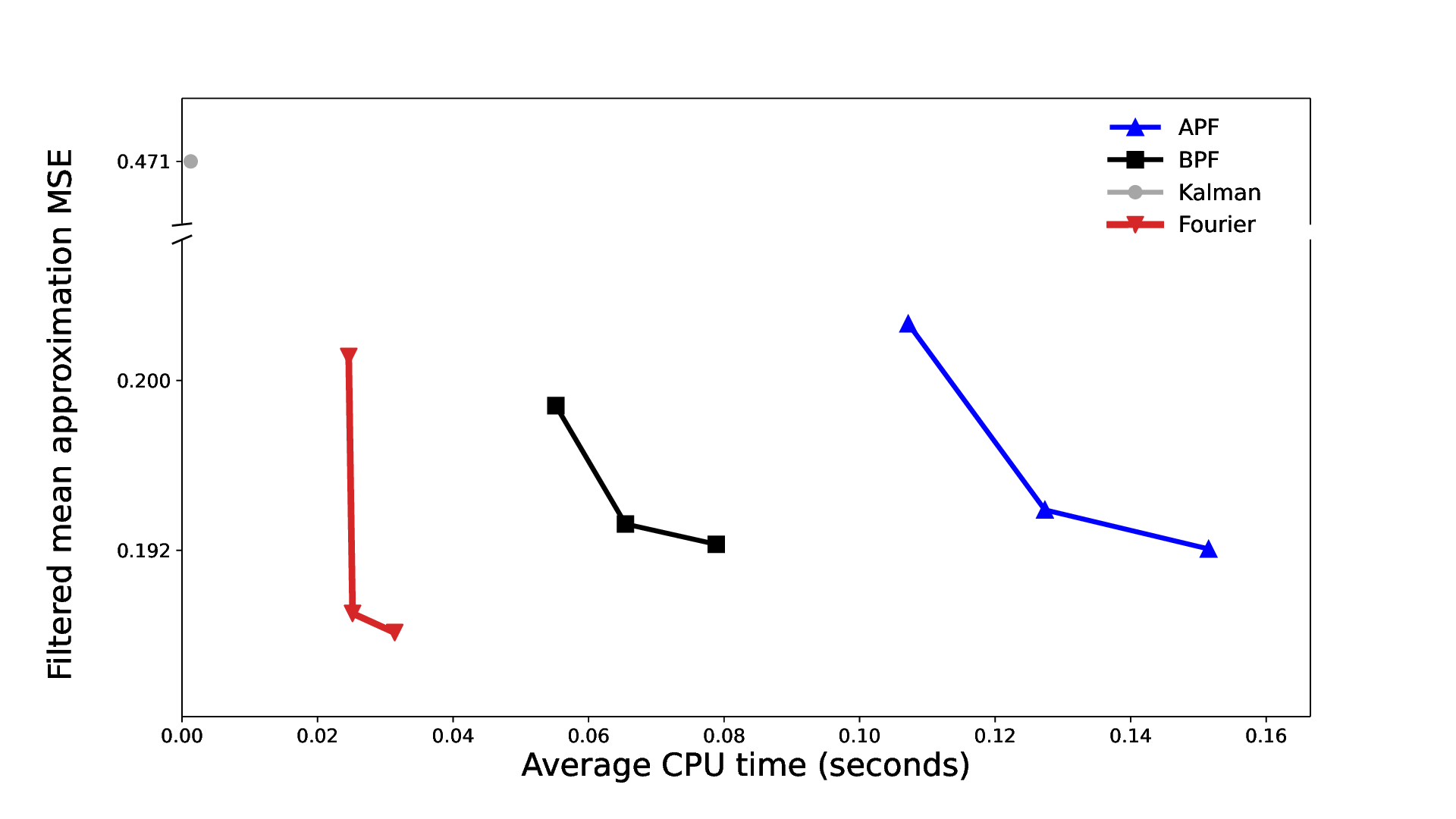}
			\subcaption*{(a) Discrete-time stochastic volatility LMM}
		\end{minipage}
		\hspace{-0.015\textwidth} \begin{minipage}[b]{0.5\textwidth}
			\centering
			\includegraphics[width=\textwidth]{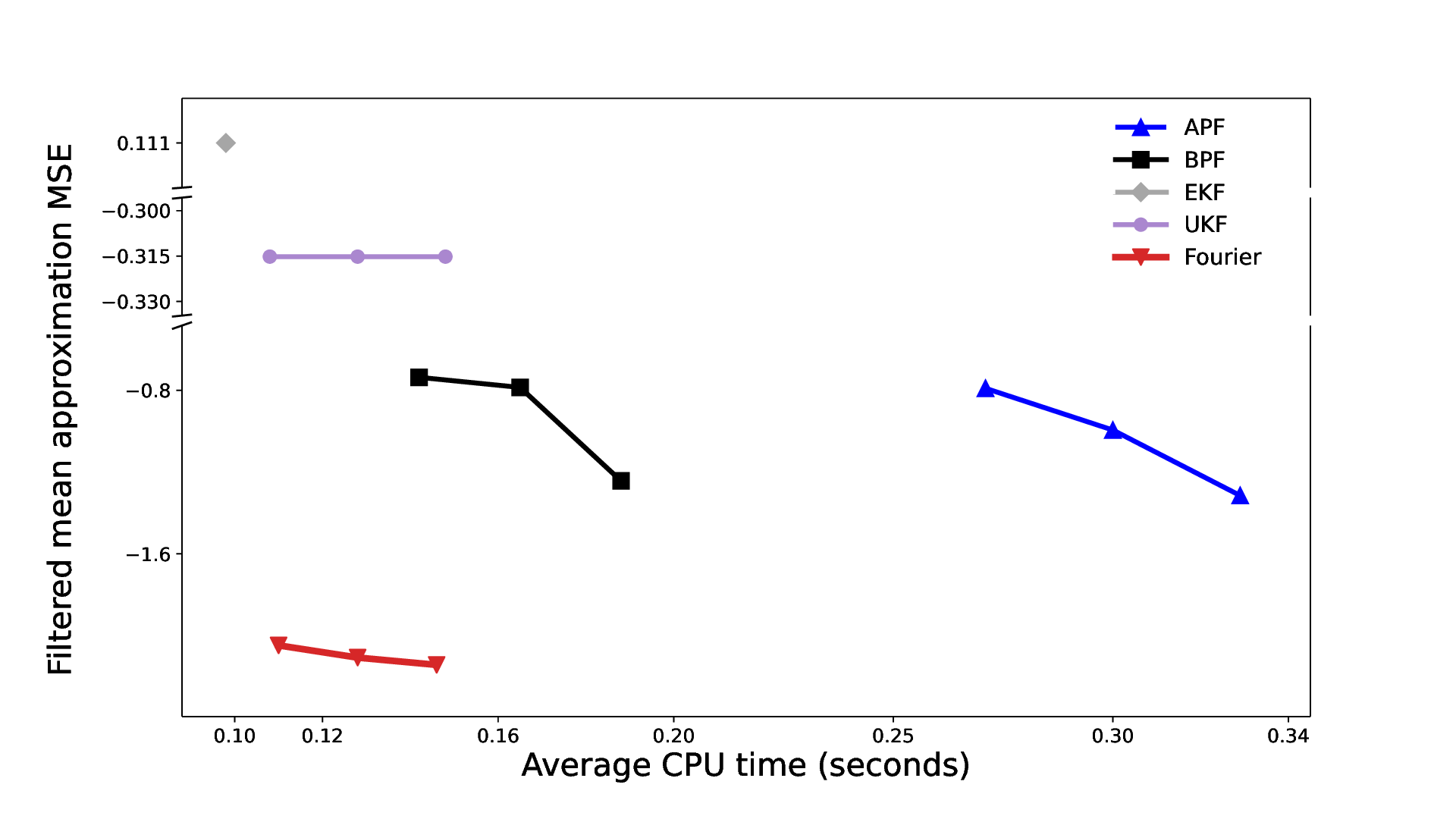}
			\subcaption*{(b) Quadratic term structure LMM}
		\end{minipage}
		\caption{{\protect\small Panels (a) and (b) compare the filtering methods for
				the discrete-time stochastic-volatility LMM and the quadratic term-structure
				LMM, respectively. Let $X_{k}$, $Y_{k}$ denote the latent state and the
				observation at time $k$, respectively. $\mu_{k}=\mathbb{E}[X_{k}%
				|\mathbf{Y}_{k};\theta]$ denotes the exact filtered mean and
				$\protect\widehat{\mu}_{k}$ denotes its approximation. The filtered-mean
				approximation error is measured by $\operatorname{MSE}_{\mu}=n^{-1}\sum
				_{k=1}^{n}\mathbb{E}[|\protect\widehat{\mu}_{k}-\mu_{k}|^{2}]$. By the
				orthogonality property of conditional expectations, $\mathbb{E}[|X_{k}%
				-\protect\widehat{\mu}_{k}|^{2}]=\mathbb{E}[|X_{k}-\mu_{k}|^{2}]+\mathbb{E}%
				[|\protect\widehat{\mu}_{k}-\mu_{k}|^{2}]$. The $\operatorname{MSE}_{\mu}$ of
				different methods can be equivalently ranked using directly computable
				$\operatorname{MSE}_{X}=n^{-1}\sum_{k=1}^{n}\mathbb{E}[|X_{k}%
				-\protect\widehat{\mu}_{k}|^{2}]$. The vertical axis reports $\log
				(\operatorname{MSE}_{X})$, and the horizontal axis reports average CPU time.
				Results are based on 200 independent sample paths of length $n=500$.}}%
		\label{fig:intro}%
	\end{figure}
	
	Relative to conventional approaches, the proposed method can be positioned
	along three dimensions. First, it broadens the scope of filtering and
	likelihood-based inference beyond SSMs and HMMs. This extension is not only
	computational but also theoretical: our analysis applies to a more general
	latent Markov environment. Second, the Fourier recursion provides a more
	accurate and stable framework for filtering and likelihood evaluation. Figure
	\ref{fig:intro} compares the accuracy of the Fourier recursion with that of
	the other methods.\footnote{In the discrete-time stochastic volatility model,
		log volatility follows a persistent latent process and log-squared returns
		provide noisy observations of the latent state (see, e.g.,
		\cite{KimShephardChib1998}). In the quadratic term structure model, bond
		yields depend quadratically on latent factors that follow persistent dynamics
		(see, e.g., \cite{ahn2002qtsm,DaiSingleton2003}).} Third, the method provides
	a unified framework for parameter estimation and associated statistical
	inference. The score and Hessian are generated recursively within the same
	Fourier system, supporting gradient-based MMLE and online, consistent
	estimation of MMLE's asymptotic covariance matrix. By contrast, alternative
	methods face additional challenges in delivering these two quantities.
	
	The remainder of the paper is organized as follows. Section 2 introduces the
	LMMs, defines the filtering and likelihood inference problems. Section 3
	develops the Fourier recursion, studies its special cases, and presents the
	main implementation details, including the positivity-preserving truncated
	density. Section 4 establishes the approximation error bounds and the
	associated asymptotic properties of the truncated Fourier recursion. Section 5
	reports simulation evidence. Section 6 applies the method to a latent Poisson
	Markov model of bankruptcy and financial stress.
	
	\section{The latent Markov models, estimation and inference problems}
	
	\subsection{The latent Markov models\label{sec:lmms}}
	
	We start with the definition of latent Markov models. Consider a latent Markov
	model defined on a probability space $(\Omega,\mathcal{F},\mathbb{P})$ with
	state vector $(X_{k},Y_{k},U_{k})$, in which $Y_{k}$ is observed with its
	state space denoted by $\mathcal{Y}$, while $X_{k}$ is usually considered to
	be a latent process with its state space denoted by $\mathcal{X}$, and $U_{k}$
	denotes a vector of observed explanatory variables. This formulation provides
	a flexible representation encompassing a broad class of economic and
	econometric models. A model originally formulated without latent variables,
	involving only $(Y_{k},U_{k})$, can be naturally enriched by introducing a
	latent state $X_{k}$, thereby enhancing the model's ability to capture
	unobserved heterogeneity. A model that already contains latent variables can
	be further enriched by augmenting it with observed explanatory variables
	$U_{k}$. Within this broad framework, several commonly used models can be
	written in the following form
	\begin{equation}
		Y_{k}=g(X_{k},Y_{k-1},U_{k};\theta)+\varepsilon_{k},\quad X_{k}=f(X_{k-1}%
		,Y_{k-1},U_{k};\theta)+\eta_{k}, \label{eq:general smm 2}%
	\end{equation}
	where $f$ and $g$ are measurable functions, $\eta_{k}$, $\varepsilon_{k}$
	are state innovations and observation shocks, and are assumed to be independent of $\{(X_{k-j},Y_{k-j},U_{k-j})\}^{k-1}_{j=0}$, with $\theta$ collecting the
	structural or statistical parameters.
    In addition to LMMs formulated directly
	in terms of $(X_{k},Y_{k},U_{k})$, latent models with finite-order lagged
	dependence can often be represented in this Markov form by augmenting the
	state vector with the relevant lagged variables.
	
	The covariates $U_{k}$ may include deterministic inputs, exogenous regressors
	that affect either the latent-state transition or the observation mechanism.
	Since the role of $U_{k}$ is conceptually distinct from that of the latent and
	observed state variables, its inclusion does not alter the structure of the
	filtering and likelihood updating equations. We explain this point formally in
	\ref{sec:covariates}. For expositional clarity, we therefore suppress $U_{k}$
	in the notation momentarily and proceed under the baseline assumption that
	$(X_{k},Y_{k})$ forms an LMM. To illustrate the breadth of the
	framework in (\ref{eq:general smm 2}), we next consider several important
	classes of LMMs, beginning with the case where $f$ and $g$ are linear functions.
	
	\begin{example}
		\label{example:linear smm}The linear state-space models:
		\begin{equation}
			Y_{k} =\Lambda X_{k}+\varepsilon_{k},\quad X_{k}=AX_{k-1}+\eta_{k},
			\label{eq:linear smm 2}%
		\end{equation}
		
	\end{example}
	
	Linear state-space models have been widely used in economics. For example,
	after solving for the equilibrium, linearized dynamic stochastic general
	equilibrium (DSGE) models, including the New Keynesian models, are often
	represented in state-space form for filtering, and likelihood-based
	estimation.\footnote{Typically, the rational-expectations solution takes a
		recursive form: $X_{k}=P\left(  \theta\right)  X_{k-1}+Q\left(  \theta\right)
		\eta_{k},$ where $X_{k}$ collects the relevant equilibrium states, and the
		transition matrices $P\left(  \theta\right)  $ and $Q\left(  \theta\right)  $
		are induced by the equilibrium solution and are therefore functions of the
		structural parameters $\theta.$ The observable macroeconomic variables $Y_{k}$
		are then linked to these latent equilibrium states through a measurement
		equation similar to (\ref{eq:linear smm 2}).}. State-space formulations also provide a natural framework for robust estimation, control and partial information (see, e.g., \cite{HansenSargent2007,HansenMayerSargent2010}). They are also
	commonly used to formulate predictive regressions in finance. In this case,
	$Y_{k}$ denotes observable financial variables, such as asset returns and
	predictors, while $X_{k}$ represents latent predictive states, such as
	conditional expected returns, expected dividend growth, or time-varying
	regression coefficients (see, e.g.,
	\cite{DanglHalling2012}). Moreover, linear state-space models are often used as dynamic factor models.
	In this setting, $X_{k}$ is interpreted as a low-dimensional vector of latent
	common factors, and $Y_{k}$ is a high-dimensional vector of observed economic
	variables. The factor structure provides a parsimonious way to reduce the
	dimensionality of large economic and financial datasets while preserving their
	main dynamic information (see, e.g., \cite{chen2022factor}).
	
	However, the linearity in Example \ref{example:linear smm} does not imply that
	$\varepsilon_{k}$ or $\eta_{k}$ are simple Gaussian distributions. For
	instance, discrete-time stochastic volatility models (see, e.g.,
	\cite{KimShephardChib1998}) also fit naturally
	into (\ref{eq:linear smm 2}) after a simple log-square transformation. Namely,
	suppose the return $R_{k}$ satisfies%
	\begin{equation}
		R_{k}=\sigma_{k}U_{k},\text{ }\log\left(  \sigma_{k}^{2}\right)  =\alpha
		+\beta\log\left(  \sigma_{k-1}^{2}\right)  +\sigma V_{k},
		\label{eq:discrete sv}%
	\end{equation}
	where $U_{k}$ and $V_{k}$ are two independent standard normal shocks. Defining
	$Y_{k}=\log\left(  R_{k}^{2}\right)  ,$ $X_{k}=\log\left(  \sigma_{k}%
	^{2}\right)  ,$ $\left(  X_{k},Y_{k}\right)  $ immediately turns into Example
	\ref{example:linear smm} where $\varepsilon_{k}=\log U_{k}^{2}$ follows a
	log-$\chi_{1}^{2}$ distribution. Another example is the factor ARCH models
	(see, e.g., \cite{bollerslev1994arch}):
	\begin{equation}
		\label{eq:factor arch}Y_{k}=\Lambda X_{k}+\varepsilon_{k}, \quad X_{k}%
		=\eta_{k}, \quad\eta_{k} \sim\mathcal{N}(0,\operatorname{diag}(\omega
		+AX^{2}_{k-1})) ,
	\end{equation}
	where $\varepsilon_{k}\sim\mathcal{N}(0,\sigma_{\varepsilon}I)$. Indeed, the
	relaxation of Gaussian assumptions is remarkably ubiquitous across broader
	latent variable frameworks; for instance, \cite{HallYao2003} study inference
	with heavy-tailed errors, while \cite{FernandezVillaverdeRubioRamirez2007}
	emphasize likelihood-based estimation for non-normal state-space
	representations of macroeconomic models. Crucially, this widespread
	non-Gaussianity severely complicates inference; even within linear frameworks
	like Example \ref{example:linear smm}, approaches developed for Gaussian
	situations may become inadequate and demand non-trivial computational effort.
	
	Besides the linear specifications of $f$, $g$ in (\ref{eq:general smm 2}),
	nonlinear forms of these functions are also common in economics and naturally
	lead to nonlinear latent Markov models. In macroeconometrics, DSGE models
	solved without linearization generally give rise to nonlinear policy functions
	and nonlinear transition laws, and hence to nonlinear state-space
	representations (see, e.g.,
	\cite{FernandezVillaverdeRubioRamirez2007,GustHerbstLopezSalidoSmith2017}). A
	representative example is given as follows.
	
	\begin{example}
		\label{example:rbc}The nonlinear state-space representation of the real
		business cycle (RBC) model:
		\begin{equation}
			Y_{k}=g(K_{k},Z_{k};\theta)+\xi_{k},\quad(K_{k+1},\log Z_{k+1})=(h_{K}%
			(K_{k},Z_{k};\theta),\rho_{z}\log Z_{k}+\sigma_{z}\eta_{k+1}).
			\label{eq:rbc_state_capital}%
		\end{equation}
		Here $X_{k}=(K_{k},\log Z_{k})^{\top}$ is the latent state, where $K_{k}$
		denotes capital and $Z_{k}$ denotes productivity. The observation vector
		$Y_{k}$ may contain output, consumption, investment, and labor. The
		measurement function $g$ is determined by production technology, while $h_{K}$
		describes capital accumulation implied by optimal decisions.
	\end{example}
	
	The same reasoning by which a DSGE model can be cast as a nonlinear LMM also
	applies to heterogeneous-agent environments, where the latent state must also
	summarize the evolution of cross-sectional distributions. This perspective is
	important in recent heterogeneous-agent New Keynesian models (HANK, see, e.g.,
	\cite{KaplanMollViolante2018}).\footnote{In
		such models, the latent state can be written as $X_{k}=(Z_{k},\mu_{k})$, where
		$Z_{k}$ collects aggregate state variables and $\mu_{k}$ denotes the
		cross-sectional distribution of households. A compact representation is
		$Y_{k}=h(Z_{k},\mu_{k})+\varepsilon_{k}$ and $X_{k}=F(X_{k-1},\eta_{k})$,
		where the transition of $\mu_{k}$ is induced by households' policy rules and
		idiosyncratic shocks. In practice, $\mu_{k}$ is usually
		approximated by a finite-dimensional object (see, e.g.,
		\cite{Reiter2009}), yielding a finite-dimensional nonlinear state-space model within our LMM framework.} In microeconometrics and
	industrial organization, nonlinear latent Markov structures are also very
	common. \cite{OlleyPakes1996} provide a classic example in which latent
	productivity follows a Markov process; \cite{DoraszelskiJaumandreu2013} extend
	this framework by allowing firms' investment to influence the evolution of
	latent productivity. In financial time-series modelling and asset pricing,
	nonlinear latent Markov structures are even more common. As illustrated in
	Figure \ref{fig:intro}, the quadratic term structure models (QTSMs; see, e.g.,
	\cite{ahn2002qtsm,DaiSingleton2003})
	\begin{equation}
		X_{k+1}=\mu+\Phi X_{k}+\Sigma\eta_{k+1},\quad Y_{k}=A+BX_{k}+X_{k}^{\intercal
		}CX_{k}+\xi_{k} \label{eq:qtsm}%
	\end{equation}
	are often used to model the yield curve and to infer latent economic factors
	driving bond yields. Yield forecasts in turn support fixed-income pricing and
	monetary-policy analysis.
	
	When the functions $f$ and $g$ in (\ref{eq:general smm 2}) are specified by
	deep neural networks, the resulting model corresponds to a recurrent neural
	network (RNN) or a stochastic recurrent neural network (SRNN).
	Methodologically, rather than imposing stringent structural restrictions, RNNs
	and SRNNs employ deep neural networks to flexibly characterize the complex
	dynamic relationships between $X_{k}$ and $Y_{k}.$ Specifically, these neural
	networks are represented through hierarchical layers of nonlinear
	transformations parameterized by a high-dimensional vector $\theta,$ for
	example%
	\[
	f_{\theta}\left(  z\right)  =W_{L}\sigma\left(  W_{L-1}\sigma\left(
	\cdots\sigma\left(  W_{1}z+b_{1}\right)  \cdots\right)  +b_{L-1}\right)
	+b_{L},
	\]
	where $\theta=\left\{  \left(  W_{l},b_{l}\right)  \right\}  _{l=1}^{L}$
	denotes the collection of network weights and biases, and $\sigma\left(
	\cdot\right)  $ is a nonlinear activation function. This perspective naturally
	leads to the following LMM.\footnote{Although neural networks are often viewed
		as nonparametric or highly flexible function approximators, a network with a
		fixed architecture is a finite-dimensional parametric map indexed by its
		weights and biases. Thus, once the network architecture is fixed, the model
		admits a finite-dimensional parameterization and can be treated within the
		parametric framework considered here.}
	
	\begin{example}
		\label{example:SRNN}The stochastic recurrent neural network:
		\begin{equation}
			Y_{k}=g_{\theta}(X_{k})+\varepsilon_{k},\quad X_{k}=f_{\theta}\left(
			X_{k-1}\right)  +\eta_{k}, \label{eq:SRNN}%
		\end{equation}
		where $f_{\theta},g_{\theta}$ are two known deep neural networks. When the
		noise terms $\left(  \varepsilon_{k},\eta_{k}\right)  $ are removed, an SRNN
		reduces to a standard RNN. In this sense, a conventional RNN can be regarded
		as a degenerate LMM without stochasticity. SRNNs, RNNs, and their variants
		have been widely used in modern deep learning.\footnote{Prominent examples
			include long short-term memory networks (LSTMs), variational recurrent neural networks
			(VRNNs), and
			autoregressive transformers with linear attention (see, e.g.,
			\cite{KatharopoulosVyasPappasFleuret2020} and references therein).}
	\end{example}
	
	In addition to the latent Markov models summarized by (\ref{eq:general smm 2}%
	), several other latent Markov structures are also widely used in economics
	and finance. A prominent example is the class of regime-switching models,
	where the latent state is a discrete regime capturing unobserved economic or
	financial conditions. A commonly used specification of such models is given by:
	
	\begin{example}
		\label{example:regime}The Markov regime-switching models:%
		\begin{equation}
			Y_{k}=\mu\left(  X_{k}\right)  +\phi Y_{k-1}+\Sigma\left(  X_{k}\right)
			^{1/2}V_{k},\quad\mathbb{P}\left(  X_{k}=j|X_{k-1}=i,Y_{k-1}=y\right)
			=\psi_{j}\left(  i,y\right)  ,
		\end{equation}
		where $V_{k}$ is a standard normal noise vector, $X_{k}\in\left\{
		1,2,...,M\right\}  $ denotes the discrete hidden regime, and $\psi_{j}$
		represents a particular link function. After being first proposed by
		\cite{hamilton1989new} to analyze nonstationary time series and business cycle
		fluctuations, the Markov regime-switching framework has evolved to incorporate
		more complex dynamics.
		More recently, \cite{Pouzo2022} established rigorous asymptotic properties and
		maximum likelihood estimation theories for this generalized class of models.
	\end{example}
	
	Another important class of latent Markov structures is provided by state-space
	dynamic generalized linear models (GLMs), or more broadly by LMMs with
	conditionally exponential-family observations. This formulation is
	particularly useful when the observed data are discrete, bounded (see, e.g., Chapter 5-6 of
	\cite{TsayChen2019NonlinearTimeSeries}). For example, when the observation is binary, one
	often uses the following dynamic system.
	
	\begin{example}
		The dynamic binary-response model:%
		\begin{equation}
			Y_{k}|X_{k}\sim\mathrm{Bernoulli}(p_{k}),\quad g(p_{k})=\alpha+X_{k},\quad
			X_{k}=AX_{k-1}+\eta_{k}, \label{eq:logit}%
		\end{equation}
		where $g$ is a link function. When $g(p)=\mathrm{logit}(p)=\log(p/(1-p))$,
		\eqref{eq:logit} becomes a dynamic logit model, whereas when
		$g(p)=\mathrm{probit}(p),$ it becomes a dynamic probit model, where
		$\mathrm{probit}(p)$ denotes the inverse of the standard normal cumulative
		distribution function. Moreover, if the observation $Y_{k}$ is count-valued, a
		Poisson observation equation is commonly adopted. The same perspective also
		provides a natural route to dynamic network data. Relative to the binary
		system in \eqref{eq:logit}, a dynamic network model replaces the single latent
		state by a collection of node-level latent states and replaces the single
		Bernoulli observation by a set of dyadic Bernoulli observations. A prominent
		example is the dynamic latent space model (see, e.g., \cite{SewellChen2015}).
	\end{example}
	
	A further important source of latent Markov structures is provided by
	continuous-time models observed at discrete times.
	
	\begin{example}
		\label{example:heston}The Duffie-Pan-Singleton (DPS, see
		\cite{duffie2000transform}) model:
		\begin{subequations}
			\begin{align*}
				d\left(  \log S_{t}\right)   &  =(\mu-V_{t}/2)dt+\sqrt{V_{t}}dW_{1t}%
				+J_{1t}dN_{1t},\\
				dV_{t}  &  =\kappa(\alpha-V_{t})dt+\xi\sqrt{V_{t}}\left(  \rho dW_{1t}%
				+\sqrt{1-\rho^{2}}dW_{2t}\right)  +J_{2t}dN_{2t},
			\end{align*}
			where $W_{1t}$ and $W_{2t}$ are two independent standard Brownian motions,
			$N_{1t}$ and $N_{2t}$ are Poisson processes governing jumps in returns and
			volatility, respectively, and $J_{1t}$ and $J_{2t}$ denote the corresponding
			jump sizes.
		\end{subequations}
	\end{example}
	
	Besides the classic Brownian-driven models, the L\'{e}vy-driven models provide
	a versatile continuous-time framework for accommodating stochastic jumps,
	heavy-tailed distributions, and discontinuous asset dynamics. A canonical
	model in this class is specified as follows.
	
	\begin{example}
		\label{example:bns}The Barndorff-Nielsen and Shephard (BNS, see
		\cite{barndorff2001non}) model:
		\begin{subequations}
			\begin{align}
				d\left(  \log S_{t}\right)   &  =(\mu+\beta V_{t})dt+\sqrt{V_{t}}dW_{t}+\rho
				dZ_{\lambda t},\\
				dV_{t}  &  =-\lambda V_{t}dt+dZ_{\lambda t},
			\end{align}
			where $Z_{t}$ is a subordinator.
		\end{subequations}
	\end{example}
	
	For these two continuous-time models, let the data be sampled at discrete
	times $t_{k}=k\Delta$, the observed variable $Y_{k}$ corresponds to the
	discrete log-return $\log S_{k\Delta}$ or $\log\left(  S_{k\Delta
	}/S_{(k-1)\Delta}\right)  $ while the latent state $X_{k}$ is specified as the
	stochastic volatility state $V_{k\Delta}.$ Consequently, these continuous-time
	models with discrete observations can be similarly interpreted as
	discrete-time LMMs.
	
	Some point-process models also constitute an important class of LMMs. After
	augmenting the state vector with the relevant intensity or memory variables,
	many point-process models can be represented within a latent Markov framework.
	An example is the exponential-kernel Hawkes process, together with its
	multivariate and marked extensions, which is frequently used for modelling
	jumps (see, e.g., \cite{AitSahaliaCachoDiazLaeven2015}).
	
	\begin{example}
		\label{example:hawkes} Hawkes process with exponential decay (see, e.g.,
		\cite{Hawkes_1971}): let $N_{t}$ be a counting process adapted to the
		filtration $\{\mathcal{F}_{t}\}_{t\geq0}$. Its conditional intensity
		$\lambda_{t}$ satisfies
		\begin{equation}
			d\lambda_{t}=-\beta(\lambda_{t}-\mu)dt+\alpha dN_{t}, \label{eq:hawkes_sde}%
		\end{equation}
		where $\mu,\alpha,\beta$ are three positive parameters. For the Hawkes process
		above, the augmented process $(N_{t},\lambda_{t})$ is Markov and therefore
		falls within the latent Markov framework considered in this paper. Other
		standard examples include Markov-modulated Poisson processes, Cox processes,
		and Poisson-type models with latent or time-varying intensities.
	\end{example}
	
	More generally, many models that are not Markovian in their original
	formulation can nevertheless be embedded into the latent Markov framework
	through a suitable Markov augmentation and finite-dimensional approximations.
	This augmentation technique is fundamental in traditional time series
	analysis.\footnote{For instance, an ARMA(p,q) process $Y_{k}:=\sum_{i=1}%
		^{p}\phi_{i}Y_{k-i}+\sum_{j=0}^{q}\psi_{j}\eta_{k-j}$\ can be represented in a
		linear state-space model by collecting lagged observations and innovations
		into an augmented state $X_{k}$ (see, e.g., \cite{DurbinKoopman2012}).}
	\cite{Sargent2026} show that reduced-rank vector autoregression estimated by
	dynamic mode decomposition can be linked to linear state-space models. Another
	example is the rough Volterra volatility models (see, e.g.,
	\cite{LiPhillipsShiYu2025WeakIdentification}). These models can be represented
	as
	\begin{subequations}
		\begin{align}
			d\log S_{t}  &  =\mu_{1}\left(  V_{t}\right)  dt+\sigma_{1}\left(
			V_{t}\right)  dW_{1t},\label{eq:general_rough_volatility 1}\\
			V_{t}  &  =V_{0}+\int_{0}^{t}K(t-s)\mu_{2}\left(  V_{s}\right)  ds+\int%
			_{0}^{t}K(t-s)\sigma_{2}\left(  V_{s}\right)  dW_{2s},
			\label{eq:general_rough_volatility 2}%
		\end{align}
		where $\mu_{1},\mu_{2},\sigma_{1},\sigma_{2}$ are the drift and volatility
		coefficient functions, and the kernel $K$ determines how past shocks and past
		volatility affect current volatility. $(\log S_{t},V_{t})$ is generally not a
		finite-dimensional Markov process. Nevertheless, this class of models often
		admits a Markovian lift.\footnote{If the kernel admits the representation
			$K(t)=\int_{0}^{\infty}e^{-\beta t}\zeta(d\beta)$ for some nonnegative measure
			$\zeta$, the dependence on the past can be encoded by a continuum of auxiliary
			memory factors $X_{t}(\beta)=\int_{0}^{t}e^{-\beta(t-s)}\mu_{2}(V_{s}%
			)ds+\int_{0}^{t}e^{-\beta(t-s)}\sigma_{2}(V_{s})dW_{2s}$, so that
			$V_{t}=V_{0}+\int_{0}^{\infty}X_{t}(\beta)\zeta(d\beta)$. Thus, although
			$V_{t}$ itself is not finite-dimensional Markov, the lifted process $(\log
			S_{t},\{X_{t}(\beta):\beta>0\})$ is Markov.} Moreover, this Markov lift can be
		approximated by a finite-dimensional Markovian system, which can be treated as
		a standard LMM (see \cite{LiPhillipsShiYu2025WeakIdentification} and references therein).
		
		\subsection{The likelihood update, filtering, and prediction\label{sec:bayes updating}}
		 To handle different types of variables within a unified probabilistic
		 framework, we assume that $\mathcal{X}$ and $\mathcal{Y}$ are endowed with
		 dominating measures $\nu_{X}$ and $\nu_{Y}$, respectively. Specifically, this
		 setup accommodates both continuous and discrete variables: in the continuous
		 case, the dominating measure is the Lebesgue measure, whereas in the discrete
		 case, it is the counting measure. Here, we consider the LMM parametrized by a
		 vector of parameters $\theta$ that belongs to a compact set
		 $\Theta\subseteq\mathbb{R}^{d}$. Without loss of generality, we assume that
		 $(X_{k},Y_{k})$ is a time-homogeneous Markov process. For an LMM, the most
		important characteristic is its transition dynamics. Usually, the transition
		dynamics can be characterized by the transition probability function
		\[
		p_{(X,Y)}(x,y|x_{0},y_{0};\theta):=\mathbb{E}\left[  \delta\left(
		X_{1}-x\right)  \delta\left(  Y_{1}-y\right)  |X_{0}=x_{0},Y_{0}=y_{0}%
		;\theta\right]  ,
		\]
		where $\delta\left(  \cdot\right)  $ represents the Dirac delta
		function.\footnote{Here $\delta(\cdot)$ is understood in a generalized measure-theoretic sense, providing a unified representation across different state spaces and facilitating convenient computation. Its validity is rigorously grounded in the theory of Dirac measures on measurable spaces. When the latent or observed variables are continuous,
			$\delta(\cdot)$ is understood as the Dirac measure concentrated at a given
			point; when the variables take discrete values, it plays the role of an
			indicator function (see, e.g.,
			Chapter 1 of \cite{kallenberg2021foundations}).}
			
		In the analysis of latent Markov models, the main objectives are to infer the
		unknown parameter vector $\theta\in\Theta$ from the observed sequence
		$\mathbf{y}_{k}=(y_{k},y_{k-1},...,y_{0})$ and to predict the latent states.
		Parameter inference is typically carried out by maximum likelihood estimation,
		which requires the evaluation of the marginal likelihood of the data. The
		marginal likelihood function can be constructed sequentially as
	\end{subequations}
	\begin{equation}
		\mathcal{L}(\theta)=\prod_{k=1}^{n}\mathcal{L}_{k}(\theta),\quad
		\label{eq: likelihood iteratively expression}%
	\end{equation}
	where $\mathcal{L}_{k}(\theta)=\mathbb{E}\left[  \delta\left(  Y_{k}%
	-y_{k}\right)  |\mathbf{Y}_{k-1}=\mathbf{y}_{k-1};\theta\right]  .$ The
	associated marginal maximum likelihood estimator (MMLE, hereafter) of $\theta
	$, obtained by maximizing the marginal likelihood function $\mathcal{L}%
	(\theta)$, or equivalently the marginal $\log$-likelihood function
	$\ell(\theta)=\log\mathcal{L}(\theta)$, is defined as follows:
	\[
	\hat{\theta}^{(n)}=\underset{\theta\in\Theta
	}{\operatorname{argmax}}\mathcal{L}(\theta)=\underset{\theta\in\Theta
	}{\operatorname{argmax}}\ell(\theta).
	\]

	Besides likelihood evaluation, it is also crucial to track the dynamics of the
	unobservable latent states under an estimated or pre-specified value of
	$\theta$. This leads to the fundamental tasks of filtering and prediction.
	Conditional on a given $\theta$, the recursively computed filtered density
	\begin{equation}
		p(x_{k}|\mathbf{y}_{k};\theta):=\mathbb{E}\left[  \delta\left(  X_{k}%
		-x_{k}\right)  |\mathbf{Y}_{k}=\mathbf{y}_{k};\theta\right]  ,\quad
		k=1,2,\cdots,n. \label{eq: filtered density}%
	\end{equation}
	characterizes the latent state given the available observation history up to
	time $k.$ Closely related to filtering is prediction, which aims to forecast
	the future latent state of the LMM. The one-step predictive density for the
	latent state characterizes the distribution of the next unobservable state
	given the current observation history and is formulated as follows:%
	\begin{equation}
		p(x_{k+1}|\mathbf{y}_{k};\theta):=\mathbb{E}\left[  \delta\left(
		X_{k+1}-x_{k+1}\right)  |\mathbf{Y}_{k}=\mathbf{y}_{k};\theta\right]  ,\quad
		k=1,2,\cdots,n.
	\end{equation}
	Throughout this paper, we use $p(\cdot|\cdot;\theta)$ to denote the
	conditional density with respect to the dominating measure $\nu_{X}$ on
	$\mathcal{X}$. This unified convention allows us to use the term
	\textquotedblleft density\textquotedblright\ for both continuous-state and
	discrete-state cases. In particular, when $X_{k}$ takes values in a finite or
	countable state space, $\nu_{X}$ is taken as the counting measure, in which
	case $p(\cdot|\cdot;\theta)$ corresponds to a probability mass function. When
	$X_{k}$ is continuous-valued, $\nu_{X}$ is taken as the Lebesgue measure, and
	$p(\cdot|\cdot;\theta)$ is the usual probability density.
	
	A commonly used technique for obtaining the marginal likelihood and updating
	the filtered density is the Bayes updating system. Specifically, when a new observation $y_{k}$
	becomes available, the filtered density can be used to update the likelihood
	according to%
	\begin{align}
		\mathcal{L}_{k}(\theta)  &  =\int_{\mathcal{X}}p_{Y}(y_{k}|x_{k-1}%
		,\mathbf{y}_{k-1};\theta)p(x_{k-1}|\mathbf{y}_{k-1};\theta)\nu_{X}%
		(dx_{k-1})\nonumber\\
		&  =\int_{\mathcal{X}}p_{Y}(y_{k}|x_{k-1},y_{k-1};\theta)p(x_{k-1}%
		|\mathbf{y}_{k-1};\theta)\nu_{X}(dx_{k-1}), \label{eq:likelihood update}%
	\end{align}
	where the second equality follows from the Markov property of $(X_{k},Y_{k})$.
	The marginal transition density $p_{Y}(y|x_{0},y_{0};\theta)$ is defined by
	integrating out the forward latent variable $x$ in the transition density.
	Next, by the definition of the conditional density, we have%
	\begin{equation}
		p(x_{k-1}|\mathbf{y}_{k-1};\theta)=\frac{p(x_{k-1},y_{k-1}|\mathbf{y}%
			_{k-2};\theta)}{\mathcal{L}_{k-1}(\theta)}. \label{eq:conditional density}%
	\end{equation}
	To establish the connection between two successive filtered densities, namely
	$p(x_{k-1}|\mathbf{y}_{k-1};\theta)$ and $p(x_{k}|\mathbf{y}_{k};\theta)$, one
	further integrates out $x_{k-1}$ in (\ref{eq:conditional density}). This
	yields%
	\begin{align}
		p(x_{k}|\mathbf{y}_{k};\theta)  &  =\frac{1}{\mathcal{L}_{k}(\theta)}%
		\int_{\mathcal{X}}p(x_{k},y_{k}|x_{k-1},\mathbf{y}_{k-1};\theta)p(x_{k-1}%
		|\mathbf{y}_{k-1};\theta)\nu_{X}(dx_{k-1})\nonumber\\
		&  =\frac{1}{\mathcal{L}_{k}(\theta)}\int_{\mathcal{X}}p_{(X,Y)}%
		(x_{k}\mathbf{,}y_{k}|x_{k-1},y_{k-1};\theta)p(x_{k-1}|\mathbf{y}_{k-1}%
		;\theta)\nu_{X}\left(  dx_{k-1}\right)  , \label{eq:filter update}%
	\end{align}
	where the second equality follows directly from the Markov property of
	$(X_{k},Y_{k}).$ Using the filtered density $p(x_{k}|\mathbf{y}_{k};\theta),$
	the one-step predictive density $p(x_{k+1}|\mathbf{y}_{k};\theta)$ can be
	written as%
	\begin{equation}
		p(x_{k+1}|\mathbf{y}_{k};\theta)=\int_{\mathcal{X}}p_{X}(x_{k+1}|x_{k}%
		,y_{k};\theta)p(x_{k}|\mathbf{y}_{k};\theta)\nu_{X}\left(  dx_{k}\right).
		\label{eq:prediction update}%
	\end{equation}

	\section{A novel approach: Fourier expansion recursion}
	
	The limitations of existing filtering methods discussed in Section
	\ref{sec:related work} motivate the need for a filtering and likelihood
	evaluation framework that is both accurate and robust for general LMMs, while
	also admitting transparent error analysis and asymptotic guarantees. We develop
	such a framework in this section. For expositional clarity and technical
	convenience, we consider the case in which the latent state is one-dimensional
	and compactly supported on $[b_{0},b_{1}]$. This compact-state formulation need
	not constitute a substantive restriction on the underlying latent process.\footnote{In
		particular, for a noncompact state space, one may introduce a smooth bijective
		transformation $\tilde{X}_{k}=D\left(  X_{k}\right)  ,$ where $D\left(
		\cdot\right)  :\mathcal{X\rightarrow}[b_{0},b_{1}].$ For example, we can use
		$D\left(  x\right)  =b_{0}+(b_{1}-b_{0})(\arctan x/\pi+1/2).$ Since this is a
		one-to-one reparameterization, it preserves the likelihood of the observed
		process exactly. Estimation and inference can therefore be carried out
		equivalently based on the transformed LMM $(\tilde{X}_{k},Y_{k}).$ We
		therefore work with the compact-state representation throughout this paper.}
	Moreover, the results below can be extended to the multivariate case in a
	natural way by replacing the associated term with its multivariate counterpart.
	
	\subsection{The exact Fourier recursion and model-specific specializations\label{sec:exact fourier}}
	
	We begin by representing the filtered density by the Fourier expansion
	\begin{equation}
		p(x_{k}|\mathbf{y}_{k};\theta)=\sum_{l\in\mathbb{Z}}c_{k,l}(\theta)\exp\left(
		\frac{2\pi ilx_{k}}{b_{1}-b_{0}}\right)
		,\label{eq:fourier filtered density representation}%
	\end{equation}
	for $x_{k}\in\lbrack b_{0},b_{1}]$. When $X_{k}$ is continuously distributed
	on $[b_{0},b_{1}]$, a simple sufficient condition for the Fourier representation in (\ref{eq:fourier filtered density representation}) is that $p(x_{k}|\mathbf{y}_{k};\theta)$ is continuously differentiable w.r.t. $x_k$ (see, e.g., Chapter 1 of \cite{duoandikoetxea2024fourier}). The Fourier coefficients are given by
	\begin{equation}
		c_{k,l}(\theta)=\frac{1}{b_{1}-b_{0}}\int_{\mathcal{X}}\exp\left(  -\frac{2\pi
			ilx_{k}}{b_{1}-b_{0}}\right)  p(x_{k}|\mathbf{y}_{k};\theta)dx_{k},\quad
		l\in\mathbb{Z}.\label{eq:fourier coefficient continuous}%
	\end{equation}
	When $X_{k}$ is discrete-valued in the finite set $\{0,1,\ldots,M-1\}$,
	(\ref{eq:fourier filtered density representation}) is interpreted as the
	inverse discrete Fourier transform (IDFT, hereafter) with $b_{0}=0$ and
	$b_{1}=M$. In this case,
	\begin{equation}
		c_{k,l}(\theta)=M^{-1}\sum_{m=0}^{M-1}p_{k,m}(\theta)\exp\left(  \frac{-2\pi
			ilm}{M}\right)  ,\text{ }\label{eq:fourier coefficient discrete}%
	\end{equation}
	for $0\leq l\leq M-1$ and $c_{k,l}(\theta)=0$ otherwise, with $p_{k,m}%
	(\theta)$ as the probability of $X_{k}=m$ given $\mathbf{Y}_{k}=\mathbf{y}%
	_{k},$ i.e., $p_{k,m}(\theta)=\mathbb{P}(X_{k}=m|\mathbf{Y}_{k}=\mathbf{y}%
	_{k};\theta).$ 
	
	The Fourier expansion is widely used as a powerful tool in economic analysis
	and is closely related to sieve estimation in nonparametric and semiparametric
	econometrics (see, e.g., \cite{chen2012estimation,chen2014sieve}). The proposed Fourier expansion in
	(\ref{eq:fourier filtered density representation}) is closely related to sieve
	methods in that both use a set of basis functions to represent an
	infinite-dimensional object. While conventional sieve estimation typically
	approximates a fixed unknown function with a sieve space whose dimension grows
	with the sample size, our method applies the same principle to the filtered
	density, which evolves recursively with incoming observations.
	
	To state the main result, we first introduce several transform objects that
	summarize the transition dynamics of the LMMs. For $u,s\in\mathbb{R}$, we
	define
	\begin{align}
		\Psi(u,y_{1}|x_{0},y_{0};\theta)  &  =\mathbb{E}\left[  \delta(Y_{1}%
		-y_{1})\exp(iuX_{1})|X_{0}=x_{0},Y_{0}=y_{0};\theta\right]
		,\label{eq:def psi}\\
		\Psi_{X}(u|x_{0},y_{0};\theta)  &  =\mathbb{E}\left[  \exp(iuX_{1}%
		)|X_{0}=x_{0},Y_{0}=y_{0};\theta\right]  . \label{eq:def psi X}%
	\end{align}
	The function $\Psi(u,y_{1}|x_{0},y_{0};\theta)$ is the conditional Fourier
	transform of the joint transition law of $(X_{1},Y_{1})$ w.r.t. $X_{1}$,
	evaluated at $Y_{1}=y_{1}$. The function $\Psi_{X}(u|x_{0},y_{0};\theta)$ is
	the corresponding conditional Fourier transform of the latent transition
	alone. We shall call $\Psi$ the conditional Fourier transition kernel. We
	further define $H(u,y_{1},s,y_{0};\theta)$ (resp. $H_{X}(u,s,y_{0};\theta)$)
	as the Fourier transform of $\Psi(u,y_{1}|x_{0},y_{0};\theta)$ (resp.
	$\Psi_{X}(u|x_{0},y_{0};\theta)$) w.r.t. the previous latent state $x_{0},$
	i.e.,
	\begin{align}
		H(u,y_{1},s,y_{0};\theta)  &  =\int_{\mathcal{X}}\exp(isx_{0})\Psi
		(u,y_{1}|x_{0},y_{0};\theta)\nu_{X}(dx_{0}),\label{eq:def H}\\
		H_{X}(u,s,y_{0};\theta)  &  =\int_{\mathcal{X}}\exp(isx_{0})\Psi_{X}%
		(u|x_{0},y_{0};\theta)\nu_{X}(dx_{0}). \label{eq:def H X}%
	\end{align}

	The Fourier transformed functions $\Psi$, $\Psi_{X}$, $H$ and $H_{X}$ are well
	defined under broad conditions. In particular, $\Psi$ and $\Psi_{X}$ exist
	whenever the transition density $p_{(X,Y)}$ of the LMM is well defined, for instance, when it is continuously differentiable (see, e.g., Chapter 1 of \cite{duoandikoetxea2024fourier}). $H$
	and $H_{X}$ require the integrability w.r.t. the previous latent state $x_{0}%
	$. This condition is automatically satisfied when the latent state is
	compactly supported. With these notations, the likelihood update, filtering, and one-step prediction admit an exact representation entirely in terms of the Fourier coefficients of the filtered density. The following theorem establishes the resulting exact Fourier recursion.
	
	\begin{theorem}
		\label{thm:fourier series update 1} Suppose that the latent state $X_{k}$ is
		compactly supported on $\mathcal{X}=[b_{0},b_{1}]$. Then the likelihood update
		function $\mathcal{L}_{k}(\theta)$ is given by
		\begin{equation}
			\mathcal{L}_{k}(\theta)=\sum_{s\in\mathbb{Z}}c_{k-1,s}(\theta)H\left(
			0,y_{k},\frac{2\pi s}{b_{1}-b_{0}},y_{k-1};\theta\right)  ,
			\label{eq:fourier likelihood update}%
		\end{equation}
		where $\left\{  c_{k,l}(\theta)\right\}  _{l\in\mathbb{Z}}$ denotes the Fourier
		coefficients of the filtering density $p(x_{k}|\mathbf{y}_{k};\theta)$ defined
		in (\ref{eq:fourier filtered density representation}). These coefficients are
		updated according to%
		\begin{equation}
			c_{k,l}(\theta)=\sum_{s\in\mathbb{Z}}\frac{c_{k-1,s}(\theta)}{\mathcal{L}%
				_{k}(\theta)(b_{1}-b_{0})}H\left(  \frac{-2\pi l}{b_{1}-b_{0}},y_{k}%
			,\frac{2\pi s}{b_{1}-b_{0}},y_{k-1};\theta\right)  ,\quad l\in\mathbb{Z}.
			\label{eq:fourier filter update}%
		\end{equation}
		Moreover, the one-step predictive density admits the representation
		\begin{equation}
			p(x_{k}|\mathbf{y}_{k-1};\theta)=\sum_{l\in\mathbb{Z}}c_{k,l}^{-}(\theta
			)\exp\left(  \frac{2\pi ilx_{k}}{b_{1}-b_{0}}\right)  ,
			\label{eq:fourier predictive density}%
		\end{equation}
		where $\left\{  c_{k,l}^{-}(\theta)\right\}  _{l\in\mathbb{Z}}$ can be
		obtained from $\left\{  c_{k,l}(\theta)\right\}  _{l\in\mathbb{Z}}$ as
		\begin{equation}
			c_{k,l}^{-}(\theta)=\sum_{s\in\mathbb{Z}}\frac{c_{k-1,s}(\theta)}{b_{1}-b_{0}%
			}H_{X}\left(  \frac{-2\pi l}{b_{1}-b_{0}},\frac{2\pi s}{b_{1}-b_{0}}%
			,y_{k-1};\theta\right)  ,\quad l\in\mathbb{Z}.
			\label{eq:fourier prediction update}%
		\end{equation}
		The recursion is initialized by specifying an initial filtered density
		$p(x_{0}|y_{0};\theta)$ and the corresponding initial Fourier coefficients
		\[
		c_{0,l}(\theta)=\frac{1}{b_{1}-b_{0}}\int_{\mathcal{X}}\exp\left(  -\frac{2\pi
			ilx_{0}}{b_{1}-b_{0}}\right)  p(x_{0}|y_{0};\theta)\nu_{X}(dx_{0}),\text{
		}l\in\mathbb{Z}.
		\]
		
	\end{theorem}
	
	Theorem \ref{thm:fourier series update 1} establishes that likelihood
	evaluation, filtering, and prediction operate entirely within the Fourier
	coefficient space. Specifically, the likelihood is computed as a linear
	functional of $\{c_{k,l}(\theta)\}_{l\in\mathbb{Z}}$, while deterministic
	operators map $\{c_{k-1,l}(\theta)\}_{l\in\mathbb{Z}}$ to the filtering
	$\{c_{k,l}(\theta)\}_{l\in\mathbb{Z}}$ and prediction $\{c_{k,l}^{-}%
	(\theta)\}_{l\in\mathbb{Z}}$ coefficients, respectively. Thus, the infinite
	sequence $\{c_{k,l}(\theta)\}_{l\in\mathbb{Z}}$ serves as
	information-preserving elements for general LMMs, rigorously generalizing the
	role of the first two moments in the Kalman filter. Importantly, Theorem
	\ref{thm:fourier series update 1} is an exact reformulation of the Bayesian
	filtering system rather than an ad hoc numerical discretization. By preserving
	the continuous filtering distribution through its full sequence of Fourier
	coefficients, intractable Bayes integrals are converted into structured
	recursions, with model-specific nonlinear and non-Gaussian features absorbed
	into $H$ and $H_{X}$.
	
	As an immediate consequence of Theorem \ref{thm:fourier series update 1}, the
	task of predicting latent state can also be recovered from the corresponding
	Fourier coefficients. Usually, prediction of the latent state is based on
	either the filtered mean or the one-step predictive mean, depending on the
	information set available to the econometrician. Both quantities can be
	obtained directly from the Fourier representation of the corresponding
	conditional distribution. In general, by
	(\ref{eq:fourier filtered density representation}), the filtered mean can be
	obtained as
	\begin{equation}
		\mathbb{E}\left[  X_{k}|\mathbf{Y}_{k}=\mathbf{y}_{k};\theta\right]
		=\sum_{l\in\mathbb{Z}}c_{k,l}(\theta)\int_{\mathcal{X}}x_{k}\exp\left(
		\frac{2\pi ilx_{k}}{b_{1}-b_{0}}\right)  \nu_{X}(dx_{k}).
		\label{eq:filtered mean}%
	\end{equation}
	When $X_{k}$ is continuous-valued on $[b_{0},b_{1}]$ and $\nu_{X}$ now is the
	Lebesgue measure, (\ref{eq:filtered mean}) admits closed-form expression as%
	\[
	\mathbb{E}\left[  X_{k}|\mathbf{Y}_{k}=\mathbf{y}_{k};\theta\right]
	=c_{k,0}(\theta)\frac{b_{1}^{2}-b_{0}^{2}}{2}+\sum_{l\in\mathbb{Z},\,l\neq
		0}c_{k,l}(\theta)\frac{(b_{1}-b_{0})^{2}}{2\pi il}\exp\left(  \frac{2\pi
		ilb_{0}}{b_{1}-b_{0}}\right)  .
	\]
	A similar closed-form expression for predicted mean can be derived similarly.\footnote{In
		addition to the filtered mean, one may also use the maximum a posteriori
		estimator (MAP) to estimate the latent state, defined by $\hat{\mu}%
		_{k}^{\text{MAP}}=\arg\max_{x\in\mathcal{X}}p(x|\mathbf{Y}_{k}=\mathbf{y}%
		_{k};\theta)$. This estimator is particularly useful when the filtering
		distribution is discrete-valued, as in regime-switching
		models in Example \ref{example:regime}, or when the filtering density is
		multimodal, as in the QTSM model defined in (\ref{eq:qtsm}). $\hat{\mu
		}_{k}^{\text{MAP}}$ can also be recovered from the Fourier coefficients.} Here, we also provide a straightforward interpretation of the Fourier
	coefficients $\{c_{k,l}(\theta)\}_{l\in\mathbb{Z}}.$ For expositional clarity
	and recognizing that the state space of the latent variable can be
	appropriately scaled without loss of generality, we assume that $2\pi
	/(b_{1}-b_{0})$ is much smaller than one. Under this premise, the
	low-frequency coefficients, namely, when $\left\vert l\right\vert $ is small,
	mainly encode the large-scale features of the filtered density, such as its
	mean value and variance. Let $\omega_{l}=2\pi l/(b_{1}-b_{0})$, and for
	sufficiently low frequencies where $\left\vert \omega_{l}\right\vert <1,$ the Taylor expansion gives us%
	\[
	c_{k,l}(\theta)=\frac{1}{b_{1}-b_{0}}\mathbb{E}\left[  \left.  \exp\left(
	-i\omega_{l}X_{k}\right)  \right\vert \mathbf{Y}_{k}=\mathbf{y}_{k}%
	;\theta\right]  =\sum_{j=0}^{\infty}\frac{\left(  -i\right)  ^{j}\omega
		_{l}^{j}}{(b_{1}-b_{0})j!}\mathbb{E}\left[  X_{k}^{j}|\mathbf{Y}%
	_{k}=\mathbf{y}_{k};\theta\right]  .
	\]
	Intuitively, when $\left\vert l\right\vert $ is small, say $l=1,2,$ by
	dropping the higher order terms in the above Taylor expansion as $\left\vert
	\omega_{l}\right\vert <1$, we have%
	\[
	c_{k,l}(\theta)\approx\frac{1}{b_{1}-b_{0}}\left(  1-i\omega_{l}%
	\mathbb{E}\left[  X_{k}|\mathbf{Y}_{k}=\mathbf{y}_{k};\theta\right]
	-\frac{\omega_{l}^{2}}{2}\mathbb{E}\left[  X_{k}^{2}|\mathbf{Y}_{k}%
	=\mathbf{y}_{k};\theta\right]  \right)  .
	\]
	This implies that the low-frequency coefficients $c_{k,1}(\theta)$ and
	$c_{k,2}(\theta)$ primarily encode the filtered mean and variance.
	High-frequency coefficients, by contrast, describe the local features of the
	filtered density. They become important when the filtered distribution is
	sharply concentrated.\footnote{For instance, suppose that, for some small
		$h>0$, $p(x_{k}|\mathbf{y}_{k};\theta)\approx h^{-1}\mathbf{1}_{|x_{k}%
			-x_{k}^{\ast}|\leq h/2}$, with $[x_{k}^{\ast}-h/2,x_{k}^{\ast}+h/2]\subset
		\lbrack b_{0},b_{1}]$. Then, by the definition of Fourier coefficients,
		$c_{k,l}(\theta)\approx(b_{1}-b_{0})^{-1}e^{-i\omega_{l}x_{k}^{\ast}%
		}\operatorname{sinc}(\omega_{l}h/2)$, where $\operatorname{sinc}(z)=\sin
		(z)/z$. Hence $|c_{k,l}(\theta)|\approx(b_{1}-b_{0})^{-1}|\operatorname{sinc}%
		(\pi lh/(b_{1}-b_{0}))|$, which does not decay rapidly as $|l|$ increases.}
	Accurately representing such a localized filtered density requires retaining
	sufficiently high-frequency Fourier coefficients. Therefore, by incorporating
	high-frequency Fourier coefficients, the proposed Fourier recursion is
	suitable for sharply peaked filtered densities, which arise in a wide range of
	modern econometric settings and are often difficult to handle with existing
	filtering methods.
	
	Obtaining point estimation of MMLEs alone, however, is not sufficient for
	statistical inference in applications. Researchers must
	also quantify estimation uncertainty, construct confidence intervals, and test
	economically meaningful restrictions. As shown in Section \ref{sec:theory},
	the asymptotic covariance matrix of the MMLE can be consistently estimated by
	the score vector and the Hessian matrix of the sample log-likelihood
	$S_{n}^{\ell}(\theta)=\sum_{k=1}^{n}\log\mathcal{L}_{k}(\theta)$. This
	theoretical connection makes the computation of these first- and second-order
	derivatives an essential component of likelihood-based inference.
	
	Accordingly, we examine whether the score and Hessian of the sample
	log-likelihood can be evaluated efficiently within the filtering framework in
	Theorem \ref{thm:fourier series update 1}. In particular, the
	Fourier-coefficient recursion used for likelihood updating and filtering can
	be augmented with recursive equations for its first- and second-order
	derivatives w.r.t. the model parameters, namely, $\nabla_{\theta}%
	c_{k,l}\left(  \theta\right)  $ and $\nabla_{\theta}^{2}c_{k,l}\left(
	\theta\right)  $. This extension yields a deterministic online recursion for
	the gradient and Hessian matrix of the sample log-likelihood, thereby
	providing the derivative quantities required for the inferential analysis
	developed later. Let $H_{k}^{(l,s)}(\theta):=H\left(  \frac{2\pi l}%
	{b_{1}-b_{0}},y_{k},\frac{2\pi s}{b_{1}-b_{0}},y_{k-1};\theta\right)  $,
	$\nabla_{\theta}S_{n}^{\ell}(\theta)$ denotes the column gradient, and
	$\nabla_{\theta}^{2}S_{n}^{\ell}(\theta)$ denotes the Hessian matrix. We have
	the following recursion for $\nabla_{\theta}S_{n}^{\ell}(\theta)$ and
	$\nabla_{\theta}^{2}S_{n}^{\ell}(\theta).$

\begin{proposition}
	\label{prop:hessian updating}
	Assume that $H_{k}^{(l,s)}(\theta)$ is twice differentiable with respect to
	$\theta\in\Theta$. Then, $\nabla_{\theta}S_{k}^{\ell}(\theta)$ and
	$\nabla_{\theta}^{2}S_{k}^{\ell}(\theta)$ can be recursively updated as
	\[
		\nabla_{\theta}S_{k}^{\ell}(\theta)=\nabla_{\theta}S_{k-1}^{\ell}%
		(\theta)+\frac{\nabla_{\theta}\mathcal{L}_{k}(\theta)}{\mathcal{L}_{k}%
			(\theta)},\quad\nabla_{\theta}^{2}S_{k}^{\ell}(\theta)={}\nabla_{\theta}%
		^{2}S_{k-1}^{\ell}(\theta)+\frac{\nabla_{\theta}^{2}\mathcal{L}_{k}(\theta
			)}{\mathcal{L}_{k}(\theta)}-\frac{\nabla_{\theta}\mathcal{L}_{k}(\theta
			)\nabla_{\theta}^{\top}\mathcal{L}_{k}(\theta)}{\mathcal{L}_{k}^{2}(\theta)},
	\]
	where $\nabla_{\theta}\mathcal{L}_{k}(\theta)$ and
	$\nabla_{\theta}^{2}\mathcal{L}_{k}(\theta)$ can also be recursively updated. Namely, we have
	\[
		\nabla_{\theta}\mathcal{L}_{k}(\theta)
		=\mathcal{M}_{k,0}^{(1)}[\mathbf{c}_{k-1}(\theta)],\quad
		\nabla_{\theta}^{2}\mathcal{L}_{k}(\theta)
		=\mathcal{M}_{k,0}^{(2)}[\mathbf{c}_{k-1}(\theta)],
	\]
     where $\mathbf{c}_{k-1}(\theta)=\{c_{k-1,s}(\theta)\}_{s\in\mathbb Z}$, and for $\theta$-dependent sequence $\mathbf{v}=\{v_s\}_{s\in \mathbb{Z}},$ the operators $\mathcal{M}_{k,l}^{(1)}$ and $\mathcal{M}_{k,l}^{(2)}$
are defined by 
     \begin{align}
     &\mathcal{M}_{k,l}^{(1)}[\mathbf{v}]:=\sum_{s\in\mathbb Z}
		[
		H_{k}^{(-l,s)}(\theta)\nabla_{\theta}v_s
		+v_s\nabla_{\theta}H_{k}^{(-l,s)}(\theta)
		],\\ \nonumber
         &\mathcal{M}_{k,l}^{(2)}[\mathbf{v}]
		:=\sum_{s\in\mathbb Z}
		[
		H_{k}^{(-l,s)}(\theta)\nabla_{\theta}^{2}v_s
		+2Sym(
		\nabla_{\theta}v_s,
		\nabla_{\theta}H_{k}^{(-l,s)}(\theta)
		)
		+v_s\nabla_{\theta}^{2}H_{k}^{(-l,s)}(\theta)
		],\nonumber
     \end{align}
	 with $\text{Sym}(a,b)=(ab^{\intercal}+ba^{\intercal})/2$. Moreover, the first and second derivatives of the filtering Fourier coefficient can be recursively updated as
	\begin{align}
	    \nabla_{\theta}c_{k,l}(\theta)
		&=
		\frac{\mathcal{M}_{k,l}^{(1)}[\mathbf{c}_{k-1}(\theta)]}
		{\mathcal{L}_{k}(\theta)(b_{1}-b_{0})}
		-\frac{c_{k,l}(\theta)\nabla_{\theta}\mathcal{L}_{k}(\theta)}
		{\mathcal{L}_{k}(\theta)},\label{eq:gradient fourier coeff} \\
        \nabla_{\theta}^{2}c_{k,l}(\theta)
		&=
		\frac{\mathcal{M}_{k,l}^{(2)}[\mathbf{c}_{k-1}(\theta)]}
		{\mathcal{L}_{k}(\theta)(b_{1}-b_{0})}
		-\frac{
		c_{k,l}(\theta)\nabla_{\theta}^{2}\mathcal{L}_{k}(\theta)
		+2\text{Sym}\!\left(
		\nabla_{\theta}c_{k,l}(\theta),
		\nabla_{\theta}\mathcal{L}_{k}(\theta)
		\right)}
		{\mathcal{L}_{k}(\theta)}.
		\label{eq:hessian fourier coeff}
	\end{align}
\end{proposition}

	Proposition \ref{prop:hessian updating} illustrates how to obtain the Hessian matrix
	for the marginal likelihood under LMMs through $\{(c_{k,l}(\theta
	),\nabla_{\theta}c_{k,l}(\theta),\nabla_{\theta}^{2}c_{k,l}(\theta
	))\}_{l\in\mathbb{Z}}$. This gives a deterministic online method.\footnote{In contrast,
	existing approaches for Hessian computation for LMMs often rely on
	backward sampling, MCMC, or particle smoothing, and are therefore typically
	computationally demanding (see, e.g., \cite{PoyiadjisDoucetSingh2011}).} The
	proposed method only requires tractable $H$, $\nabla_{\theta}H,$ and
	$\nabla_{\theta}^{2}H.$ With the availability of the Hessian matrix, the
	asymptotic properties of the MMLEs can be established; detailed theoretical
	analysis is deferred to Section \ref{sec:theory}.
	
	 For practical implementation, can the exact Fourier recursion in Theorem \ref{thm:fourier series update 1} be made operational in practice? The first question is whether $H$ and $H_X$ are tractable. To effectively tackle the tractability of $H$ and $H_{X}$, it is often fruitful
	to exploit the inherent properties present in certain LMM architectures. Among
	the broad class of LMMs, state-space models (SSMs, hereafter) and affine
	models constitute two structurally pivotal subclasses as illustrated by the
	diverse specifications in Examples \ref{example:linear smm}-\ref{example:bns}.
	Next, we will see how the Fourier recursion operates under these two different
	classes of LMMs.
	
	SSMs exhibit a special structure that significantly simplifies inference.
	Namely, two extra conditional independence assumptions are satisfied, i.e.,
	\begin{subequations}
		\begin{align}
			\mathbb{P}\left(  X_{k}\in dx|\mathbf{X}_{k-1}=\mathbf{x}_{k-1},\mathbf{Y}%
			_{k-1}=\mathbf{y}_{k-1};\theta\right)   &  =\mathbb{P}\left(  X_{k}\in
			dx|X_{k-1}=x_{k-1};\theta\right)  ,\label{eq:SSM 1}\\
			\mathbb{P}\left(  Y_{k}\in dy|\mathbf{X}_{k}=\mathbf{x}_{k},\mathbf{Y}%
			_{k-1}=\mathbf{y}_{k-1};\theta\right)   &  =\mathbb{P}\left(  Y_{k}\in
			dy|X_{k}=x_{k};\theta\right)  . \label{eq:SSM 2}%
		\end{align}
		Let $q_{X}\left(  x_{k}|x_{k-1};\theta\right)  =\mathbb{E}\left[
		\delta\left(  X_{k}-x_{k}\right)  |X_{k-1}=x_{k-1};\theta\right]  ,$
		$q_{Y}\left(  y_{k}|x_{k};\theta\right)  =\mathbb{E}\left[  \delta\left(
		Y_{k}-y\right)  |X_{k}=x_{k};\theta\right]  ,$ and $\varphi_{Y}\left(
		y,u;\theta\right)  =\int_{\mathcal{X}}q_{Y}\left(  y|x;\theta\right)
		e^{iux}\nu_{X}(dx)$ be the Fourier transform of $q_{Y}\left(  y|x;\theta
		\right)  $ w.r.t. $x.$ The Fourier recursion for SSMs can be simplified as follows.
	\end{subequations}
	\begin{corollary}
		\label{thm:fourier series update 3}When the LMM reduces to an SSM, i.e.,
		(\ref{eq:SSM 1}) and (\ref{eq:SSM 2}) are satisfied, we assume that $X_{k}$ is
		compactly supported on $[b_{0},b_{1}]$. Then, the likelihood update function
		$\mathcal{L}_{k}(\theta)$ is given by%
		\[
		\mathcal{L}_{k}\left(  \theta\right)  =\sum_{l\in\mathbb{Z}}c_{k,l}^{-}\left(
		\theta\right)  \varphi_{Y}\left(  y_{k},\frac{2\pi l}{b_{1}-b_{0}}%
		;\theta\right)  ,
		\]
		where $\{c_{k,l}^{-}\left(  \theta\right)  \}_{l\in\mathbb{Z}}$ denote the
		Fourier coefficients for one-step predictive density defined in
		(\ref{eq:fourier predictive density}). The Fourier coefficients for filtered
		density, i.e., $\left\{  c_{k,l}\left(  \theta\right)  \right\}
		_{l\in\mathbb{Z}}$ can be updated from $\{c_{k,l}^{-}\left(  \theta\right)
		\}_{l\in\mathbb{Z}}$ as
		\[
		c_{k,l}\left(  \theta\right)  =\sum_{r\in\mathbb{Z}}\frac{c_{k,r}^{-}\left(
			\theta\right)  }{\mathcal{L}_{k}\left(  \theta\right)  \left(  b_{1}%
			-b_{0}\right)  }\varphi_{Y}\left(  y_{k},\frac{2\pi\left(  r-l\right)  }%
		{b_{1}-b_{0}};\theta\right)  ,
		\]
		and the Fourier coefficients $\{c_{k+1,l}^{-}\left(  \theta\right)
		\}_{l\in\mathbb{Z}}$ can be recovered from $\left\{  c_{k,l}\left(
		\theta\right)  \right\}  _{l\in\mathbb{Z}}$ exactly the same as
		(\ref{eq:fourier prediction update}), i.e.,
		\[
		c_{k+1,l}^{-}\left(  \theta\right)  =\sum_{r\in\mathbb{Z}}\frac{c_{k,r}\left(
			\theta\right)  }{b_{1}-b_{0}}H_{X}^{\text{SSM}}\left(  \frac{-2\pi l}%
		{b_{1}-b_{0}},\frac{2\pi r}{b_{1}-b_{0}};\theta\right)  ,
		\]
		with $H_{X}^{\text{SSM}}(u,s;\theta):=H_{X}(u,s,y_{0};\theta)$, where we omit
		$y_{0}$ due to (\ref{eq:SSM 1}) for simplicity.
	\end{corollary}
	
	We provide the detailed derivation and discussion of Corollary
	\ref{thm:fourier series update 3} in \ref{appendix:proofs}.\footnote{In
		Bayesian analysis, a well-estimated filtered density in SSMs also provides a
		natural method for latent path recovery, namely sampling latent trajectories
		from the posterior smoothing law. For SSMs, we have the backward factorization
		$p(x_{k}|\mathbf{y}_{n};\theta)=p(x_{n}|\mathbf{y}_{n};\theta)\prod^{n}_{k=0}
		p(x_{k}|x_{k+1},\mathbf{y}_{k};\theta)$, and $p(x_{k}|x_{k+1},\mathbf{y}%
		_{k};\theta)\propto p(x_{k}|\mathbf{y}_{k};\theta)q_{X}(x_{k+1}|x_{k};\theta
		)$. Once the filtered density $p(x_{k}|\mathbf{y}_{k};\theta)$ is obtained, a
		latent path can be recovered by first drawing $x_{n}^{*}\sim p(x_{n}%
		|\mathbf{y}_{n};\theta)$, and then recursively drawing backward, for
		$k=n-1,...,0$, $x_{k}^{*} \sim p(x_{k}|x_{k+1}^{*},\mathbf{y}_{k};\theta)$.}
	In fact, for SSMs, the required transformed functions $\varphi_{Y}$ and
	$H_{X}^{\text{SSM}}$ are tractable in many cases. Panel A of Table
	\ref{tab:representative-fourier-recursion} in Online Appendix reports the
	model-specific expressions of these functions for the examples introduced in
	Section \ref{sec:lmms}. Besides likelihood update, prediction and filtering,
	smoothing in SSMs, i.e., estimating the latent state $X_{k}$ given all
	observations $\mathbf{y}_{n},$ is also a core issue. Smoothing requires us to
	evaluate $p\left(  x_{k}|\mathbf{y}_{n};\theta\right)  $, in order to obtain a
	more accurate estimation of the latent state.\footnote{Conditioning on the full sequence $\mathbf{y}_{n}$ allows
		smoothing to revise filtering estimates and separate persistent latent-state
		changes from transitory signals, such as temporary recessionary episodes,
		volatility spikes, or apparent structural breaks (see, e.g.,
		\cite{Bates2019JF}).} Smoothing has been computationally demanding, as commonly used approaches
	often rely on simulation-based procedures, and may require substantial
	model-specific choices in their implementation. We
	also provide a unified smoothing method based on Corollary
	\ref{thm:fourier series update 3}. See the detailed derivation and
	discussion in Online Appendix \ref{sec:smoothing}.
	
	\begin{proposition}
		\label{prop:smoothing}Under the same conditions of Corollary
		\ref{thm:fourier series update 3}, let $r_{k}(x_{k}%
		;\theta):=p\left(  x_{k}|\mathbf{y}_{n};\theta\right)  /p\left(
		x_{k}|\mathbf{y}_{k};\theta\right)$ denote the smoothing ratio. Assume that $r_{k}(x_{k};\theta)$
		admits the Fourier representation $r_{k}(x_{k};\theta)=\sum_{l\in\mathbb{Z}%
		}\bar{b}_{k,l}\left(  \theta\right)  \exp\left(  \frac{2\pi ilx_{k}}%
		{b_{1}-b_{0}}\right)  ,$ where the Fourier coefficients $\{\bar{b}%
		_{k,l}\left(  \theta\right)  \}_{l\in\mathbb{Z}}$ can be recursively computed
		backward by%
		\begin{equation}
			\bar{b}_{k,l}\left(  \theta\right)  =\sum_{s\in\mathbb{Z}}\frac{\bar
				{b}_{k+1,s}\left(  \theta\right)  }{(b_{1}-b_{0})\mathcal{L}_{k+1}(\theta
				)}H^{\text{SSM}}\left(  \frac{2\pi s}{b_{1}-b_{0}},y_{k+1},\frac{-2\pi
				l}{b_{1}-b_{0}};\theta\right)  . \label{eq:smoothing 1}%
		\end{equation}
		Then the smoothed density $p\left(  x_{k}|\mathbf{y}_{n};\theta\right)  $ has
		Fourier representation $p\left(  x_{k}|\mathbf{y}_{n};\theta\right)
		=\sum_{l\in\mathbb{Z}}b_{k,l}\left(  \theta\right)  \exp\left(  \frac{2\pi
			ilx_{k}}{b_{1}-b_{0}}\right)  .$ Its Fourier coefficients $\{b_{k,l}\left(
		\theta\right)  \}_{l\in\mathbb{Z}}$ can be obtained by the following
		convolution:%
		\begin{equation}
			b_{k,l}\left(  \theta\right)  =\sum_{j+s=l}\bar{b}_{k,s}(\theta)c_{k,j}%
			(\theta), \label{eq:smoothing 2}%
		\end{equation}
		where $\{c_{k,j}(\theta)\}_{j\in\mathbb{Z}}$ denotes the Fourier coefficients
		of filtered density $p\left(  x_{k}|\mathbf{y}_{k};\theta\right)  .$
	\end{proposition}
	
	Besides the SSMs, another prominent class of LMMs merits special attention,
	particularly those arising from continuous-time formulations (see, e.g.,
	\cite{duffie2000transform}, \cite{bates2006maximum}), including the DPS model
	in Example \ref{example:heston} and the BNS model in Example \ref{example:bns}%
	.\footnote{Time-changed L\'{e}vy processes provide another affine class, in
		which the stochastic clock is driven by a latent Markov activity process (see, e.g.,
		\cite{carr2004time}). Moreover, the L\'{e}vy-driven LMMs,
		such as additive-noise L\'{e}vy-driven SDEs, arise by replacing the
		Brownian-driven noise in latent diffusion models with a general L\'{e}vy
		process. This constitutes a further affine class} For this class of LMMs, the analytically
	tractable object is often not the transition density $p_{(X,Y)}$ itself, but
	rather its conditional characteristic function
	\[
	\phi_{(X,Y)}(u,v|x_{0},y_{0};\theta)=\mathbb{E}\left[  e^{iuX_{1}+ivY_{1}%
	}|X_{0}=x_{0},Y_{0}=y_{0};\theta\right]  =e^{C_{0}(u,v,y_{0};\theta
		)+C_{1}(u,v,y_{0};\theta)x_{0}},
	\]
	where $C_{0}$ and $C_{1}$ are known functions, possibly depending on the model
	parameter $\theta$. For affine models, the Fourier recursion can also be simplified. Let
	\begin{align*}
		\Phi\left(  u,v,y_{0},s;\theta\right)   &  =\int_{b_{0}}^{b_{1}}e^{ix_{0}}%
		\phi_{(X,Y)}(u,v|x_{0},y_{0};\theta)dx_{0}\\
		&  =\frac{e^{\left(  is+C_{1}(u,v,y_{0};\theta)\right)  b_{1}}-e^{\left(
				is+C_{1}(u,v,y_{0};\theta)\right)  b_{0}}}{is+C_{1}(u,v,y_{0};\theta)}%
		e^{C_{0}(u,v,y_{0};\theta)}
	\end{align*}
	if $is+C_{1}(u,v,y_{0})\neq0,$ and $\Phi\left(  u,v,y_{0},s;\theta\right)
	=\left(  b_{1}-b_{0}\right)  e^{C_{0}(u,v,y_{0};\theta)}$ if $is+C_{1}%
	(u,v,y_{0})=0.$ The following result summarizes the recursion implied by
	Theorem \ref{thm:fourier series update 1}. We refer to online Appendix \ref{sec:affine} for the detailed derivation.
	
	\begin{corollary}
		\label{thm:fourier series update 2} Suppose the observed process $Y_{k}$ takes
		values in $[a_{0},a_{1}]$, and $X_{k}$ is compactly supported on $[b_{0}%
		,b_{1}]$. Then the likelihood update function $\mathcal{L}_{k}(\theta)$ is
		given by
		\begin{equation}
			\mathcal{L}_{k}(\theta)=\sum_{s,l\in\mathbb{Z}}\frac{c_{k-1,s}(\theta)}%
			{a_{1}-a_{0}}e^{\frac{-2\pi ily_{k}}{a_{1}-a_{0}}}\Phi\left(  0,\frac{2\pi
				l}{a_{1}-a_{0}},y_{k-1},\frac{2\pi s}{b_{1}-b_{0}};\theta\right)  ,
			\label{eq:fourier chf 1}%
		\end{equation}
		where the Fourier coefficients $\left\{  c_{k,l}(\theta)\right\}
		_{l\in\mathbb{Z}}$ satisfy
		\begin{equation}
			c_{k,l}(\theta)=\sum_{s,r\in\mathbb{Z}}\frac{c_{k-1,s}(\theta)}{m_{a,b}%
				\mathcal{L}_{k}(\theta)}e^{-\frac{2\pi iry_{k}}{a_{1}-a_{0}}}\Phi\left(
			\frac{-2\pi l}{b_{1}-b_{0}},\frac{2\pi r}{a_{1}-a_{0}},y_{k-1},\frac{2\pi
				s}{b_{1}-b_{0}};\theta\right)  , \label{eq:fourier chf 2}%
		\end{equation}
		with $m_{a,b}=(b_{1}-b_{0})\left(  a_{1}-a_{0}\right).$
	\end{corollary}
	
	\subsection{Truncated Fourier recursion implementation\label{sec:truncation}}

	The Fourier recursion developed in Theorem \ref{thm:fourier series update 1} and Corollary \ref{thm:fourier series update 3}-\ref{thm:fourier series update 2} are exact when the infinite sequence of Fourier coefficients is retained. However, infinite coefficients are not directly implementable. The central computational step is to finite-dimensionalize this recursion by truncating the Fourier expansion to a finite set of coefficients. In this section we introduce a truncated Fourier recursion for implementation.
	The system is designed so that, regardless of the truncation level, it
	reliably performs filtering and likelihood updates, highlighting a key advantage of the Fourier recursion. For clarity,
	we present the truncation strategy for Theorem
	\ref{thm:fourier series update 1}, and those for Corollary
	\ref{thm:fourier series update 3} and \ref{thm:fourier series update 2} follow
	similarly. Importantly, 
	the implementation does not come at the expense of theoretical guarantees: the truncation error remains uniformly controlled as established later in Section \ref{sec:theory}, while the resulting approximation retains high accuracy and numerical stability as demonstrated in Section \ref{sec:simulation}.
	
	A naive approach is to select a truncation level $L$, and
	approximate the filtered density $p(x_{1}|y_{1};\theta)$ using the truncated
	Fourier series $\sum_{l\in D_{L}}c_{1,l}(\theta)\exp\left({\frac{ 2\pi
			ilx_{1}}{b_{1}-b_{0}}}\right),$ where $D_{L}=\left\{  l:\left\vert l\right\vert \leq L\right\}.$ However, this approximation does not
	guarantee positivity. The solution lies in applying the Fej\'{e}r kernel, a
	powerful tool in Fourier and harmonic analysis, to construct a modified series
	$\hat{p}^{(L)}(x_{1}|y_{1};\theta)=\sum_{l\in D_{L}}\hat{c}_{1,l}^{(L)}%
	(\theta)\exp\!\left(  \frac{2\pi ilx_{1}}{b_{1}-b_{0}}\right)  , $ where
	$\hat{c}_{1,l}^{(L)}(\theta)=c_{1,l}(\theta)(1-\left\vert l\right\vert \left(
	L+1\right)  ^{-1}).$ For any $L,$ we can show that $\hat{p}^{(L)}(x_{1}%
	|y_{1};\theta)$ is nonnegative and satisfies $\int_{\mathcal{X}}\hat
	{p}^{(L)}(x_{1}|y_{1};\theta)dx_{1}=1.$ 
	This implies the truncated filtered density $\hat{p}^{(L)}(x_{1}|y_{1}%
	;\theta)$ is a true density function. The formal justification of this
	property is presented in online Appendix \ref{sec:uni convergence fourier}.
	
	Given this fact, the likelihood update at $k=2$ in Theorem
	\ref{thm:fourier series update 1} can be written in integral form as
	\[
	\mathcal{\hat{L}}_{2}^{(L)}(\theta)=\sum_{s\in D_{L}}\hat{c}_{1,s}%
	^{(L)}(\theta)H\left(  0,y_{2},\frac{2\pi s}{b_{1}-b_{0}},y_{1};\theta\right)
	=\int_{\mathcal{X}}p_{Y}\left(  y_{2}|x_{1},y_{1};\theta\right)  \hat{p}%
	^{(L)}(x_{1}|y_{1};\theta)dx_{1}.
	\]
	Since $\hat{p}^{(L)}(x_{1}|y_{1};\theta)$ remains positive, $\mathcal{\hat{L}%
	}_{2}^{(L)}(\theta)$ is strictly positive, which is essential for
	implementation as a negative likelihood update will prevent the Fourier
	updating system from moving forward. Meanwhile, by Theorem
	\ref{thm:fourier series update 1}, the Fej\'{e}r-adjusted Fourier coefficients
	at $k=2$ can be updated as
	\[
	\hat{c}_{2,l}^{(L)}\left(  \theta\right)  =\sum_{s\in D_{L}}\frac
	{(1-\left\vert l\right\vert \left(  L+1\right)  ^{-1})\hat{c}_{1,s}%
		^{(L)}(\theta)}{\mathcal{\hat{L}}_{2}^{(L)}(\theta)(b_{1}-b_{0})}H\left(
	\frac{-2\pi l}{b_{1}-b_{0}},y_{2},\frac{2\pi s}{b_{1}-b_{0}},y_{1}%
	;\theta\right)  
	\]
	for $l\in D_{L}.$ The resulting estimates of the filtered density
	$p(x_{2}|\mathbf{y}_{2};\theta)$ and likelihood update $\mathcal{L}_{3}%
	(\theta)$, i.e.,%
	\[
	\mathcal{\hat{L}}_{3}^{(L)}(\theta) =\sum_{s\in D_{L}}\hat{c}_{2,s}%
	^{(L)}(\theta)H\left(  0,y_{3},\frac{2\pi s}{b_{1}-b_{0}},y_{2};\theta\right)
	,\quad\hat{p}^{(L)}(x_{2}|\mathbf{y}_{2};\theta) =\sum_{l\in D_{L}}\hat
	{c}_{2,l}^{(L)}(\theta)\exp\!\left(  \frac{2\pi ilx_{2}}{b_{1}-b_{0}}\right)
	\]
	are well defined because the Fej\'{e}r-adjusted Fourier coefficients preserve
	a valid, nonnegative, and normalized filtered density, ensuring that the
	corresponding likelihood update remains strictly positive. By using similar
	arguments, nonnegative $\hat{p}^{(L)}(x_{3}|\mathbf{y}_{3};\theta)$ and
	$\mathcal{\hat{L}}_{4}^{(L)}\left(  \theta\right)  $ can be obtained from
	$\hat{p}^{(L)}(x_{2}|\mathbf{y}_{2};\theta)$, and so on up to $\hat{p}%
	^{(L)}(x_{n}|\mathbf{y}_{n};\theta)$ and $\mathcal{\hat{L}}_{n}^{(L)}\left(
	\theta\right)  $.
	
	\begin{proposition}
		\label{prop:implement}
		Under the same conditions in Theorem \ref{thm:fourier series update 1}, for
		any truncation level $L\geq2,$ and every $k\geq2,$ the
		approximated likelihood update
		\begin{equation}
			\mathcal{\hat{L}}_{k}^{(L)}(\theta)=\sum_{s\in D_{L}}\hat{c}_{k-1,s}%
			^{(L)}(\theta)H\left(  0,y_{k},\frac{2\pi s}{b_{1}-b_{0}},y_{k-1}%
			;\theta\right)  \label{eq:implement 1}%
		\end{equation}
		is strictly positive, where the Fej\'{e}r-adjusted Fourier coefficients can be
		updated by
		\begin{equation}
			\hat{c}_{k,l}^{(L)}\left(  \theta\right)  =\sum_{s\in D_{L}}\frac
			{(1-\left\vert l\right\vert \left(  L+1\right)  ^{-1})\hat{c}_{k-1,s}%
				^{(L)}(\theta)}{\mathcal{\hat{L}}_{k}^{(L)}(\theta)(b_{1}-b_{0})}H\left(
			\frac{-2\pi l}{b_{1}-b_{0}},y_{k},\frac{2\pi s}{b_{1}-b_{0}},y_{k-1}%
			;\theta\right)  ,\text{ }l\in D_{L}. \label{eq:implement 2}%
		\end{equation}
		The associated filtered density $\hat{p}^{(L)}(x_{k}%
		|\mathbf{y}_{k};\theta)=\sum_{l\in D_{L}}\hat{c}_{k,l}^{(L)}(\theta)
		e^{\frac{ 2\pi ilx_{k}}{b_{1}-b_{0}}}$ is nonnegative on $[b_{0},b_{1}].$
	\end{proposition}
	
	The main implication of the proposition is
	that this Fej\'{e}r-adjusted truncated recursion generates a sequence of valid
	approximated filtered densities and strictly positive likelihood updates,
	ensuring that the entire recursion remains well-defined and numerically
	stable. Next, we analyze the computation cost of the Fourier recursion implementation.
	The likelihood update requires summing over $s\in D_{L},$ and therefore has
	cost $O\left(  \left\vert D_{L}\right\vert \right)  =O(\left(  2L+1\right)
	).$ The dominant cost comes from the update of the Fourier coefficients. For
	each $l\in D_{L},$ the coefficient $\hat{c}_{k+1,l}^{(L)}(\theta)$ is computed
	by summing over all $s\in D_{L}.$ Since there are $\left\vert D_{L}\right\vert
	$ coefficients to update and each coefficient requires $\left\vert
	D_{L}\right\vert $ summation terms, one recursive update costs $O(\left\vert
	D_{L}\right\vert ^{2})=O(\left(  2L+1\right)  ^{2}).$ Repeating this recursion
	over $n$ observations gives the total computational cost $O\left(
	n(2L+1)^{2}\right)  .$ We also provide the discussion for the multivariate
	latent states and the associated acceleration of Proposition 
	\ref{prop:implement} in Online Appendix \ref{sec:multi accelerate}.
	
	\section{Approximation and large-sample theory of LMMs\label{sec:theory}}
	
	Estimation and inference for truncated Fourier recursion involve three distinct technical difficulties. First, truncation errors arise at every filtering step and propagate recursively. Second, even abstracting from truncation, each likelihood update $\ell_k(\theta)$ depends on the evolving filter rather than only on the current state variables, so the marginal likelihood cannot be written directly as a standard additive functional. This prevents a direct application of routine MLE theory. Third, the Fisher information of MMLE is intrinsically difficult to evaluate. In this section, we address these three difficulties in turn. First, we establish the uniform
	convergence of the truncated Fourier recursion and analyze how the truncation
	level affects the accuracy. Second, we prove the consistency and asymptotic
	normality of the approximate MMLE (AMMLE, hereafter)
	obtained from the truncated Fourier filter for general LMMs. Third, we develop a feasible framework
	for estimating the asymptotic covariance matrix of the proposed estimators by
	using the Hessian of the sample log-likelihood. Finally, we extend the asymptotic theory to the misspecified setting, where the maintained LMM family may not contain the true transition law.
	
	Hereafter, we focus exclusively on the case where the latent process $X_{k}$
	is continuous-valued and the dominating measure $\nu_{X}$ turns into the
	Lebesgue measure. When $X_{k}$ is discrete-valued, the infinite summations in
	Theorem \ref{thm:fourier series update 1} reduce to finite summations. This
	renders the truncation exact and entirely eliminates approximation errors.
	Therefore, the asymptotic properties of MMLEs, including the consistency and
	asymptotic normality, can thus be straightforwardly established following the related literature. We start by imposing appropriate conditions.
	
	\begin{assumption}
		\label{ass:S6} Let $d_X$, $d_Y$ be the dimensions of the latent and observed state spaces, respectively. The state spaces $\mathcal{X}$ and $\mathcal{Y}$ are both
		compactly supported and take the form $\mathcal{X}=[b_{0},b_{1}]^{d_{X}}$ and $\mathcal{Y}%
		=[a_{0},a_{1}]^{d_{Y}},$ for  constants $b_{0}<b_{1}$ and
		$a_{0}<a_{1}$.
	\end{assumption}
	
	\begin{assumption}
		\label{assump: bound derivative}Assume that the derivative $\partial
		p(x_{1},y_{1}|x_{0},y_{0};\theta)/\partial x_{1}$ is uniformly bounded, i.e.,
		\[
		\sup_{y_{1},y_{0},x_{0}}\sup_{\theta\in\Theta}\left\vert \frac{\partial
			p_{(X,Y)}(x_{1},y_{1}|x_{0},y_{0};\theta)}{\partial x_{1}}\right\vert \leq
		C_{1}
		\]
		for some positive constant $C_{1}$. For every $k\in\mathbb{N}^{\ast}$, there
		exists a constant $0<\epsilon_{k}\leq1$ and a reference probability measure
		$\lambda_{k}$, such that
		\[
		\epsilon_{k}\lambda_{k}(A)\leq\int_{A}p_{(X,Y)}(x_{1},y_{k}|x_{0}%
		,y_{k-1};\theta)dx_{1}\leq\epsilon^{-1}_{k}\lambda_{k}(A)
		\]
		for any $x_{0}\in\mathcal{X}$ and the Borel set $A\subset\mathcal{X}$. In
		particular, we assume $\epsilon_{k}$ is uniformly bounded away from zero,
		i.e., $\inf_{k\geq1}\epsilon_{k}>\epsilon>0$ for some deterministic constant
		$\epsilon>0$.
	\end{assumption}
	
	The compact-support assumption is mainly introduced for technical convenience,
	as it facilitates uniform control of the transition density and its
	derivatives. As explained in Section \ref{sec:exact fourier}, the compactness can be obtained by introducing a one-to-one mapping without changing inherent property of LMMs. The condition in Assumption \ref{assump: bound derivative} is a
	standard mixing condition used to establish exponential filter stability and
	uniform-in-time particle approximation bounds (see, e.g., \cite{gland2004}).
	Next, we assume that the marginal likelihood estimator $\hat{\theta
	}^{(n)}:=\arg\max_{\theta\in\Theta}\prod^{n}_{k=1}\mathcal{L}_{k}(\theta)$ exists and is
	unique in $\Theta.$ Since the truncated filtered density $\hat{p}^{(L)}%
	(x_{k}|\mathbf{y}_{k};\theta)$ generated by Proposition
	\ref{prop:implement} is positive for $k=1,2,...,n,$ approximated
	log-likelihood $\hat{\ell}^{(n,L)}=\sum_{k=1}^{n}\hat{\ell}_{k}^{(L)}(\theta)$
	is well defined. Let $\hat{\theta}_{\text{AMMLE}}^{(n,L)}%
	=\operatorname*{argmax}_{\theta\in\Theta}\sum_{k=1}^{n}\hat{\ell}_{k}%
	^{(L)}(\theta)$, we assume $\hat{\theta}_{\text{AMMLE}}^{(n,L)}$ also exists
	and is unique in $\Theta.$
	
	\subsection{Uniform error control of filtered densities and likelihood
		updates\label{sec:uniform bound}}
	
	Now we establish a framework for uniform error control regarding the
	filtered density and the likelihood update derived under the truncated Fourier
	recursion in Proposition \ref{prop:implement}. The theoretical guarantee
	of these error bounds stems from the forgetting property of filter\footnote{The forgetting property of filter means that perturbations in the filtered density are
		contracted as the recursion moves forward. So, approximation errors
		generated at earlier steps do not amplify over time.} as well as the uniform convergence of Fourier series (see Online Appendix \ref{sec:uni convergence fourier}). Consequently, the
		one-step approximation errors remain uniformly controlled when
		propagated through the recursion. The precise statements of these
	bounds are presented in Theorem \ref{thm:uni converge of filter} below. Its proof is relegated to \ref{appendix:proofs}.
	
	\begin{theorem}
		\label{thm:uni converge of filter}Under Assumptions \ref{ass:S6} and
		\ref{assump: bound derivative}, the uniform error between true and
		approximated filtered density can be controlled as
		\begin{equation}
			\sup_{f:\mathcal{X}\rightarrow\mathbb{R},\,\left\Vert f\right\Vert _{\infty
				}=1}\left\vert \int_{\mathcal{X}}\left\vert p(x|\mathbf{y}_{n};\theta)-\hat
			{p}^{(L)}(x|\mathbf{y}_{n};\theta)\right\vert f(x)\,dx\right\vert \leq
			\tilde{C}_{d_{X}}\frac{\log L}{L} \label{eq:fourier error}%
		\end{equation}
		for any $\theta\in\Theta$ and $n\geq1,$ where $\tilde{C}_{d_{X}}$ is a constant
		that depends on the latent dimension.
	\end{theorem}
	
	
	A novel feature of this result is that it provides
	a deterministic, tractable bound, decisively bypassing the simulation
	noise and the local degeneracy typically associated with most Monte Carlo-based filters. As a direct
	consequence, this density-level bound ensures that all filtered moments are
	automatically well-controlled. Because the latent space is assumed to be
	bounded, the functions $f(x)=x$ or $f(x)=x^{2}$ are bounded up to a scaling
	constant. By applying these appropriately scaled test functions to the
	integral in (\ref{eq:fourier error}), the approximation errors for the
	filtered mean and variance inherit the same convergence order of $(\log L)/L.$
	Moreover, recall that the $k$th-step true and approximated likelihood update
	functions are obtained from the updating system as $\mathcal{\hat{L}}%
	_{k}^{(L)}(\theta)=\int_{\mathcal{X}}p_{Y}(y_{k}|x_{k-1},y_{k-1};\theta
	)\hat{p}^{(L)}(x_{k-1}|\mathbf{y}_{k-1};\theta)dx_{k-1},$ while the marginal
	density $p_{Y}(y_{k}|x_{k-1},y_{k-1};\theta)$ admits the upper bound as
	illustrated in Lemma~\ref{lemma: likelihood bound} in Online Appendix. The
	error bounds of likelihood updates can also be well controlled by taking
	$f(x_{k-1})=p_{Y}(y_{k}|x_{k-1},y_{k-1};\theta)$ in Theorem
	\ref{thm:uni converge of filter}.
	
	\begin{corollary}
		\label{coro:uniform error for likelihood}Under Assumptions \ref{ass:S6} and
		\ref{assump: bound derivative}, the error between $k$th-step likelihoods and
		log-likelihoods can be controlled as
		\[
		\sup_{k\geq2}\lvert\mathcal{\hat{L}}_{k}^{(L)}(\theta)-\mathcal{L}_{k}%
		(\theta)\rvert\leq C_{d_{X}}^{\mathcal{L}}\frac{\log L}{L}\quad \mbox{and} \quad \sup_{k\geq
			2}\lvert\hat{\ell}_{k}^{(L)}(\theta)-\ell_{k}(\theta)\rvert\leq C_{d_{X}%
		}^{\ell}\frac{\log L}{L}
		\]
		for any $\theta\in\Theta$, where $C_{d_{X}}^{\mathcal{L}}$ and $C_{d_{X}%
		}^{\ell}$ are two constants that depend on the latent dimension.
	\end{corollary}
	
	\subsection{Large-sample theory of approximated MMLEs}
	
	We also prove the asymptotic properties of $\hat{\theta}^{(n)}$ and
	$\hat{\theta}_{\text{AMMLE}}^{(n,L)}$ for general LMMs. To obtain consistency
	and asymptotic normality of MLEs, a routine strategy typically requires uniform
	convergence of the normalized log-likelihood, a central limit theorem for the
	score, and convergence of the normalized Hessian, which are usually
	established by representing or approximating the likelihood and its
	derivatives through stationary and ergodic one-step contributions (see, e.g.,
	\cite{vanderVaart1998}). In SSMs and HMMs, the additional
	conditional-independence structure in (\ref{eq:SSM 1})-(\ref{eq:SSM 2}) makes
	this possible (see, e.g.,
	\cite{cappe2005inference,Pouzo2022}).\footnote{For instance, filter forgetting
		typically implies that $\ell_{k}(\theta)$ can be approximated by a stationary
		counterpart driven by the infinite observation history, i.e., $\ell_{k}%
		(\theta)\approx\ell^{*}_{k}(\theta)=\log p(Y_{k}|Y_{-\infty:k-1};\theta)$,
		with $|\ell_{k}(\theta)-\ell^{*}_{k}(\theta)|\le C \rho^{k}$. $\ell^{*}%
		_{k}(\theta)$ can be written as the same one-step function of $(Y_{k}%
		,p(\cdot|Y_{-\infty:k-1};\theta))$ for every $k$. Consequently, the sample
		log-likelihood admits the stationary additive approximation $n^{-1}\sum
		^{n}_{k=1}\ell_{k}(\theta;\nu)+o_{p}(1)=n^{-1}\sum^{n}_{k=1}\ell^{*}%
		_{k}(\theta)$. The right-hand side is thus an additive functional generated by
		the same one-step function, to which standard ergodic arguments can be
		applied. Analogous approximations can be established for the score and
		Hessian.} However, general LMMs do not necessarily satisfy these additional
	restrictions. Furthermore, both $\ell_{k}(\theta)$ and $\partial\ell_{k}(\theta)/\partial
	\theta$ cannot be written directly as additive
	functionals of the original Markov process $(X_{k},Y_{k}).$ 

	To overcome those difficulties, we
	introduce a function-valued Markovian augmentation that restores the ergodic
	additive structure required by standard MLE asymptotic theory. The observations $\mathbf{Y}_{k}=\mathbf{y}_{k}$ are
	regarded as a fixed realized sample path generated under the true parameter value
	$\theta_{0}.$ When evaluating the likelihood at a candidate value
	$\theta,$ the filtering density $p(\cdot|\mathbf{y}_{k};\theta)$ is computed
	according to the Bayes updating system in Section \ref{sec:bayes updating}
	under this candidate parameter value, i.e., $p(\cdot|\mathbf{Y}_{k};\theta)=\left.
	p(\cdot|\mathbf{y}_{k};\theta)\right\vert _{\mathbf{y}_{k}=\mathbf{Y}_{k}}.$ This shows that
	$p(\cdot|\mathbf{Y}_{k};\theta)$ can be viewed as a function-valued process
	that maps vector process $\mathbf{Y}_{k}$ into a probability density space $\mathcal{S}$.
	Subsequently, the filter derivative can be defined as the vector-valued
	process $\partial p(\cdot|\mathbf{Y}_{k};\theta)/\partial\theta=\left(
	\partial p(\cdot|\mathbf{Y}_{k};\theta)/\partial\theta_{1},\ldots,\partial
	p(\cdot|\mathbf{Y}_{k};\theta)/\partial\theta_{d}\right)  ^{\intercal}$ defined on a proper space $\dot{\mathcal{S}}$. We leave the rigorous definitions of $\mathcal{S}$ and $\dot{\mathcal{S}}$ in Online Appendix \ref{eq:S_state}-\ref{eq:Sdot_state}.
	
	We establish the consistency and asymptotic normality of the MMLE by first
	embedding $\left(  X_{k},Y_{k}\right)  $ into larger processes  $\xi_{k}%
	^{+}(\theta)=(Y_{k+1},\xi_{k}(\theta)),$ where $\xi_{k}%
	(\theta):=\left(  X_{k},Y_{k},p(\cdot|\mathbf{Y}_{k};\theta),\partial
	p(\cdot|\mathbf{Y}_{k};\theta)/\partial\theta\right).$  Then, we analyze the associated
	asymptotic properties.\footnote{\cite{TadicDoucet2005} use Markov augmentation $\xi_k(\theta)$ to study the stability and long-run behavior of the filtering process and its parameter derivatives. We use a related but further enlarged Markov augmentation $\xi^{+}_{k}(\theta)$ for a different purpose: to rewrite the likelihood and its score as Markov processes so that we can develop MMLE theory and analyze the feasible approximate MMLE based on Fourier truncation.} The key role of this augmentation is that the
	likelihood increment and its score can now be represented as one-step
	functions of the enlarged state, i.e., $\ell_{k}(\theta)=\ell\left(  \xi
	_{k-1}^{+}(\theta);\theta\right)  ,$ $\partial\ell_{k}(\theta)/\partial
	\theta=S\left(  \xi_{k-1}^{+}(\theta);\theta\right)  ,$ for two one-step
	functions $\ell\left(  \cdot;\theta\right)  ,S\left(  \cdot;\theta\right)  $.
	Hence, once the augmented process $\xi_{k}^{+}(\theta)$ is shown to be ergodic
	and Markov, the path dependence of the likelihood is absorbed into the state,
	allowing standard large-sample arguments for additive functionals of ergodic
	processes to be applied. Let $\Xi=\mathcal{X}\times\mathcal{Y}\times\mathcal{S}\times\dot{\mathcal{S}}$,
	for any fixed $\theta\in\Theta$. For $\xi_{0}=(x_{0},y_{0},f_{0},g_{0})\in\Xi
	$, we denote the one-step MMLE score function as
	\[
	S(y,\xi_{0};\theta)=\frac{\int_{\mathcal{X}}\frac{\partial p_{Y}}%
		{\partial\theta}(y|\tilde{x}_{0},y_{0};\theta)f_{0}(\tilde{x}_{0})\,d\tilde
		{x}_{0}+\int_{\mathcal{X}}p_{Y}(y|\tilde{x}_{0},y_{0};\theta)g_{0}(\tilde
		{x}_{0})\,d\tilde{x}_{0}}{\int_{\mathcal{X}}p_{Y}(y|\tilde{x}_{0},y_{0}%
		;\theta)f_{0}(\tilde{x}_{0})\,d\tilde{x}_{0}},
	\]
	and the one-step MMLE Fisher-information ingredient $I^{M}(y,\xi_{0}%
	;\theta_{0})=S(y,\xi_{0};\theta_{0})S(y,\xi_{0};\theta_{0})^{\intercal}.$ As formally established in \ref{appendix:proofs}, the augmented Markov process
	$\xi_{k}^{+}(\theta)$ is ergodic and thus admits a unique stationary
	distribution. Let the random vector $\xi^{+}(\theta_{0})$ follow the stationary distribution of
	the ergodic Markov process $\xi_{k}^{+}(\theta_{0})$.
	Finally, the MMLE Fisher information is defined as the stationary expectation
	of this one-step contribution, i.e.,
	\begin{equation}
		I^{M}=\mathbb{E}[I^{M}(\xi^{+}(\theta_{0});\theta_{0})]=\mathbb{E}[I^{M}%
		(Y^{+},\xi(\theta_{0});\theta_{0})]. \label{eq:fisher inform}%
	\end{equation}
	Building upon this crucial regularity condition, we are now
	ready to state the principal asymptotic properties of both
	$\hat{\theta}_{\text{AMMLE}}^{(n,L)}$ obtained from our Fourier recursion and $\hat{\theta}^{(n)}$.
	
	\begin{theorem}
		\label{thm:mimle-long-span}Under the same conditions in Theorem \ref{thm:uni converge of filter}, we further assume Assumptions \ref{ass:S2},
		\ref{ass:S4} and \ref{ass:S5} hold. Then the MMLE $\hat{\theta}^{(n)}$ is
		consistent and asymptotically normal:
		\[
		\hat{\theta}^{(n)}\overset{p}{\rightarrow}\theta_{0},\text{ }n^{1/2}%
		(\hat{\theta}^{(n)}-\theta_{0})\overset{d}{\rightarrow}N(0,\left(
		I^{M}\right)  ^{-1}),
		\]
		as $n\rightarrow\infty,$ where $I^{M}$ is defined in (\ref{eq:fisher inform}). 
        
        If Assumption \ref{assump:hessian bound} also holds in addition, the following two assertions hold.
        
		(i) For any fixed truncation level $L,$ the error between approximated MMLE
		$\hat{\theta}_{\text{AMMLE}}^{(n,L)}$ and the true parameter $\theta_{0}$ can
		be controlled as
		\begin{equation}
			\left\Vert \hat{\theta}_{\text{AMMLE}}^{(n,L)}-\theta_{0}\right\Vert ^{2}\leq
			C_{d_{X}}\frac{\log L}{L}+r_{n},\text{ with }r_{n}=o_{p}(1),
			\label{eq:par est bound}%
		\end{equation}
		where $\left\Vert \cdot\right\Vert $ denotes the Euclidean norm, 
        and $C_{d_{X}}$ is a constant that depends on the latent dimension.
        
		(ii) If the truncation level $L_{n}$ satisfies $n\log L_{n}/L_{n}%
		\rightarrow0,$ the asymptotic normality of $\hat{\theta}_{\text{AMMLE}%
		}^{(n,L_{n})}$ can be established as
		\begin{equation}
			n^{1/2}(\hat{\theta}_{\text{AMMLE}}^{(n,L_{n})}-\theta_{0}%
			)\overset{d}{\rightarrow}N(0,\left(  I^{M}\right)  ^{-1}),\quad
			n\rightarrow\infty. \label{eq:par est clt}%
		\end{equation}
		
	\end{theorem}
	
		We provide a unified approach to prove the asymptotic behaviors of
		(\ref{eq:par est bound})-(\ref{eq:par est clt}), see the detailed technical arguments
		 in \ref{appendix:proofs}.
	The proof proceeds by introducing a function-valued Markovian augmentation.
	Namely, by embedding the filtering density into an augmented state space
	$\xi_{k}(\theta)$, we are able to restore the ergodicity necessary for
	complex, path-dependent likelihood evaluation. Furthermore, by integrating the
	deterministic uniform Fourier error bound $O(\log L/L)$ shown in Theorem
	\ref{thm:uni converge of filter}, the proof avoids the intractable stochastic
	bounds typically associated with the simulation methods. This powerful synthesis of Fourier
	analysis and large sample theory ensures the asymptotic efficiency of the
	truncated estimator.
	
	\subsection{Consistent estimation of MMLE Fisher information}
	
	Obtaining $I^{M}$ analytically is difficult even for the simplest
	linear-Gaussian SSM as the expectation in (\ref{eq:fisher inform}) is taken
	over the functional space $\Xi$. We further develop a consistent estimation
	method for $I^{M}$ based on the theoretical results obtained above. From the proof
	of Theorem \ref{thm:mimle-long-span}, the MMLE Fisher information can be
	consistently estimated by the negative normalized Hessian of the sample
	log-likelihood. Specifically, the proof of Lemma \ref{lemma:mimle 1}
	establishes that, under the regularity conditions of the theorem, the
	following uniform convergence in probability holds:
	\[
	\sup_{\theta\in\Theta}\left\Vert \frac{1}{n}\frac{\partial^{2}}{\partial
		\theta\partial\theta^{\intercal}}\sum_{k=0}^{n-1}\log p(Y_{k+1}|\mathbf{Y}%
	_{k};\theta)-H_{0}(\theta)\right\Vert \overset{p}{\rightarrow}0,
	\]
	where $-H_{0}(\theta_{0})=I^{M}.$ Since the MMLE is consistent, namely,
	$\hat{\theta}^{(n)}\overset{p}{\rightarrow}\theta_{0}$, and $H_{0}(\theta)$ is
	continuous at $\theta_{0}$, the above uniform convergence implies that
	\begin{equation}
		\widehat{I}_{n}^{M}\left(  \hat{\theta}^{(n)}\right)  :=-\frac{1}{n}\left.
		\frac{\partial^{2}}{\partial\theta\partial\theta^{\intercal}}\sum_{k=0}%
		^{n-1}\log p(Y_{k+1}|\mathbf{Y}_{k};\theta)\right\vert _{\theta=\hat{\theta
			}^{(n)}}\overset{p}{\rightarrow}I^{M}. \label{eq:fisher true}%
	\end{equation}
	However, directly obtaining $\widehat{I}_{n}^{M}$ in analytical form remains
	difficult. Now, the proposed Fourier recursion becomes useful, as it provides
	a practical computational route for evaluating $\widehat{I}_{n}^{M}$. Namely,
	when an observed sample path $\mathbf{y}_{n}=(y_{0},\ldots,y_{n})$ is
	available, the exact one-step likelihood $p(y_{k+1}\mid\mathbf{y}_{k};\theta)$
	can be replaced by its truncated approximation $\hat{p}^{(L)}(y_{k+1}%
	\mid\mathbf{y}_{k};\theta)$ obtained by our Fourier recursion. 
    This
	leads to the computable approximation
	\begin{equation}
		\widehat{I}_{n,L}^{M}\left(  \hat{\theta}_{\text{AMMLE}}^{(n,L)}\right)
		=-\frac{1}{n}\left.  \frac{\partial^{2}}{\partial\theta\partial\theta
			^{\intercal}}\sum_{k=0}^{n-1}\log\hat{p}^{(L)}(y_{k+1}|\mathbf{y}_{k}%
		;\theta)\right\vert _{\theta=\hat{\theta}_{\text{AMMLE}}^{(n,L)}}.
		\label{eq:fisher fourier}%
	\end{equation}
	Since $\widehat{I}_{n}^{M}\left(  \theta\right)  =-n^{-1}\nabla_{\theta}%
	^{2}S_{n}^{\ell}(\theta)$, Proposition \ref{prop:hessian updating} can be
	utilized to obtain $\widehat{I}_{n,L}^{M}\left(  \theta\right)  $ using the
	truncation to (\ref{eq:gradient fourier coeff}%
	)-(\ref{eq:hessian fourier coeff}). This avoids the need for second-order
	numerical differentiation.
	
	Under additional regularity conditions, the proof strategy underlying Theorem
	\ref{thm:uni converge of filter} can be extended to the first- and
	second-order parameter derivatives of the filtering recursion leading to
	similar asymptotic results.\footnote{
		Although $\partial^{2} p(x_{k}%
		|\mathbf{Y}_{k};\theta)/\partial{\theta^2}$ and $\partial^2 \hat{p}^{(L)}(x_{k}%
		|\mathbf{Y}_{k};\theta)/\partial{\theta^2}$ are
		signed measures with zero total mass, the filtering recursion still
		exponentially attenuates the propagation of previously accumulated
			approximation errors. One can establish $\sup_{k\geq0}\sup_{\theta\in\Theta}%
		\Vert\partial^2 p(x_{k}%
		|\mathbf{Y}_{k};\theta)/\partial{\theta^2}-\partial^2 \hat{p}^{(L)}(x_{k}%
		|\mathbf{Y}_{k};\theta)/\partial{\theta^2}\Vert\rightarrow0$ as $L\rightarrow\infty$, which
		consequently implies $\sup_{n\geq1}\sup_{\theta\in\Theta}\Vert\widehat{I}%
		_{n,L}^{M}(\theta)-\widehat{I}_{n}^{M}(\theta)\Vert\rightarrow0$ as
		$L\rightarrow\infty$. Combining the preceding uniform approximation result
		with the consistency established in Theorem \ref{thm:mimle-long-span} gives
		$\widehat{I}_{n,L_{n}}^{M}\bigl(\hat{\theta}_{\text{AMMLE}}^{(n,L_{n}%
			)}\bigr)\overset{p}{\longrightarrow}I^{M}$, for an appropriate sequence
		$L=L_{n}\rightarrow\infty$. The detailed derivative-level calculations are
		omitted for brevity.} We abbreviate $\widehat{I}_{n,L}^{M}\bigl(\hat{\theta
	}_{\text{AMMLE}}^{(n,L)}\bigr)$ as $\widehat{I}_{n,L}^{M}$. The asymptotic
	normality of the MMLE further allows standard likelihood-based
	inference. For each component $\theta_{j}\in\theta$, the asymptotic variance
	and the corresponding standard error can be estimated by $\widehat{\text{SE}}%
	(\hat{\theta}_{\text{AMMLE},j}^{(n,L)})=n^{-1/2}[(\hat{I}_{n,L}^{M}%
	)^{-1}]_{jj}^{1/2}.$ Hence an approximate $95\%$ confidence interval for $\theta_{j}$ is given by
	$\hat{\theta}_{\text{AMMLE},j}^{(n,L)}\pm1.96\,\widehat{\text{SE}}(\hat
	{\theta}_{\text{AMMLE},j}^{(n,L)})$.
	\subsection{Estimation and inference under model misspecification}
	Model misspecification is pervasive in statistical modeling, as parametric specifications rarely coincide exactly with the true data-generating process. What happens if the maintained model is misspecified? Can $\hat{\theta}^{(n)}$ and $\hat{\theta}^{(n)}_{\text{AMMLE}}$ still converge even when the true
	transition law falls outside the maintained parametric family? To address
	these questions, it is necessary for us to discuss the model misspecification.
	Let the maintained family of LMM transition densities be $\mathcal{P=}%
	\left\{  p_{(X,Y)}(x_{1},y_{1}|x_{0},y_{0};\theta):\theta\in\Theta\right\}  .$
	We allow the data to be generated by a stationary Markov process whose
	transition density is $p_{(X,Y)}^{\ast}(x_{1},y_{1}|x_{0},y_{0}).$ However,
	$p_{(X,Y)}^{\ast}$ is not necessarily contained in $\mathcal{P}$. Accordingly,
	correct specification $H_{0}$ and the model misspecification $H_{1}$
	correspond to
	\begin{align*}
		H_{0}: \exists\theta_{0}\in\Theta, \quad p_{(X,Y)}^{\ast}(\cdot)=p_{(X,Y)}(\cdot;\theta_{0}); \quad H_{1}:p_{(X,Y)}^{\ast}\not \in \mathcal{P}.
	\end{align*}
	When $\left(  X_{k},Y_{k}\right)  $ evovles under $p_{(X,Y)}^{\ast},$
	additional notation is useful to distinguish the true data-generating process
	from the maintained LMM. We define $\xi_{k}^{+,\ast}(\theta)=(Y_{k+1},\xi
	_{k}^{\ast}(\theta)),$ where $\xi_{k}^{\ast}(\theta):=\left(  X_{k}%
	,Y_{k},p(\cdot|\mathbf{Y}_{k};\theta),\partial p(\cdot|\mathbf{Y}_{k}%
	;\theta)/\partial\theta\right)  .$ Under the misspecified counterparts of the
	stability conditions in Lemma \ref{lem:filter-ergodic}, $\xi_{k}^{+,\ast}(\theta)$ admits a
	stationary distribution $\xi^{+,\ast}(\theta)$ for each $\theta.$
	
	To answer the asymptotic behavior of $\hat{\theta}^{(n)}$ and $\hat{\theta}^{(n)}_{\text{AMMLE}}$ under misspecification, we define the population marginal log-likelihood
	under the true law as $\ell^{\ast}(\theta)=\mathbb{E}\left[  \ell\left(
	\xi^{+,\ast}(\theta);\theta\right)  \right]  ,$ where $\ell(  \xi
	_{k}^{+,\ast}(\theta);\theta)  =\log p(Y_{k+1}|\mathbf{Y}_{k};\theta),$ and the one-step score contribution by $S(  \xi_{k}%
	^{+,\ast}(\theta);\theta)  =\nabla_{\theta}\log p(Y_{k+1}|\mathbf{Y}%
	_{k};\theta).$ The pseudo-true parameter set is then defined as $\Theta^{\ast}=\arg
	\max_{\theta\in\Theta}\ell^{\ast}(\theta).$
	In misspecification, $\Theta^{\ast}$ provides the natural replacement.\footnote{This definition nests the
		correctly specified case. Under $H_{0}$ and the identification condition, $\Theta^{\ast
		}=\left\{  \theta_{0}\right\}  .$ Under misspecification, by contrast, such
		$\theta_{0}$ need not exist within the maintained family. } We focus on the
	point-identified case in which $\Theta^{\ast
	}$ is a singleton,
	i.e., $\Theta^{\ast}=\left\{  \theta^{\ast}\right\}  \subset \mathrm{int}(\Theta),$
	following the convention of \cite{White1982}. The singleton restriction is not
	essential in principle; more general cases can also be considered, but are
	beyond the scope of this paper. The following proposition shows the asymptotic
	behaviour of $\hat{\theta}^{(n)}$ when the LMM is misspecified. Its proof is given in \ref{appendix:proofs}.
	
	\begin{proposition}
		\label{prop:misspecified_mmle}Suppose the regularity conditions in Theorem \ref{thm:mimle-long-span} continue to hold under the true law $ p^{\ast}_{(X,Y)}$ with $\theta_0$ replaced by the pseudo-true parameter $\theta^{\ast}$. Then, it holds that
		\[
		n^{1/2}\left(  \hat{\theta}^{(n)}-\theta^{\ast}\right)
		\overset{d}{\rightarrow}\mathcal{N}\left(  0,A_{\ast}^{-1}B_{\ast}A_{\ast
		}^{-1}\right)  ,
		\]
		where $A_{\ast}=-\nabla_{\theta}^{2}\ell^{\ast}(\theta^{\ast}),$ $B_{\ast
		}=\Gamma_{0}^{\ast}+\sum_{k=1}^{\infty}\left(  \Gamma_{k}^{\ast}+(\Gamma
		_{k}^{\ast})^{\intercal}\right)  ,$ and $\Gamma_{k}^{\ast}=\mathrm{Cov}%
		(  S(  \xi_{0}^{+,\ast}(\theta^{\ast});\theta^{\ast})
		,S(  \xi_{k}^{+,\ast}(\theta^{\ast});\theta^{\ast}))  .$
		For any fixed truncation level $L,$ the error between approximated MLE
		$\hat{\theta}_{\text{AMMLE}}^{(n,L)}$ and $\theta^{\ast}$ can
		be controlled as
		\begin{equation}
			\left\Vert \hat{\theta}_{\text{AMMLE}}^{(n,L)}-\theta^{\ast}\right\Vert
			^{2}\leq C_{d_{X}}\frac{\log L}{L}+r_{n}^{\ast},\quad r_{n}^{\ast}%
			=o_{p}(1).\label{eq:ammle_misspec_bound}%
		\end{equation}
		If $n\log L_{n}/L_{n}\rightarrow0,$ then $n^{1/2}(  \hat{\theta}_{\text{AMMLE}}^{(n,L)}-\theta^{\ast})
		\overset{d}{\rightarrow}\mathcal{N}\left(  0,A_{\ast}^{-1}B_{\ast}A_{\ast
		}^{-1}\right)  .\label{eq:misspec_fourier_rate}$

	\end{proposition}
	
	Under $H_{0}$, $\theta^{\ast}=\theta_{0}$ and the information equality yields
	$A_{\ast}=B_{\ast}=I^{M}.$ Thus Proposition \ref{prop:misspecified_mmle} reduces to
	Theorem \ref{thm:mimle-long-span} under correct specification. Under misspecification,
	$A_{\ast}$ and $B_{\ast}$ generally differ. Nevertheless, the Fourier recursion provides feasible estimators of the corresponding sandwich covariance matrix, and hence supports misspecification-robust inference.\footnote{Let $\hat{A}_{n,L}=$ $\widehat{I}_{n,L}^{M},$ and $\hat{\Gamma}_{j,n,L}%
		=n^{-1}\sum_{k=j}^{n-1}\hat{s}_{k,L}\hat{s}_{k-j,L}^{\intercal}$ where
		$\hat{s}_{k,L}=\nabla_{\theta}\log\mathcal{\hat{L}}_{k+1}^{(L)}(\theta
		)|_{\theta=\hat{\theta}_{\mathrm{AMMLE}}^{(n,L)}}.$ A HAC estimator of
		$B_{\ast}$ is $\hat{B}_{n,L}=\hat{\Gamma}_{0,n,L}+\sum_{j=1}^{b_{n}}%
		\omega\left(  j/b_{n}\right)  (\hat{\Gamma}_{j,n,L}+\hat{\Gamma}%
		_{j,n,L}^{\intercal}),$ where $\omega(\cdot)$ is a HAC kernel function and
		$b_{n}$ is the bandwidth (see, e.g., \cite{Andrews1991}) . The feasible misspecification-robust covariance
		estimator is $\hat{V}_{n,L}=\hat{A}_{n,L}^{-1}\hat{B}_{n,L}\hat{A}_{n,L}%
		^{-1}.$ Under the same argument as in Section 4.3 and standard HAC regularity
		conditions, these estimators are consistent for $A_{\ast},$ $B_{\ast}$ and
		$A_{\ast}^{-1}B_{\ast}A_{\ast}^{-1}$, respectively.}

	\section{Simulation\label{sec:simulation}}
	
	\subsection{Validation in the linear-Gaussian benchmark}
	
We first evaluate the accuracy and robustness of the Fourier recursion in a linear-Gaussian SSM, for which the Kalman filter provides an exact benchmark. In
	particular, we consider the following linear-Gaussian LMM:%
	\begin{equation}
		Y_{k}=X_{k}+U_{k},\quad X_{k}=\alpha+\beta X_{k-1}+\sigma V_{k},
		\label{eq:linear 2}%
	\end{equation}
	where $U_{k}$ and $V_{k}$ are two independent standard normal random variables.
	
	We simulate a time series of $(X_{k},Y_{k})$ consisting of $n=4000$
	consecutive values from (\ref{eq:linear 2}) by setting $\left(  \alpha
	,\beta,\sigma\right)  =\left(  -0.73,0.9,0.363\right)  ^{\intercal}$. We
	initialize each path at $Y_{0}=-7.3$ and sample the first latent variable
	$X_{0}$ from the stationary distribution of $X_{k}$, i.e., the Gaussian
	distribution $\mathcal{N}\left(  \alpha/(1-\beta),\sigma^{2}/(1-\beta
	^{2})\right)  .$ Accordingly, we choose the initial filtered density
	$p(x_{0}|\mathbf{y}_{0};\theta)$ as the Gaussian density, i.e., $p(x_{0}%
	|\mathbf{y}_{0};\theta)=\exp\left(  -\frac{(x_{0}-\alpha/(1-\beta))^{2}%
	}{2\sigma^{2}/\left(  1-\beta^{2}\right)  }\right)  /\left(  2\pi\left(
	\sigma^{2}/\left(  1-\beta^{2}\right)  \right)  \right)  ^{1/2}$. Under the
	true parameter setting, we compare the likelihood update, one-step
	prediction density, filtered density, the smoothed density and the MMLEs obtained from the truncated Fourier recursion in Proposition
	\ref{prop:implement} with their exact counterparts.
	
	In Figure \ref{fig:fourier_kalman_comparison}, we report the relative errors
	between the likelihood update, filtered mean, predicted mean and the smoothed
	mean obtained using the Fourier updating system and their corresponding true
	values. We only illustrate the first $500$ observations for simplicity. It can
	be observed that the relative errors become sufficiently small starting from
	$L=16$, and are virtually negligible. Table \ref{table:DSV} further shows that, for $n=1000,2000,4000$, the discrepancies between the MMLEs obtained from Fourier recursion and the
	true values are numerically negligible. Figure
	\ref{fig:lg_error_cost_tradeoff} reports the trade-off between computational
	cost and accuracy for log marginal likelihood evaluation. The Fourier recursion has
	computational complexity $O(L^{2})$ according to Section \ref{sec:truncation}
	and an error that decays as $O(\log L/L)$ from Corollary
	\ref{coro:uniform error for likelihood}, implying a theoretical accuracy-cost
	curve decreasing with the nominal cost to the power around $-1/2$. For
	particle filters, standard results yield $\sqrt{N}$ convergence with cost
	$O(N)$. \footnote{The uniform error of particle filter can be formulated as
		$\sup_{\left\Vert f\right\Vert _{\infty}=1}| \int_{\mathcal{X} }%
		p(x|\mathbf{y}_{n};\theta)f(x)dx-\sum_{j=1}^{N}\omega_{n}^{(j)}\left(
		\theta\right)  f(x_{n}^{(j)})| =O_{p}( N^{-1/2}) , $ where $\omega_{n}%
		^{(j)}\left(  \theta\right)  $ denotes the weights for the $j$th particle
		$x_{n}^{(j)}$ (see, e.g., \cite{gland2004}).} This implies that the particle
	filter will produce a similar $-1/2$ slope in the accuracy-cost trade-off. Therefore, the
	particle-filter bound is stochastic while the Fourier recursion bound in
	(\ref{eq:fourier error}) is deterministic. This distinction explains the
	pattern in Figure \ref{fig:lg_error_cost_tradeoff}: the Fourier recursion
	closely tracks the theoretical curve, whereas the particle filter displays
	substantial fluctuations around it.
	
	\begin{figure}[tbh]
		\centering
		\par%
		\begin{tabular}
			[c]{cc}%
			\includegraphics[width=0.5\textwidth]{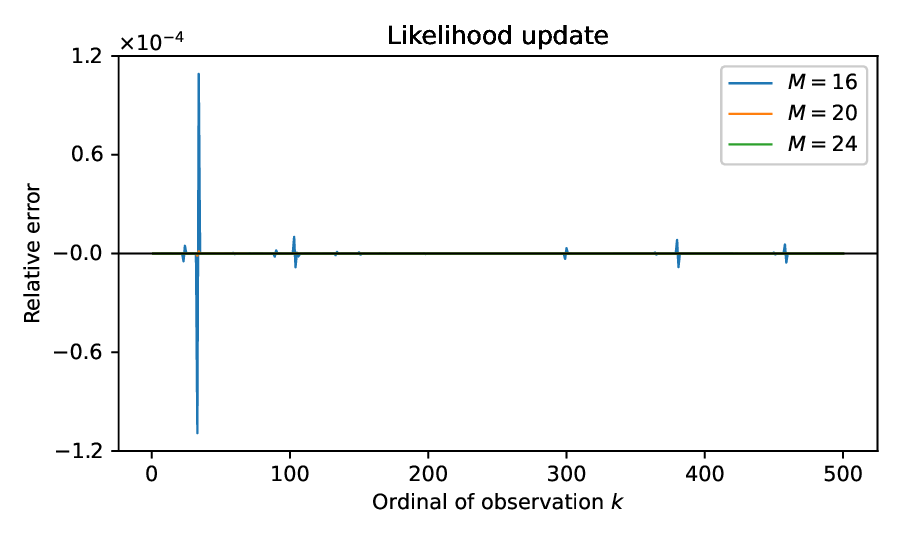} &
			\includegraphics[width=0.5\textwidth]{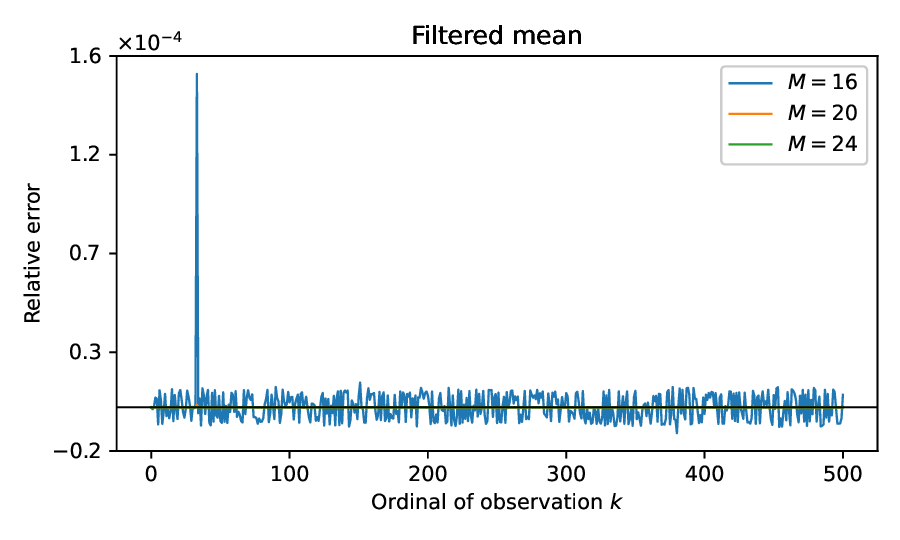}\\[0.35cm]%
			\includegraphics[width=0.5\textwidth]{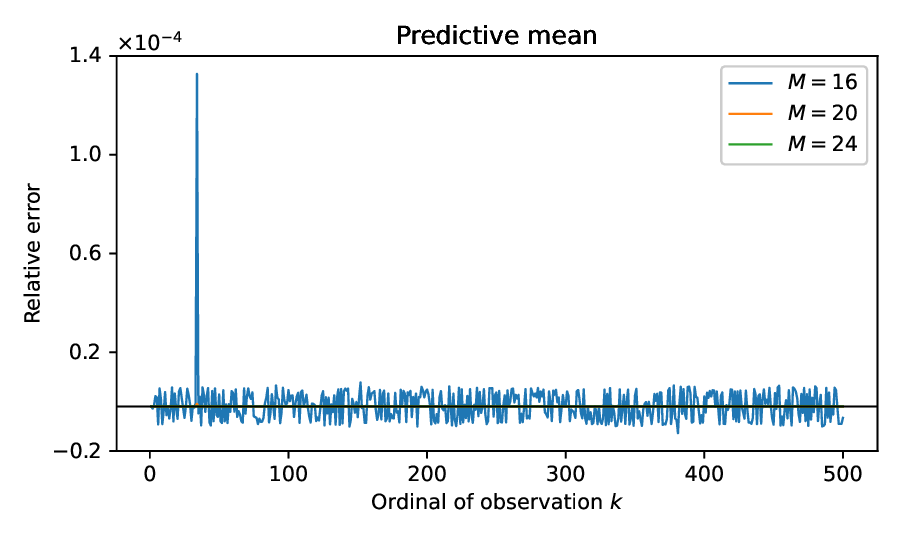} &
			\includegraphics[width=0.5\textwidth]{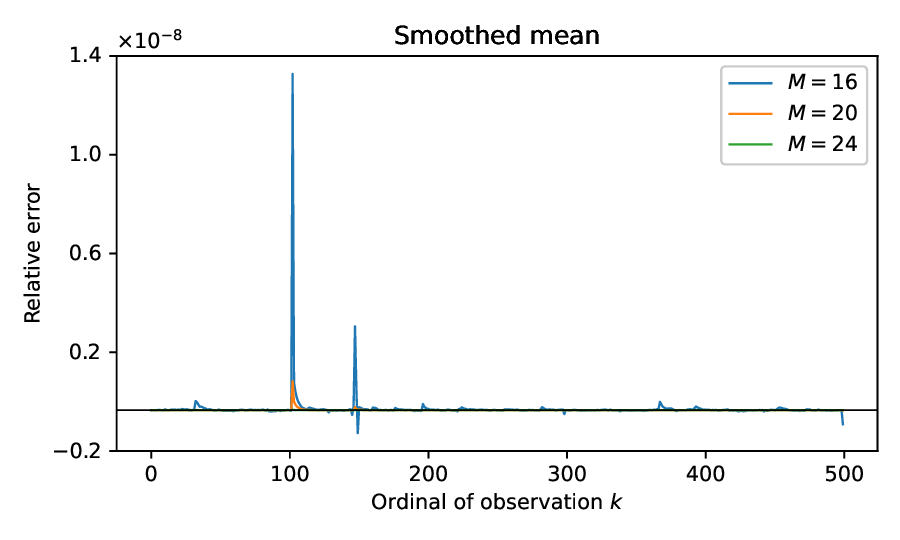}
		\end{tabular}
		\caption{{\protect\small Comparison between the Fourier updating system and
				the associated true value under the linear-Gaussian state-space model
				(\ref{eq:linear 2}). For different truncation levels, $b_{0}=-15$, and
				$b_{1}=-1$}}%
		\label{fig:fourier_kalman_comparison}%
	\end{figure}
	
	\begin{table}[th]
		\centering
		\begin{threeparttable}
			
			\caption{Comparison of MMLE results for the linear-Gaussian state space model.}
			\label{table:DSV}
			
			{\small
				\setlength{\tabcolsep}{0pt}
				\renewcommand{\arraystretch}{1.12}
				\begin{tabular*}{\textwidth}{@{\extracolsep{\fill}}lccccccc@{}}
					\hline
					Parameter & True
					& \multicolumn{2}{c}{1000 obs}
					& \multicolumn{2}{c}{2000 obs}
					& \multicolumn{2}{c}{4000 obs} \\
					&
					& Bias & Std. dev.
					& Bias & Std. dev.
					& Bias & Std. dev. \\
					\hline
					\multicolumn{8}{@{}l}{Panel A: $\hat{\theta}^{(n)}_{\mathrm{AMMLE}}$ obtained from Proposition~\ref{prop:implement}}\\
					\hline
					$\alpha$ & $-0.73$
					& $-0.132$ & $0.360$
					& $-0.046$ & $0.224$
					& $-0.023$ & $0.158$ \\
					$\beta$  & $0.9$
					& $-0.018$ & $0.049$
					& $-0.006$ & $0.031$
					& $-0.003$ & $0.022$ \\
					$\sigma$ & $0.363$
					& $0.017$ & $0.096$
					& $0.004$ & $0.067$
					& $0.003$ & $0.049$ \\
					\hline
					\multicolumn{8}{@{}l}{Panel B: $\hat{\theta}^{(n)}_{\mathrm{MMLE}}$ obtained from Kalman filter}\\
					\hline
					$\alpha$ & $-0.73$
					& $-0.131$ & $0.360$
					& $-0.046$ & $0.225$
					& $-0.023$ & $0.157$ \\
					$\beta$  & $0.9$
					& $-0.018$ & $0.049$
					& $-0.006$ & $0.030$
					& $-0.003$ & $0.022$ \\
					$\sigma$ & $0.363$
					& $0.016$ & $0.096$
					& $0.004$ & $0.067$
					& $0.003$ & $0.049$ \\
					\hline
				\end{tabular*}
			}
			
			\begin{tablenotes}[flushleft]
				\small
				\item \textit{Notes:} ``Std. dev.'' denotes the finite-sample standard deviation.
				In both panels, each bias and finite-sample standard deviation is computed based on 500 estimators.
				The Fourier recursion is implemented using Proposition~\ref{prop:implement}, with the truncation level set to $L=24$, $b_0=-15$, and $b_1=-1$.
			\end{tablenotes}
			
		\end{threeparttable}
	\end{table}
	
	\begin{figure}[ptb]
		\centering
		\par
		\begin{minipage}{0.48\textwidth}
			\centering
			\includegraphics[width=\textwidth]{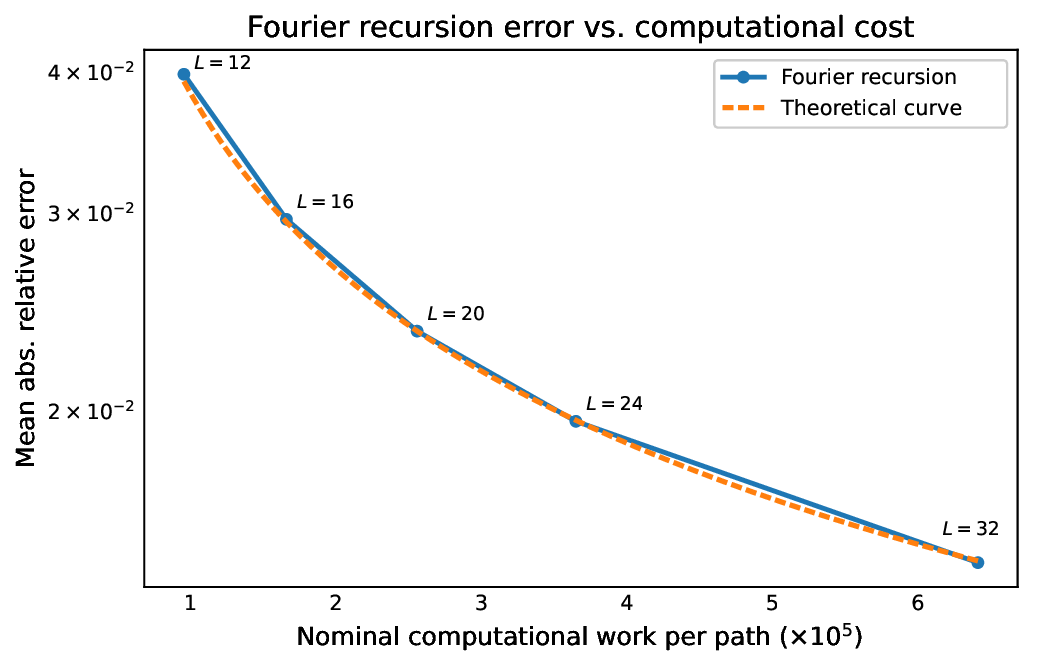}
		\end{minipage}
		\hfill\begin{minipage}{0.48\textwidth}
			\centering
			\includegraphics[width=\textwidth]{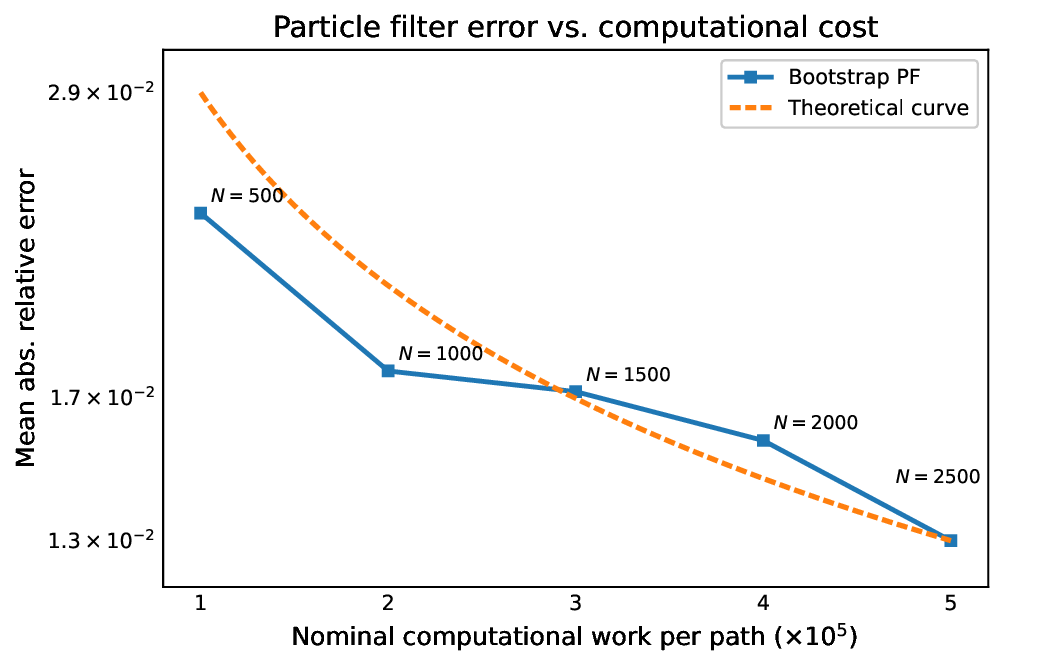}
		\end{minipage}
		\caption{{\protect\small Error-cost trade-off under the linear-Gaussian SSM.
				The vertical axis reports the mean absolute relative error of the log marginal
				likelihood, averaged over 200 simulated paths of length 50, relative to the
				exact Kalman filter. The horizontal axis reports the nominal
				computational work per path. Dashed curves show the corresponding theoretical
				reference rates.}}%
		\label{fig:lg_error_cost_tradeoff}%
	\end{figure}

	We also examine the estimation of the Fisher information matrix under the
	linear-Gaussian specification. Since the Kalman filter yields the exact
	likelihood in this setting, the observed Fisher information defined in
	(\ref{eq:fisher true}) evaluated at the true parameter value is used as the
	benchmark. We compare it with the Fourier-recursion observed Fisher
	information in (\ref{eq:fisher fourier}) evaluated at
	$\hat{\theta}_{\text{AMMLE}}^{(n,L)}$ with $(L,n)=(36,4000)$. The relative
	Frobenius error is $||\widehat{I}_{n,L}^{M}-I^{M}||_{F}\big/\left\Vert
	I^{M}\right\Vert _{F}=1.22\%,$ by taking the Kalman information matrix as the
	true benchmark.
	
	\subsection{Evaluation in nonlinear and non-Gaussian LMMs}
	
	\textit{Non-Gaussian cases.} We replace $U_{k}$ in (\ref{eq:linear 2}) with a non-Gaussian innovation,
	i.e., $U_{k}=\log\varepsilon_{k}^{2},$ where $\varepsilon_{k}$ follows a
	standard normal distribution, while keeping (\ref{eq:linear 2}) the same. The
	model is therefore immediately identical to (\ref{eq:discrete sv}) if we
	denote $Y_{k}=\log R_{k}^{2},$ $X_{k}=\log\sigma_{k}^{2}$. This situation
	illustrates how the Fourier recursion behaves under the non-Gaussian LMMs. We
	still evaluate the performance of the method by comparing the recovered
	filtering distributions, likelihood values, and reporting the AMMLE results.
	For simplicity, we still use $\left(  \alpha,\beta,\sigma\right)  =\left(
	-0.73,0.9,0.363\right)  ^{\intercal}$. Because this model nests within the
	framework of both SSMs and LMMs with explicit conditional characteristic
	function, we can proceed with our analysis by applying either Corollary
	\ref{thm:fourier series update 3} or Corollary
	\ref{thm:fourier series update 2}.
	
	We then compare the MMLE performance of the Fourier recursion, the auxiliary particle filter (APF), and the
	Kalman filter. As can be seen from Table \ref{table:logsv_mle_500}, the APF
	method produces a relatively large bias, while its finite-sample standard
	deviation remains small. This phenomenon is mainly due to the nonsmoothness of
	the simulated likelihood. When the initial value is not sufficiently close to
	the true parameter, the optimizer may become trapped near the starting point
	and fail to move effectively toward the true value.
	
	\begin{table}[tbh]
		\centering
		\begin{threeparttable}
			\caption{Comparison of MLEs for the discrete-time stochastic volatility model.}
			\label{table:logsv_mle_500}
			
			{\small
				\setlength{\tabcolsep}{0pt}
				\renewcommand{\arraystretch}{1.15}
				\begin{tabular*}{\textwidth}{@{\extracolsep{\fill}}lccccccc@{}}
					\hline
					Parameter & True
					& \multicolumn{2}{c}{APF}
					& \multicolumn{2}{c}{Kalman}
					& \multicolumn{2}{c}{Fourier recursion} \\
					&
					& Bias & Std. dev.
					& Bias & Std. dev.
					& Bias & Std. dev. \\
					\hline
					$\alpha$
					& $-0.730$
					& $-0.304$ & $0.388$
					& $-0.412$ & $0.702$
					& $-0.175$ & $0.388$ \\
					
					$\beta$
					& $0.900$
					& $-0.149$ & $0.048$
					& $-0.056$ & $0.095$
					& $-0.024$ & $0.053$ \\
					
					$\sigma$
					& $0.363$
					& $0.008$ & $0.094$
					& $0.072$ & $0.167$
					& $0.013$ & $0.080$ \\
					\hline
				\end{tabular*}
			}
			
			\begin{tablenotes}[flushleft]
				\small
				\item \textit{Notes:} The sample size is $n=500$ for each of the $500$ simulated paths.
				For each estimator, the table reports the bias and finite-sample standard deviation.
				The Fourier recursion is implemented with truncation levels $L=24$, $a_0=-2$, $a_1=2$, $b_0=-12$, and $b_1=-2$.
				The APF implementation uses 4000 particles.
			\end{tablenotes}
			
		\end{threeparttable}
	\end{table}
	
	\textit{Nonlinear cases.} We now consider the SRNN characterization
	introduced in Example \ref{example:SRNN}. In principle, SRNNs can be trained by maximizing the marginal likelihood. Since such training requires accurate likelihood evaluation, the Fourier recursion provides a practical computational tool for likelihood-based training of SRNN-type models. For illustrative purposes, we restrict to the case
	in which both $X_{k}$ and $Y_{k}$ are
	one-dimensional. To retain a benchmark for comparison, the observed data are
	still generated from the linear-Gaussian SSM in (\ref{eq:linear 2}) under the
	same parameter setting, so that the Kalman filter provides a benchmark. We then fit a nonlinear SRNN model to this
	observed time series of length $n=500$. Specifically, the transition
	and observation functions are parameterized by single-hidden-layer feedforward
	neural networks with four tanh linear unit hidden units:
	\begin{equation}
		f_{\theta}(x)=a_{f,0}+\sum_{j=1}^{4}a_{f,j}\,\varphi\!\left(  b_{f,j}%
		x+c_{f,j}\right)  ,\quad g_{\theta}(x)=a_{g,0}+\sum_{j=1}^{4}a_{g,j}%
		\,\varphi\!\left(  b_{g,j}x+c_{g,j}\right)  , \label{eq:srnn simulation}%
	\end{equation}
	where $\varphi(z)=\tanh(z)$. Although the data-generating process remains the
	linear-Gaussian SSM, the fitted model is deliberately specified as a nonlinear
	SRNN. This design allows us to assess whether the proposed Fourier recursion
	can support likelihood-based estimation of a flexible SRNN specification in a
	setting where a Kalman-filter benchmark is available. Both the state
	transition network and the observation network are one-hidden-layer scalar
	multilayer perceptrons (MLPs), which are highly nonlinear. We set $\eta
	_{k}\sim\mathcal{N}\left(  0,\sigma_{\eta}^{2}\right)  $, $\varepsilon_{k}%
	\sim\mathcal{N}\left(  0,\sigma_{\varepsilon}^{2}\right)  $. Thus, the tuning
	parameters $\theta=\{a_{f,j},b_{f,j},c_{f,j},a_{g,j},b_{g,j},c_{g,j}%
	,\sigma_{\xi},\sigma_{\eta};j=0,1,2,3,4\}$ are estimated by maximizing the
	marginal likelihood of SRNN characterized by (\ref{eq:srnn simulation}). The marginal likelihood is
	evaluated using the Fourier recursion in Proposition
	\ref{prop:implement}, and the resulting objective function is
	optimized with the Adam algorithm (see, e.g., \cite{kingma2015adam}). \begin{figure}[ptbh]
		\centering
		\includegraphics[width=0.49\textwidth]{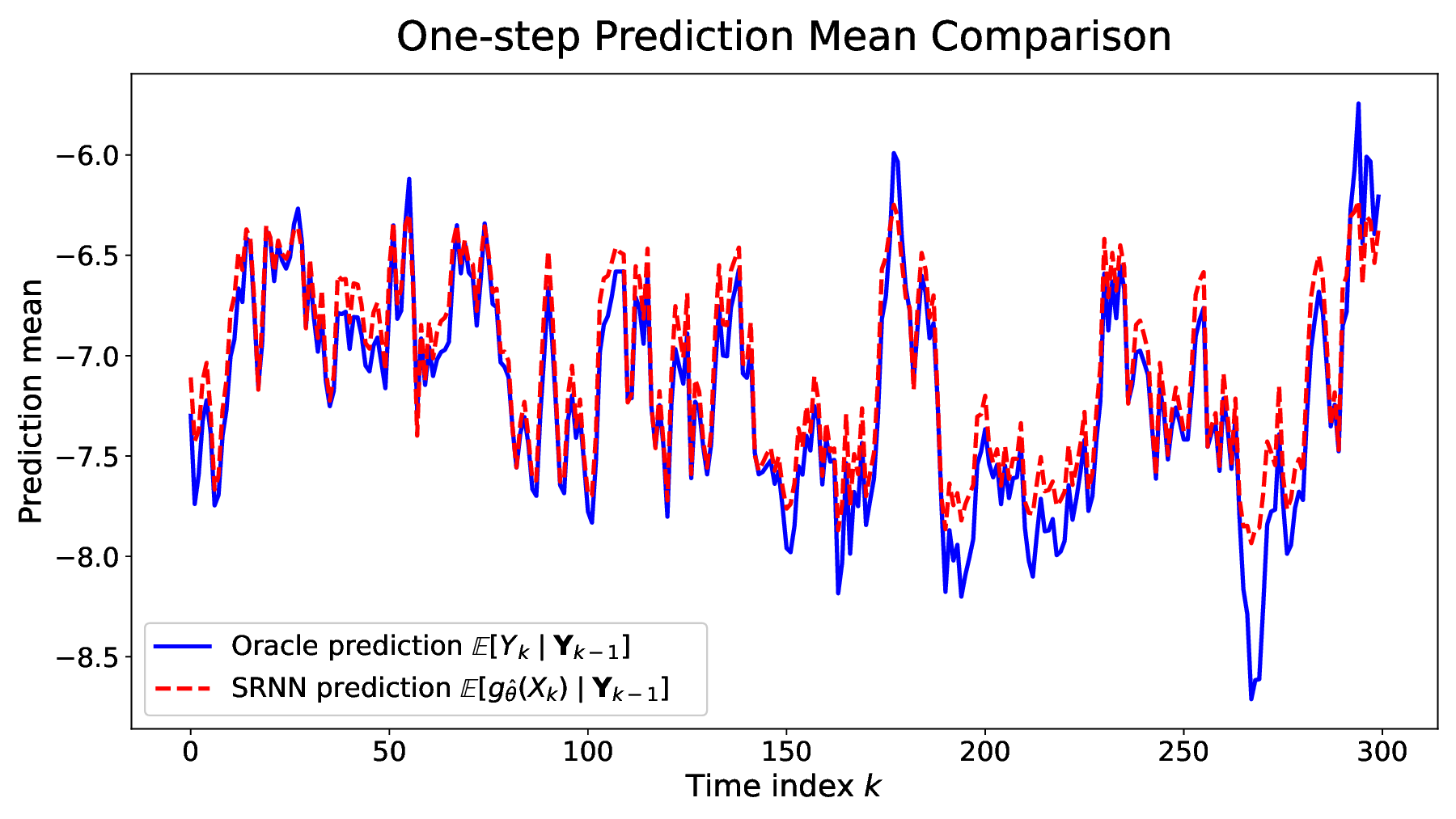}
		\hfill
		\includegraphics[width=0.49\textwidth]{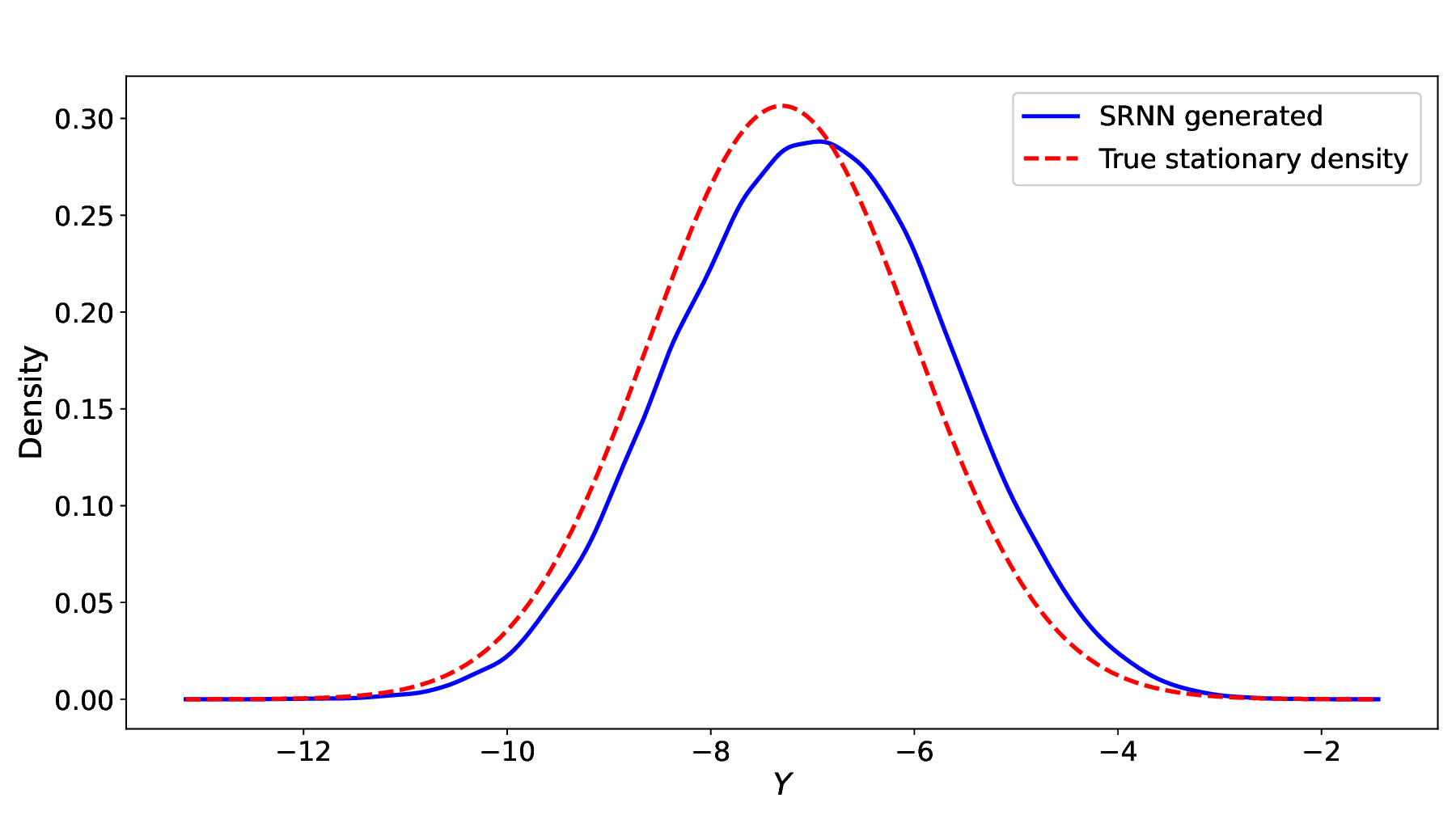}\caption{{\protect\small Comparison
				between the fitted SRNN and the oracle benchmark. The left panel reports the
				prediction trajectory of the trained SRNN together with the oracle prediction
				given by the Kalman filter under the true model. The right panel compares the
				marginal distribution at $k=500$ obtained from a large number of
				SRNN-generated sample paths with the true marginal distribution of the
				underlying data-generating model.}}%
		\label{fig:two_results}%
	\end{figure}
	
	A central criterion for assessing whether the SRNN has been effectively
	trained is its one-step-ahead predictive performance. In the SRNN formulation,
	this corresponds to comparing the fitted predictive mean $\mathbb{E}%
	[g_{\hat{\theta}}(X_{k})|\mathbf{Y}_{k-1}]$ with the oracle
	prediction $\mathbb{E}\left[  Y_{k}|\mathbf{Y}_{k-1}\right]  ,$ which is
	available from the Kalman filter. The left panel of Figure
	\ref{fig:two_results} reports this comparison and shows that the trained SRNN
	closely tracks the Kalman-filter benchmark. Beyond predictive accuracy, we
	also examine whether the trained SRNN reproduces the distributional features
	of the underlying data-generating process. To this end, we generate a large
	number of sample paths from the fitted SRNN and compare their marginal
	distribution at $k=500$ with the true marginal distribution implied by the
	data-generating model. We further report summary statistics. The true sample
	has a mean of $-7.30$, a variance of $1.70$, and an excess kurtosis of $0.02$,
	compared with $-7.00$, $1.75$, and $-0.18$, respectively, for the
	SRNN-generated sample.
	
	Next, we consider a continuous-time model, namely, the BNS model in Example
	\ref{example:bns} with Gamma subordinator, i.e., $Z_{t}-Z_{s}\sim Gamma\left(
	\nu(t-s),\gamma\right)  ,$ and $\beta=-1/2.$ Let $X_{k}=V_{k\Delta},$
	$Y_{k}=\log S_{k\Delta}/S_{(k-1)\Delta}.$ The BNS model is a typical affine
	model, namely, the conditional characteristic function can be written as
	$\phi_{(X,Y)}(u,v|x_{0},y_{0})=\exp\left(  C_{1,\Delta}(u,v)+C_{2,\Delta
	}(u,v)x_{0}\right)  $ where
	\begin{align*}
		C_{1,\Delta}(u,v)  &  =iv\mu\Delta-v\lambda\int_{0}^{\Delta}\log\left(
		1-\frac{iv\rho+iue^{-\lambda s}+\lambda^{-1}\left(  iv\beta-1/2v^{2}\right)
			\left(  1-e^{-\lambda s}\right)  }{\gamma}\right)  ds,\\
		C_{2,\Delta}(u,v)  &  =iue^{-\lambda\Delta}+\left(  iv\beta-1/2v^{2}\right)
		\frac{1-e^{-\lambda\Delta}}{\lambda}.
	\end{align*}
	The true parameter values are chosen to mimic a standard Heston volatility
	specification in \cite{bates2006maximum} through moment matching. In the
	BNS-Gamma model, the stationary latent variance satisfies $\mathbb{E}%
	(X_{t})=\nu/\gamma$ and $\mathrm{Var}(X_{t})=\nu/(2\gamma^{2})$, while its
	persistence is governed by $\mathrm{Corr}(X_{t+\Delta},X_{t})=\exp
	(-\lambda\Delta)$. We set the long-run variance level to $m=0.225$ and the
	stationary variance to $s^{2}=0.00703$. Matching $\nu/\gamma=m$ and
	$\nu/\gamma^{2}=2s^{2}$ gives $\gamma=16$ and $\nu=3.6$. The mean-reversion
	parameter is set to $\lambda=10$, and the return drift and leverage parameters
	are set to $\mu=0.1$ and $\rho=-0.2$. \begin{figure}[ptbh]
		\centering
		\vspace{-3pt} \begin{minipage}{0.8\textwidth}
			\centering
			\includegraphics[width=\textwidth]{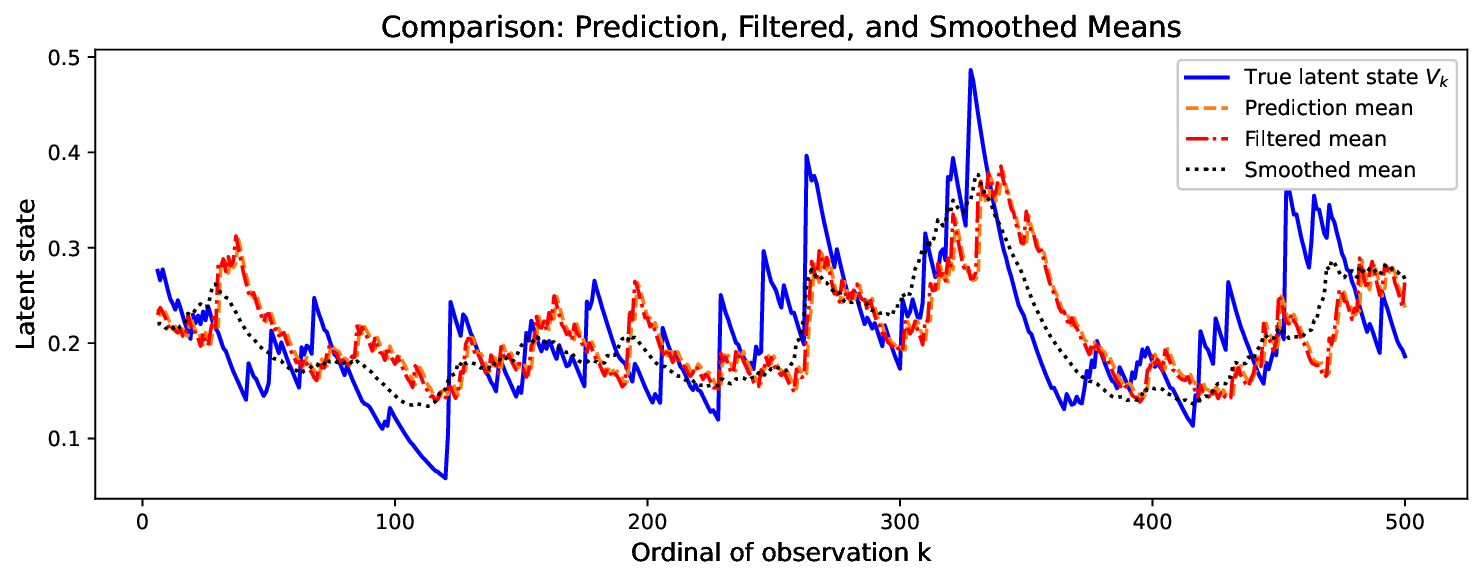}
			\vspace{-0.2cm}
		\end{minipage}
		\par
		\begin{minipage}{0.8\textwidth}
			\centering
			\includegraphics[width=\textwidth]{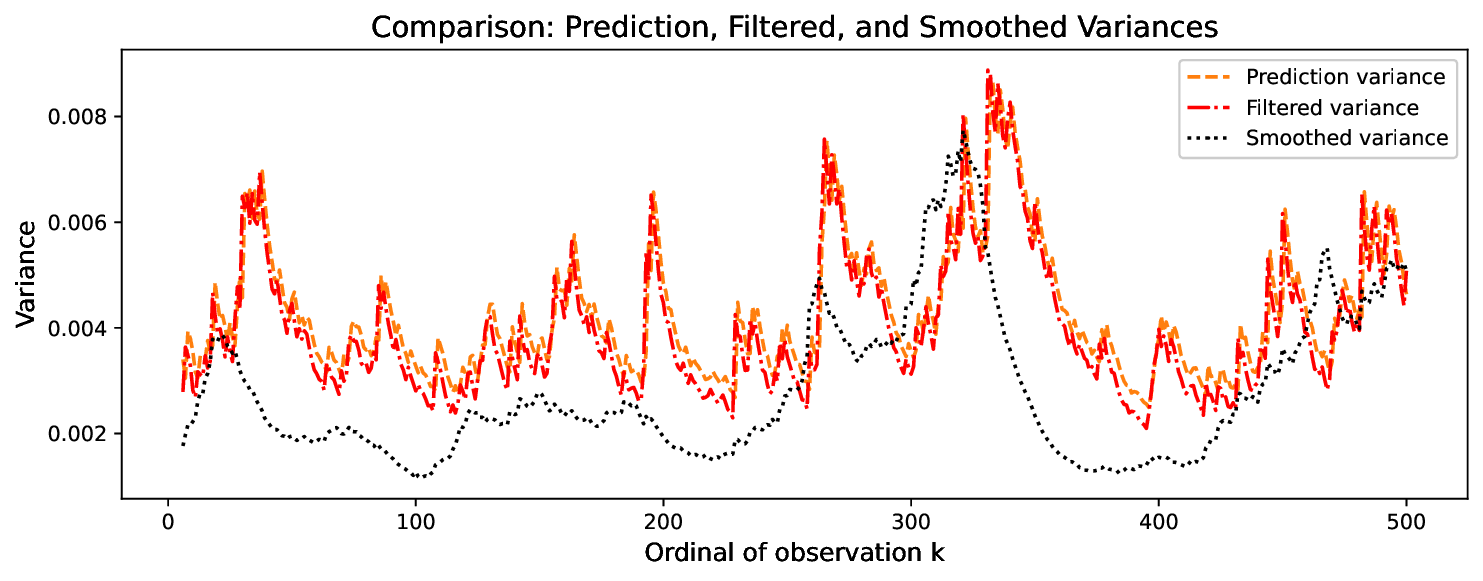}
			\vspace{-0.2cm}
		\end{minipage}
		\vspace{-15pt}\caption{{\protect\small Fourier filtering, prediction and
				smoothing for the BNS-Gamma model. The Fourier recursion is implemented using
				Corollary \ref{thm:fourier series update 2} with finite truncation and the
				smoothing is implemented using Proposition \ref{prop:smoothing}, with the
				truncation level set to $L=40$, $b_{0}=0$, and $b_{1}=0.9$.}}%
		\label{fig:bns_filter_smooth_comparison}%
	\end{figure}
	
	\begin{figure}[ptbh]
		\centering
		\includegraphics[width=0.95\textwidth]{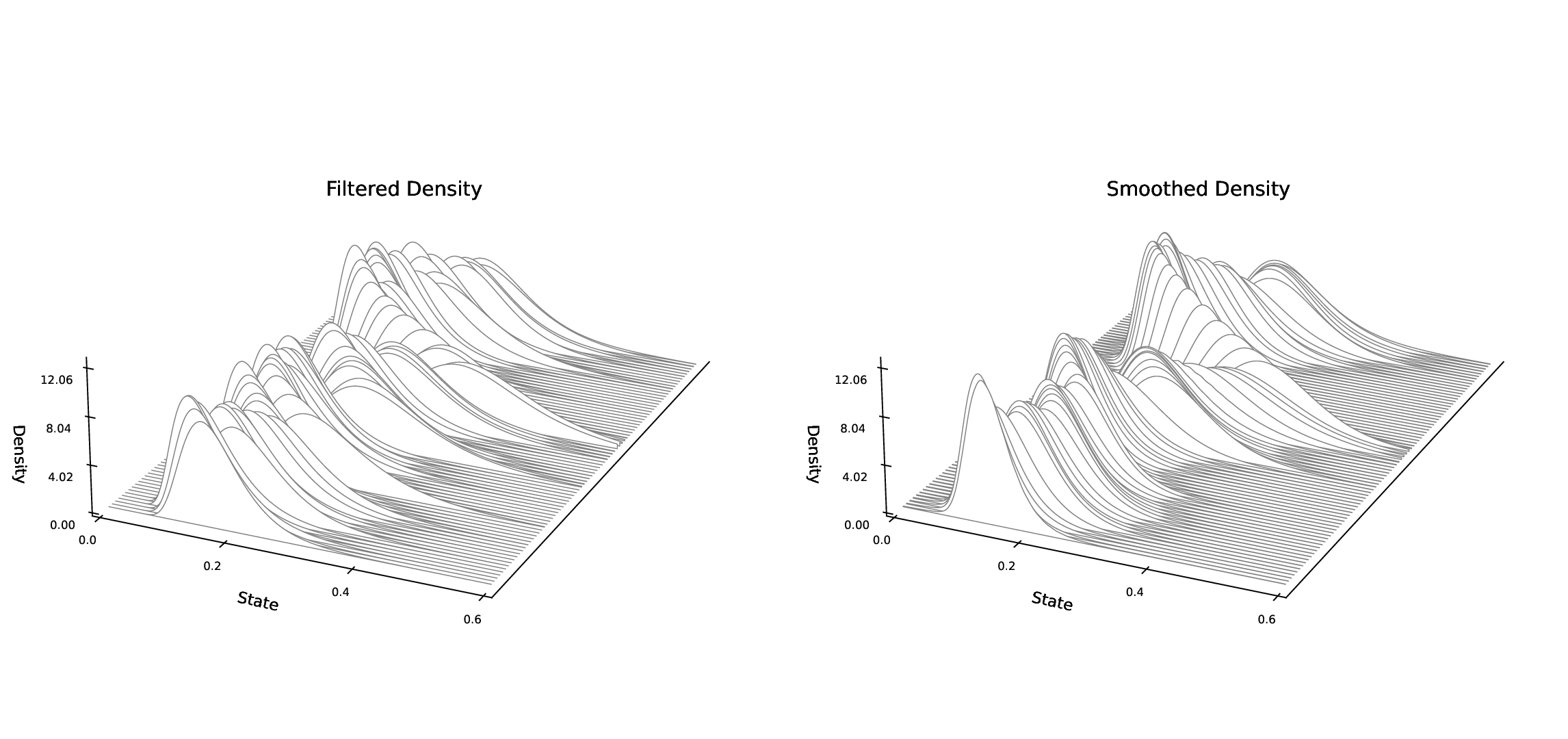}
		\caption{{\protect\small Filtered and smoothed density estimates for the
				BNS-Gamma model. The densities are recovered from the Fourier coefficients at
				selected observation times. }}%
		\label{fig:bns_filtered_smoothed_density_waterfall}%
	\end{figure}\begin{table}[tbh]
		\caption{AMMLE results for the BNS-Gamma model.}%
		\label{tab:bns_gamma_mc_fixed_beta}%
		\centering
		\par
		{\small \setlength{\tabcolsep}{0pt} \renewcommand{\arraystretch}{1.12} }
		\par
		{\small
			\begin{tabular*}
				{\textwidth}[c]{@{\extracolsep{\fill}}lccccccc@{}}\hline
				Parameter & True & \multicolumn{2}{c}{1000 obs} & \multicolumn{2}{c}{2000 obs}
				& \multicolumn{2}{c}{4000 obs}\\
				&  & Bias & Std. dev. & Bias & Std. dev. & Bias & Std. dev.\\\hline
				\multicolumn{8}{l}{$\hat{\theta}_{\mathrm{AMMLE}}^{(n)}$ obtained from Proposition \ref{prop:implement}
					}\\\hline
				$\mu$ & 0.1 & 0.063 & 0.369 & 0.054 & 0.317 & $-0.028$ & 0.251\\
				$\rho$ & $-0.2$ & $-0.029$ & 0.156 & $-0.003$ & 0.135 & 0.015 & 0.104\\
				$\lambda$ & 10 & 1.858 & 5.365 & 1.939 & 4.963 & 0.921 & 3.340\\
				$\nu$ & 3.6 & $-0.170$ & 1.208 & 0.101 & 0.942 & $-0.023$ & 0.616\\
				$\gamma$ & 16 & 0.195 & 5.953 & 0.467 & 4.721 & $-0.196$ & 2.757\\\hline
			\end{tabular*}
		}
		\par
		\begin{tablenotes}[flushleft]
			\footnotesize
			\item \textit{Notes:} The Fourier recursion is implemented using Proposition~\ref{prop:implement}, with the truncation level set to $L=40$, $b_0=0$, and $b_1=0.9$, and other settings are the same as Table~\ref{table:DSV}.
		\end{tablenotes}
	\end{table}Figure~\ref{fig:bns_filter_smooth_comparison} compares the Fourier
	prediction, filtering, and smoothing results for the BNS-Gamma model. The
	prediction and filtered means closely track the latent state, while smoothing
	further attenuates short-run fluctuations by using information from the full
	sample. The Fourier smoothed mean reduces the MSE to $0.0039$ and the RMSE to
	$0.0621$. Relative to the Fourier filtered mean, whose RMSE is $0.0805$,
	smoothing improves the RMSE by $22.90\%$.
	Figure~\ref{fig:bns_filtered_smoothed_density_waterfall} illustrates the
	resulting filtered and smoothed densities recovered from the Fourier
	coefficients. This evidence suggests that the Fourier recursion accurately
	captures the overall shape of the latent-state density. Table~\ref{tab:bns_gamma_mc_fixed_beta} reports the performance
	of the AMMLE based on the Fourier recursion. The estimates display small
	biases across all parameters, and the finite-sample standard deviations
	decline as the sample size increases from 1000 to 4000. The
	persistence and Gamma-shape parameters, $\lambda$ and $\gamma$, exhibit larger
	dispersion in smaller samples, but their precision improves substantially with
	more observations.
	
	\section{Empirical application: latent bankruptcy
		pressure\label{sec:empirical}}
	
	Systemic financial risk is inherently latent and is observed mainly through
	realized outcomes such as bankruptcies and credit-market disruptions.
	Aggregate distress may arise from default spillovers, common exposures, fire
	sales, balance-sheet linkages, and expectation-driven feedback (see, e.g.,
	\cite{AcemogluOzdaglarTahbazSalehi2015}, \cite{JacksonPernoud2021} and
	references therein). Because the underlying financial network and individual
	transmission channels are often unobserved, we model aggregate bankruptcy
	counts as noisy signals of a persistent latent bankruptcy-pressure state,
	which summarizes the aggregate effects of unobserved contagion and feedback
	mechanisms. Using the Fourier recursion, we estimate and filter this latent
	state in a nonlinear SSM with count observations. The framework also allows
	additional network, market, or macro-financial variables to be incorporated
	when available.
	
	\subsection{A latent Poisson LMM of bankruptcy pressure}
	
	Although the structural systemic forces, such as direct externalities and
	expectation-driven feedback, are conceptually clear, many of their underlying
	drivers are not separately observed. We therefore summarize these unobserved
	forces through an aggregate latent state. Namely, we model the observed
	bankruptcy count as a Poisson signal of an aggregate latent state:
	\begin{equation}
		Y_{k} \mid X_{k} \sim\mathrm{Poisson}(\lambda_{k}), \quad\log\lambda
		_{k}=\alpha+X_{k},\ X_{k} =\gamma X_{k-1}+\sigma\eta_{k-1},
		\label{eq:poisson model}%
	\end{equation}
	where $\eta_{k-1}\sim N(0,1)$. The latent state $X_{k}$ is interpreted as
	persistent aggregate bankruptcy pressure, providing a reduced-form but
	empirically tractable perspective for analyzing systemic financial risk. A
	higher value of $X_{k}$ raises the Poisson intensity and therefore represents
	an economic environment in which such filings are expected to occur more
	frequently. Accordingly, $\alpha$ determines the baseline log intensity when
	latent pressure is at its normalized long-run level, whereas $X_{k}$ captures
	persistent deviations from that baseline. Thus, $X_{k}$ is an outcome-based
	indicator of systemic financial distress rather than a structural measure of
	any individual transmission channel. The parameter vector is therefore
	$\theta=(\alpha,\gamma,\sigma)$.
	
	The specification is not restricted to bankruptcy filings. The same Poisson
	latent-state structure in \eqref{eq:poisson model} can be applied to other
	economic count outcomes in which observed events are noisy signals of an
	underlying latent pressure or risk state, such as corporate defaults, loan
	delinquencies. Explanatory variables $Z_{k}$, $U_{k}$ can be incorporated
	either into the Poisson intensity, i.e., $\log\lambda_{k}=\alpha+X_{k}%
	+Z_{k}^{\prime}\beta,$ or into the latent-state transition, i.e.,
	$X_{k}=\gamma X_{k-1}+U_{k}^{\prime}\Gamma+\sigma\eta_{k-1}$. Thus, the
	baseline model can be enriched when additional economically relevant
	information becomes observable.\footnote{The model can also be extended to
		financial network count data. Let $Y_{ij,t}$ denote the observed outcome from
		institution $i$ to institution $j$ at time $t$, such as lending transactions,
		exposure incidents, credit events, or distress transmissions. A Poisson
		network observation equation can be specified as $Y_{ij,t}|X_{t},Z_{ij,t}%
		\sim\operatorname{Poisson}(\lambda_{ij,t})$, with $\log\lambda_{ij,t}%
		=\alpha+\delta_{i}^{\mathrm{out}}+\delta_{j}^{\mathrm{in}}+X_{t}%
		+Z_{ij,t}^{\prime}\beta,$ where $\delta_{i}^{\mathrm{out}}$ and $\delta
		_{j}^{\mathrm{in}}$ capture sender and receiver heterogeneity, and $Z_{ij,t}$
		contains institution-level, or macro-financial explanatory variables.}
	
	\subsection{Data and estimated bankruptcy pressure}
	
	The main empirical series is constructed from the UCLA-LoPucki Bankruptcy
	Research Database (BRD). The BRD records large public-company bankruptcy cases
	filed in U.S. bankruptcy courts.\footnote{The UCLA-LoPucki Bankruptcy Research
		Database contains data on large public-company bankruptcy cases filed in U.S.
		bankruptcy courts since October 1, 1979. The database defines a public company
		using SEC annual-report filings and defines a large company using an asset
		threshold of at least $100$ million in 1980 dollars.} We aggregate the filing
	dates in the BRD to quarterly frequency and define $Y_{k}$ as the number of
	large public-company bankruptcy filings in quarter $k$. The resulting sample
	runs from 1980Q2 to 2022Q4 and contains 171 quarterly observations. The
	outcome is therefore a nonnegative integer count of large public-company
	bankruptcy filings, rather than a measure of all business bankruptcies. In the
	sample, the mean quarterly count is 7.12, the median is 5, and the maximum is 35.
	
	Let $\hat{\theta}$ be the AMMLEs obtained from the Fourier recursion,
	$m_{k|k}=\mathbb{E}[X_{k}|\mathbf{Y}_{k};\hat{\theta}]$ denote the filtered
	mean of the latent state and $m_{k|k-1}=\mathbb{E}[X_{k}|\mathbf{Y}_{k-1}%
	;\hat{\theta}]$ denote the one-step-ahead predicted mean. For visualization,
	we map these state estimates back into bankruptcy-count units using $\hat
	{Y}_{k|k}=\mathbb{E}[\exp(\hat{\alpha}+X_{k})|\mathbf{Y}_{k};\hat{\theta}]$
	and $\hat{Y}_{k|k-1}=\mathbb{E}[\exp(\hat{\alpha}+X_{k})|\mathbf{Y}_{k-1}%
	;\hat{\theta}]$. These two quantities are reported as the filtered and
	predicted recovered bankruptcy intensities. The likelihood itself is evaluated
	using the full predicted and filtered densities rather than only these
	posterior means.
	
	\begin{table}[tbh]
		\caption{Approximated MLEs for the Poisson latent-state model}%
		\label{tab:mle_results}
		\centering
		\par
		\begin{threeparttable}
			\small
			\renewcommand{\arraystretch}{1.15}
			\setlength{\tabcolsep}{11pt}
			
			\begin{tabular*}{\textwidth}{
					@{\extracolsep{\fill}}lccc@{}
				}
				\midrule
				Parameter
				& Estimate
				& Standard error
				& 95\% confidence interval
				\\
				\midrule
				
				$\alpha$
				& 1.5243
				& 0.2906
				& $[0.9548,\ 2.0938]$
				\\
				
				$\gamma$
				& 0.9106
				& 0.0396
				& $[0.8329,\ 0.9883]$
				\\
				
				$\sigma$
				& 0.3436
				& 0.0431
				& $[0.2591,\ 0.4281]$
				\\
				
				\midrule
			\end{tabular*}
			
			\begin{tablenotes}[flushleft]
				\footnotesize
				\item \textit{Notes:} Standard errors can be computed from the inverse observed-information
				matrix generated by the Fourier recursion in
				Proposition~\ref{prop:hessian updating}. The reported intervals are 95\% Wald confidence intervals constructed as the estimate plus or minus 1.96 standard errors..
				The numerical implementation uses Fourier truncation parameter
				\(L=36\) and latent-state box \([-5,5]\).
				The intercept implies a baseline bankruptcy intensity of
				\(\exp(\widehat{\alpha})=4.592\) bankruptcies per quarter.
			\end{tablenotes}
		\end{threeparttable}
	\end{table}The MLEs are reported in Table \ref{tab:mle_results}. The estimated
	autoregressive coefficient $\gamma$ is 0.9106, indicating strong persistence
		in the latent bankruptcy-pressure state. This
	implies a slow-moving latent factor rather than a purely transitory shock. The
	implied half-life of a state innovation is approximately 7.4 quarters. The innovation standard deviation is
	0.3436, which allows the latent state to move substantially during stress
	episodes while still maintaining a smooth path relative to the raw count series. Figure \ref{fig:fit_counts} plots the observed bankruptcy count together with
	the filtered and predicted recovered intensities. The filtered intensity
	tracks the broad movements in the data while smoothing out high-frequency
	count noise. The predicted intensity is less reactive by construction, because
	it is based only on information available before observing $Y_{k}$. The
	difference between the filtered and predicted intensity is therefore
	informative about the news contained in the current bankruptcy observation.
	
	\begin{figure}[ptb]
		\centering
		\includegraphics[width=0.88\textwidth]{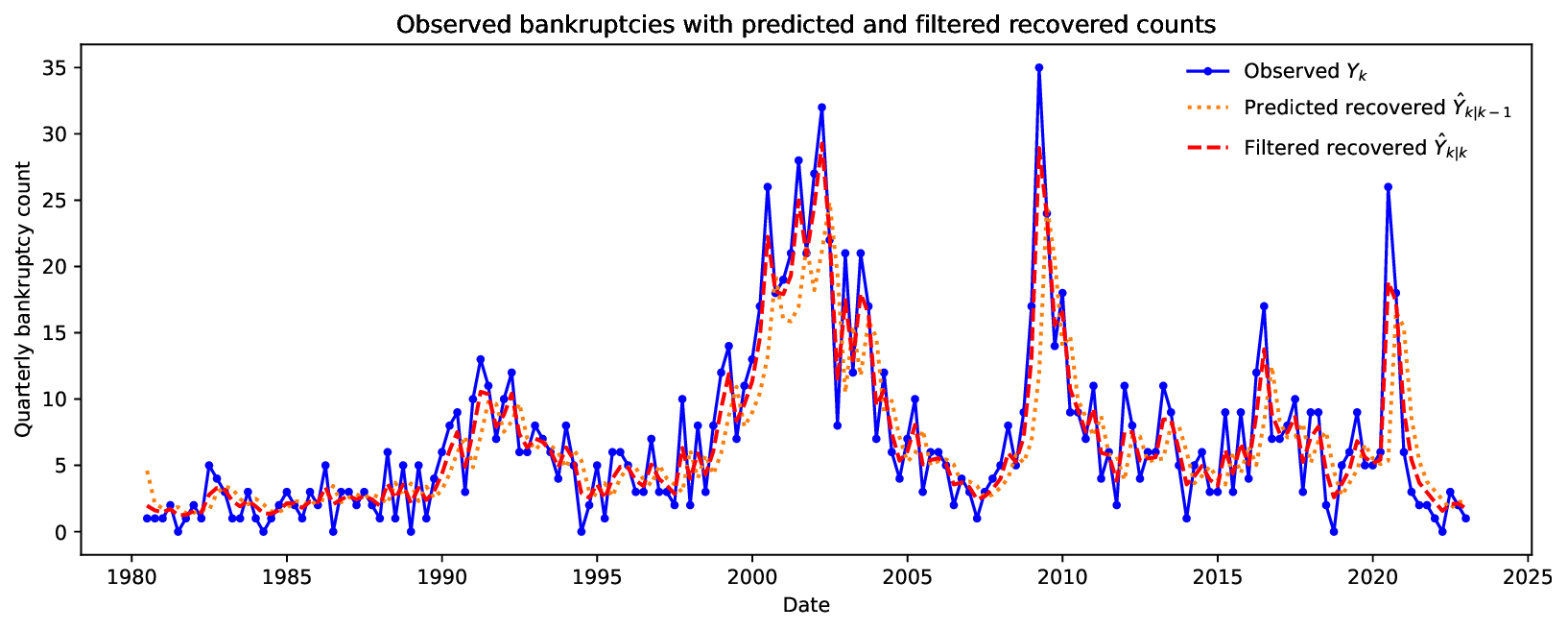}
		\caption{Observed bankruptcies, filtered intensity, and predicted intensity}%
		\label{fig:fit_counts}%
	\end{figure}
	
	\subsection{Estimating the Bankruptcy-count news}
	
	The Poisson state-space representation yields an additional object that is not
	available from the raw bankruptcy-count series. Let $J_{k}=m_{k|k}-m_{k|k-1},$
	which we refer to as bankruptcy-count news. The filtered mean $m_{k|k}%
	=m_{k|k-1}+J_{k}$ decomposes the updated assessment of bankruptcy pressure
	into a component anticipated from past filings and a revision induced by the
	current realization. Thus, $J_{k}$ measures the informational content of the
	current bankruptcy count for the persistent distress environment, rather than
	the level of realized bankruptcies itself. A positive $J_{k}$ indicates that
	the current count leads to an upward revision in latent bankruptcy pressure,
	whereas a negative $J_{k}$ indicates that the realization is weaker than the
	previously predicted state would imply.
	
	This distinction is economically important because identical, or similarly
	large, bankruptcy counts can convey different information depending on the
	distress environment anticipated before their realization. A high count may
	generate little news when bankruptcy pressure was already elevated, whereas a
	more moderate count may produce a substantial upward revision when prior
	pressure was low. Moreover, $J_{k}$ is not merely a contemporaneous count
	forecast error. Under the estimated state dynamics, its effect on the expected
	future state satisfies $\mathbb{E}[X_{k+h}|\mathbf{Y}_{k}]-\mathbb{E}%
	[X_{k+h}|\mathbf{Y}_{k-1}]=\gamma^{h}J_{k},\quad h\geq1.$ Consequently,
	$J_{k}$ measures how the current bankruptcy realization changes the expected
	path of persistent corporate distress. It should therefore be interpreted as a
	model-implied revision in bankruptcy pressure, rather than as a structurally
	identified economic shock.
	
	\begin{figure}[ptbh]
		\centering
		\includegraphics[width=0.6\textwidth]{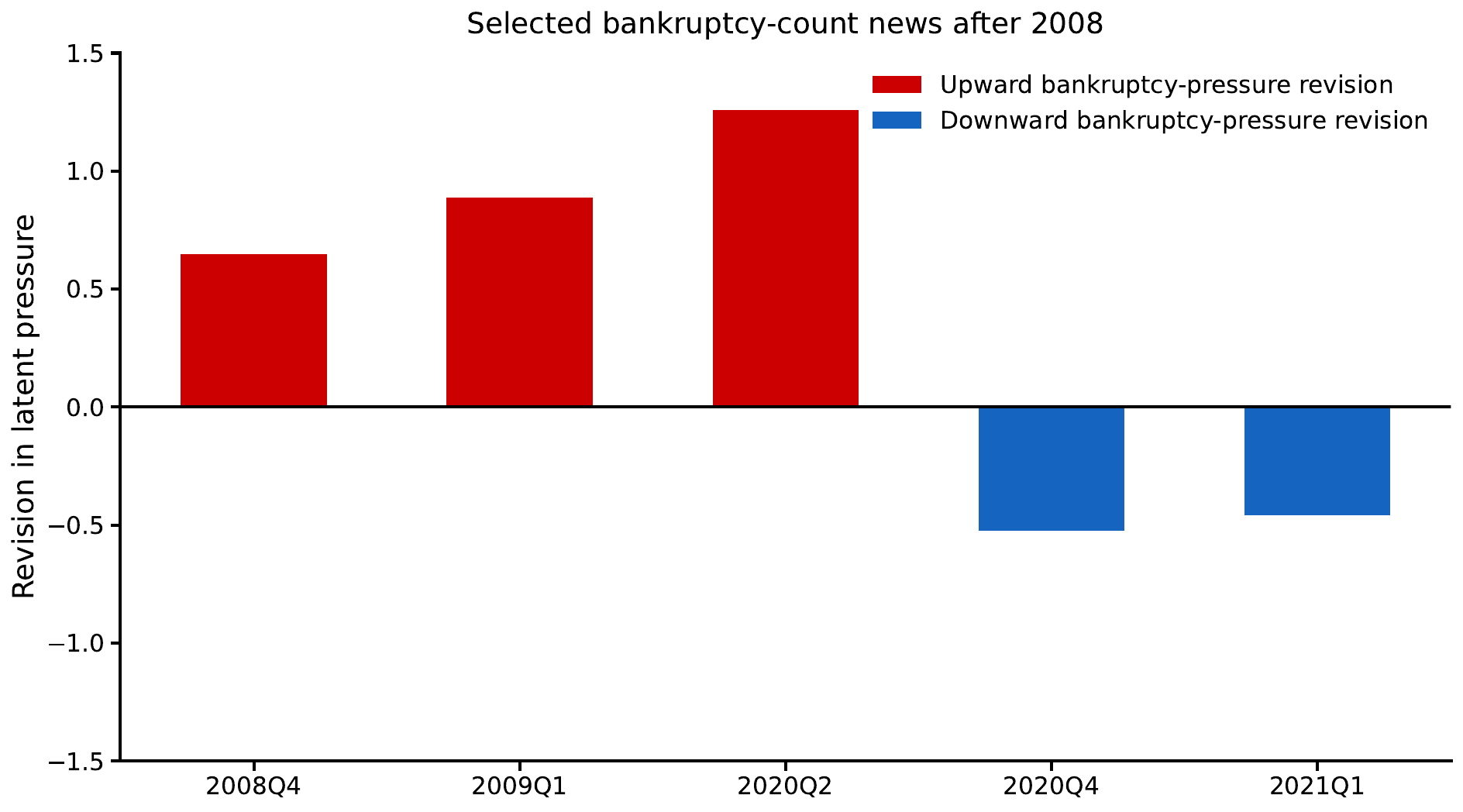}\caption{{\protect\small Selected
				large bankruptcy-count news after 2008. Positive values of $J_{k}$ indicate
				upward revisions in latent bankruptcy pressure, while negative values indicate
				downward revisions.}}%
		\label{fig:bankruptcy_count_news}%
	\end{figure}
	
	Figure \ref{fig:bankruptcy_count_news} interprets the selected revisions
	through this lens. In 2008Q4 and 2009Q1, around the collapse of Lehman
	Brothers and the acute phase of the global financial crisis, realized
	bankruptcy filings were well above the model's predicted intensities. The
	revision in 2009Q1 is particularly informative: even though the predicted
	intensity was already high at 11.91, the realization of 35 filings raised the
	filtered intensity to 28.94 and generated a state revision of 0.888. This
	suggests that the bankruptcy realization did not merely confirm an already
	weak environment, but revealed additional persistence in corporate distress as
	credit conditions deteriorated and refinancing constraints intensified. The
	same logic applies more sharply in 2020Q2. Before observing the quarter, the
	model predicted only 5.35 bankruptcies; after observing 26 filings, the
	filtered intensity rose to 18.84 and $J_{k}$ reached 1.258, consistent with
	the abrupt deterioration in corporate cash flows, economic activity, and
	uncertainty at the onset of the COVID-19 shock. By contrast, the negative
	revisions in 2020Q4 and 2021Q1 show that the model also adjusts downward when
	realized bankruptcies are lower than the previously elevated state would
	imply. This pattern is plausible because several policy interventions may have
	weakened the immediate pass-through from macro-financial stress to formal
	bankruptcy filings. These included the Federal Reserve's corporate credit
	facilities and the Main Street Lending Program.
	
	\section{Conclusion and discussions\label{sec:conclusion}}
	
	This paper develops a unified, closed-form Fourier recursion framework for
	statistical inference of latent Markov models. By representing the evolving
	filtering density through its Fourier coefficients and updating these
	coefficients recursively as new observations arrive, the proposed method
	jointly delivers filtering, prediction, likelihood evaluation, and online
	computation of the score and Hessian within a single deterministic system. The
	framework applies beyond standard SSMs and accommodates a broad
	class of latent Markov structures for which the relevant transformed
	transition functions are tractable, including models specified through
	conditional characteristic functions. Importantly, the positivity-preserving
	truncation ensures that, at every truncation level, the approximated filtering
	density remains a valid probability density and the likelihood remains well
	defined, thereby providing a numerically stable alternative to
	simulation-based likelihood and filtering procedures. This unified treatment
	of filtering, likelihood evaluation, and likelihood derivatives is one of the
	main practical advantages of the Fourier-recursion approach.
	
	The paper also establishes theoretical guarantees for the truncated recursion
	and the resulting likelihood-based estimators. We derive uniform approximation
	bounds for the filtering recursion and show how these bounds propagate to
	likelihood evaluation and latent-state moments. Building on a function-valued
	Markov augmentation, we establish consistency and asymptotic normality of the
	exact and approximate MLEs under the stated regularity conditions. When the
	truncation level increases sufficiently rapidly with the sample size, the
	approximation error becomes asymptotically negligible, so that the approximate
	estimator inherits the large-sample behavior of the infeasible exact MMLE.
	Moreover, because the score and Hessian are accumulated recursively within the
	same Fourier system, the negative normalized Hessian provides a feasible
	observed-information estimator for the standard likelihood-based inference. The misspecified-model case has also been treated  following a similar theoretical framework. Table \ref{tab:method_comparison}
	summarizes the comparison between the Fourier recursion and other methods.
	
	Several limitations of the current framework also point to promising
	directions for future research. First, an important direction for future
	research is to develop a sparse Fourier extension of the multivariate
	recursion. By reducing the number of retained Fourier coefficients, this
	extension may alleviate the curse of dimensionality and enable the framework
	to accommodate more complex data structures, such as high-dimensional network
	data. Second, the likelihood and asymptotic theory developed in this paper may
	provide a foundation for further statistical analysis of latent-variable
	models under model misspecification, including the characterization of
	pseudo-true parameters, robust inference, and specification testing.
	\newcolumntype{Y}{>{\RaggedRight\arraybackslash}X} \begin{sidewaystable}[p]
		\caption{Comparison of representative inference methods for LMMs}
		\label{tab:method_comparison}
		\centering
		
		\begingroup
		\small
		\renewcommand{\arraystretch}{1.24}
		\setlength{\tabcolsep}{7pt}
		\setlength{\aboverulesep}{0.45ex}
		\setlength{\belowrulesep}{0.65ex}
		\hyphenpenalty=10000
		\exhyphenpenalty=10000
		
		\begin{tabularx}{0.96\linewidth}{
				@{}
				>{\RaggedRight\arraybackslash}p{0.12\linewidth}
				>{\RaggedRight\arraybackslash}p{0.12\linewidth}
				Y
				Y
				Y
				Y
				@{}
			}
			\toprule
			
			\multicolumn{2}{c}{\textbf{Comparison}}
			& \multicolumn{1}{c}{\textbf{Kalman-filter type}}
			& \multicolumn{1}{c}{\textbf{Particle-filter type}}
			& \multicolumn{1}{c}{\textbf{MCMC type}}
			& \multicolumn{1}{c}{\textbf{Fourier recursion}}
			\\
			\midrule
			
			\textbf{Scope}
			& Model class
			& SSMs
			& SSMs
			& SSMs
			& LMMs with tractable \(H\) and \(H_X\)
			\\
			\midrule
			\addlinespace[5pt]
			
			\multirow[c]{2}{*}{%
				\shortstack[l]{%
					\textbf{Theoretical}\\
					\textbf{guarantee}}}
			& Filter convergence
			& \(\backslash\)
			& \(O_p(N^{-1/2})\); affected by particle degeneracy
			& \(\backslash\)
			& \(O(\log L/L)\) shown in Theorem \ref{thm:uni converge of filter}
			\\
			
			\cmidrule(lr){2-6}
			& Parameter estimation
			& Potentially biased MLEs
			& MLEs affected by Monte Carlo variability
			& Posterior mean affected by Monte Carlo variability
			& Consistency and asymptotic normality of MLEs shown in Theorem \ref{thm:mimle-long-span}
			\\
			\midrule
			\addlinespace[5pt]
			
			\multirow[c]{4}{*}{\textbf{Implementation}}
			& Optimization
			& Stable
			& Unstable
			& \(\backslash\)
			& Stable
			\\
			
			\cmidrule(lr){2-6}
			& Fisher information
			& \(\backslash\)
			& \(\backslash\)
			& \(\backslash\)
			& Tractable shown in Proposition \ref{prop:hessian updating}
			\\
			
			\cmidrule(lr){2-6}
			
			& Mode
			& Online
			& Online
			& Offline
			& Online
			\\
			
			\cmidrule(lr){2-6}
			& Tuning
			& Model-specific approximations
			& Particle number, resampling, proposal design
			& Proposal design, chain length, diagnostics
			& Truncation level \(L\) as shown in Proposition \ref{prop:implement}
			\\
			
			\bottomrule
		\end{tabularx}
		
		\endgroup
	\end{sidewaystable}
	\setcounter{section}{0} \setcounter{subsection}{0} \setcounter{equation}{0}
	\renewcommand
	{\thesection}{Appendix \Alph{section}}
	\renewcommand\thesubsection{\thesection.\arabic{subsection}} \renewcommand{\theequation}{\Alph{section}.\arabic{equation}}
	
	\section{Inference of LMMs with explanatory covariates \label{sec:covariates}}
	
	We consider the case when explanatory covariates $U_{k}$ are included
	into the LMMs. Let $\mathbf{u}_{k}=(u_{0},\ldots,u_{k})^{\intercal}$ denote
	the realized covariate history. Following the convention in
	\cite{DurbinKoopman2012}, explanatory variables are treated as observed
	external inputs. Therefore, inference is then based on the conditional
	marginal likelihood $\mathcal{L}_{n}(\mathbf{u}_{n};\theta)=p(\mathbf{y}%
	_{n}|\mathbf{u}_{n};\theta)$ and the filtered density $p(x_{k}|\mathbf{y}%
	_{k},\mathbf{u}_{k};\theta).$ We formulate it more generally by allowing the
	latent-observed transition mechanism to depend on the realized explanatory
	variables. Specifically, we work with the following conditional formulation.
	
	\begin{assumption}
		\label{assump:covariates} The covariate process $U_{k}$ is observed and
		explanatory to the LMM if the following conditions are satisfied. First,
		future covariates do not provide additional information for the current
		updating step. In particular, for $n\geq k$, $p(y_{k}|\mathbf{y}%
		_{k-1},\mathbf{u}_{n};\theta) = p(y_{k}|\mathbf{y}_{k-1},\mathbf{u}_{k}%
		;\theta), $ and $p(x_{k}|\mathbf{y}_{k},\mathbf{u}_{k+1};\theta) =
		p(x_{k}|\mathbf{y}_{k},\mathbf{u}_{k};\theta). $ Second, conditional on the
		realized covariate history $\mathbf{U}_{k}=\mathbf{u}_{k}$, $\left(
		X_{k},Y_{k}\right)  $ retains the conditional Markov structure $p_{(X,Y)}%
		(x_{k},y_{k}|\mathbf{x}_{k-1},\mathbf{y}_{k-1}, \mathbf{u}_{k};\theta) =
		p_{(X,Y)}(x_{k},y_{k}|x_{k-1},y_{k-1},u_{k};\theta).$
	\end{assumption}
	
	Under Assumption \ref{assump:covariates}, the conditional marginal likelihood
	can be decomposed as $\mathcal{L}_{n}(\mathbf{u}_{n};\theta)=\prod_{k=1}%
	^{n}\mathcal{L}_{k}(\mathbf{u}_{k};\theta),$ where $\mathcal{L}_{k}%
	(\mathbf{u}_{k};\theta)=p(y_{k}|\mathbf{y}_{k-1},\mathbf{u}_{k};\theta).$ A
	similar updating system can be obtained by replacing the transition and
	marginal transition densities without covariates by their conditional
	counterparts. More precisely, $p_{(X,Y)}(x_{k},y_{k}|x_{k-1},y_{k-1};\theta)$
	is replaced by $p_{(X,Y)}(x_{k},y_{k}|x_{k-1},y_{k-1},u_{k};\theta),$ and the
	marginal transition density $p_{Y}(y_{k}|x_{k-1},y_{k-1};\theta)$ is replaced
	by $p_{Y}(y_{k}|x_{k-1},y_{k-1},u_{k};\theta)$.
	
	Here, we derive the Bayesian updating system in Section
	\ref{sec:bayes updating} with explanatory covariates $U_{k}$. By the
	conditional Markov structure in Assumption \ref{assump:covariates}, the
	marginal transition densities satisfy $p_{Y}(y_{k}|x_{k-1},\mathbf{y}%
	_{k-1},\mathbf{u}_{k};\theta)=p_{Y}(y_{k}|x_{k-1},y_{k-1},u_{k};\theta),$ and
	$p_{X}(x_{k}|x_{k-1},\mathbf{y}_{k-1},\mathbf{u}_{k};\theta)=p_{X}%
	(x_{k}|x_{k-1},y_{k-1},u_{k};\theta).$ Together with the non-anticipation
	condition in Assumption \ref{assump:covariates}, these identities yield the
	same recursive structure as in Section \ref{sec:bayes updating}. First, the
	likelihood updating term can be written as
	\begin{align*}
		\mathcal{L}_{k}(\mathbf{u}_{k};\theta)  &  =\int_{\mathcal{X}}p_{Y}%
		(y_{k}|x_{k-1},\mathbf{y}_{k-1},\mathbf{u}_{k};\theta)p(x_{k-1}|\mathbf{y}%
		_{k-1},\mathbf{u}_{k};\theta)\nu_{X}(dx_{k-1})\\
		&  =\int_{\mathcal{X}}p_{Y}(y_{k}|x_{k-1},y_{k-1},u_{k};\theta)p(x_{k-1}%
		|\mathbf{y}_{k-1},\mathbf{u}_{k-1};\theta)\nu_{X}(dx_{k-1}).
	\end{align*}
	The second equality follows from Assumption \ref{assump:covariates}: the first
	factor is simplified by the conditional Markov structure, and the second
	factor is simplified by the non-anticipation condition. Similarly, the
	filtering update is
	\begin{align*}
		p(x_{k}|\mathbf{y}_{k},\mathbf{u}_{k};\theta)  &  =\frac{1}{\mathcal{L}%
			_{k}(\mathbf{u}_{k};\theta)}\int_{\mathcal{X}}p_{(X,Y)}(x_{k},y_{k}%
		|x_{k-1},\mathbf{y}_{k-1},\mathbf{u}_{k};\theta)p(x_{k-1}|\mathbf{y}%
		_{k-1},\mathbf{u}_{k};\theta)\nu_{X}(dx_{k-1})\\
		&  =\frac{1}{\mathcal{L}_{k}(\mathbf{u}_{k};\theta)}\int_{\mathcal{X}%
		}p_{(X,Y)}(x_{k},y_{k}|x_{k-1},y_{k-1},u_{k};\theta)p(x_{k-1}|\mathbf{y}%
		_{k-1},\mathbf{u}_{k-1};\theta)\nu_{X}(dx_{k-1}).
	\end{align*}
	This is the usual Bayes update, with the joint transition density conditioned
	on the realized covariate $u_{k}$. Finally, the one-step prediction from time
	$k$ to time $k+1$ is conditional on the next realized covariate $u_{k+1}$. By
	the law of total probability,
	\begin{align*}
		p(x_{k+1}|\mathbf{y}_{k},\mathbf{u}_{k+1};\theta)  &  =\int_{\mathcal{X}}%
		p_{X}(x_{k+1}|x_{k},\mathbf{y}_{k},\mathbf{u}_{k+1};\theta)p(x_{k}%
		|\mathbf{y}_{k},\mathbf{u}_{k+1};\theta)\nu_{X}(dx_{k})\\
		&  =\int_{\mathcal{X}}p_{X}(x_{k+1}|x_{k},y_{k},u_{k+1};\theta)p(x_{k}%
		|\mathbf{y}_{k},\mathbf{u}_{k};\theta)\nu_{X}(dx_{k}).
	\end{align*}
	The second equality applies the non-anticipation condition and the conditional
	Markov structure in Assumption \ref{assump:covariates}. Therefore, observed
	explanatory covariates only make the model-specific transition and marginal
	transition densities depend on the realized covariate path. 
	
	\section{Lemmas and proofs\label{appendix:proofs}}
	
	\subsection{Proof of Theorem \ref{thm:fourier series update 1}}
	
	\begin{proof}
		We first derive the likelihood update. By the definition of $\not \Psi $ in
		(\ref{eq:def psi}) and the Bayes updating system, we have
		\begin{equation}
			\mathcal{L}_{k}(\theta)=\int_{\mathcal{X}}\Psi(0,y_{k}|x_{k-1},y_{k-1}%
			;\theta)p(x_{k-1}|\mathbf{y}_{k-1};\theta)\nu_{X}(dx_{k-1}).
			\label{eq:proof likelihood one}%
		\end{equation}
		Substituting the Fourier representation of $p(x_{k-1}|\mathbf{y}_{k-1}%
		;\theta)$ into (\ref{eq:proof likelihood one}), we obtain gives
		\begin{equation}
			\mathcal{L}_{k}(\theta)=\sum_{s\in\mathbb{Z}}c_{k-1,s}(\theta)H\left(
			0,y_{k},\frac{2\pi s}{b_{1}-b_{0}},y_{k-1};\theta\right)  ,\nonumber
		\end{equation}
		which proves (\ref{eq:fourier likelihood update}). We next derive the
		recursion for the filtered Fourier coefficients. Similarly, using the
		filtering update in (\ref{eq:filter update}), we have
		\begin{equation}
			c_{k,l}(\theta)=\frac{1}{\mathcal{L}_{k}(\theta)(b_{1}-b_{0})}\int%
			_{\mathcal{X}}\Psi\left(  -\frac{2\pi l}{b_{1}-b_{0}},y_{k}|x_{k-1}%
			,y_{k-1};\theta\right)  p(x_{k-1}|\mathbf{y}_{k-1};\theta)\nu_{X}(dx_{k-1}).
			\label{eq:proof coefficient two}%
		\end{equation}
		By the definition of $H$ in (\ref{eq:def H}), substituting the Fourier
		representation of $p(x_{k-1}|\mathbf{y}_{k-1};\theta)$ into
		(\ref{eq:proof coefficient two}), we have
		\begin{equation}
			c_{k,l}(\theta)=\sum_{s\in\mathbb{Z}}\frac{c_{k-1,s}(\theta)}{\mathcal{L}%
				_{k}(\theta)(b_{1}-b_{0})}H\left(  -\frac{2\pi l}{b_{1}-b_{0}},y_{k}%
			,\frac{2\pi s}{b_{1}-b_{0}},y_{k-1};\theta\right)  ,\nonumber
		\end{equation}
		which proves (\ref{eq:fourier filter update}). It remains to derive the
		one-step prediction formula. The predictive density satisfies $p(x_{k}%
		|\mathbf{y}_{k-1};\theta)=\int_{\mathcal{X}}p_{X}(x_{k}|x_{k-1},y_{k-1}%
		;\theta)p(x_{k-1}|\mathbf{y}_{k-1};\theta)\nu_{X}(dx_{k-1}).$ Hence,
		(\ref{eq:fourier prediction update}) can be proved similarly.
	\end{proof}
	
	\subsection{Proof of Theorem \ref{thm:uni converge of filter}}
	
	To prove Theorem \ref{thm:uni converge of filter}, we first need a lemma that
	characterizes the one-step approximation error.
	
	\begin{lemma}
		\label{prop:one step error}Under Assumption \ref{ass:S6},
		\ref{assump: bound derivative}, the $k$th-step truncation admits the one-step
		error as
		\[
		\sup_{f:\mathcal{X\rightarrow}\mathbb{R},\left\Vert f\right\Vert _{\infty
			}=1\,}\left\vert \int_{\mathcal{X}}\left\vert \bar{p}^{(L)}(x|\mathbf{y}%
		_{k};\theta)-\hat{p}^{(L)}(x|\mathbf{y}_{k};\theta)\right\vert
		f(x)\,dx\right\vert \leq R_{1}^{(k)}(L),
		\]
		for any $\theta\in\Theta$, where $\left\Vert f\right\Vert _{\infty}=\sup
		_{x\in\mathcal{X}}\left\vert f(x)\right\vert $ and $R_{1}^{(k)}(L)=O\left(
		\log L/L\right)  $.
	\end{lemma}
	
	\begin{proof}
		See Online Appendix \ref{sec:one-step error}.
	\end{proof}
	
	With Lemma \ref{prop:one step error}, we are able to prove Theorem
	\ref{thm:uni converge of filter} by utilizing the standard procedure of
	\cite{gland2004}.
	
	\begin{proof}
		Let $\hat{r}_{n}^{(L)}(x;\theta):=p(x|\mathbf{y}_{n};\theta)-\hat{p}%
		^{(L)}(x|\mathbf{y}_{n};\theta),$ Theorem 4.8 of \cite{gland2004} implies
		that
		\[
		\sup_{f:\mathcal{X}\rightarrow\mathbb{R},\,\left\Vert f\right\Vert
			=1}\left\vert \int_{\mathcal{X}}\hat{r}_{n}^{(L)}(x;\theta)f(x)dx\right\vert
		\leq\left(  R_{1}^{(n)}(L)+\frac{2R_{1}^{(n-1)}(L)}{\epsilon_{n}^{2}}%
		+\sum_{k=1}^{n-2}\tau_{n:k+3}\frac{4R_{1}^{(k)}(L)}{\epsilon_{k+2}^{2}%
			\epsilon_{k+1}^{2}\log3}\right)  ,
		\]
		where $\tau_{n:m}=\tau_{n}\tau_{n-1}\cdots\tau_{m}$ for $m\leq n$, with
		$\tau_{k}$ denoted as the Birkhoff contraction coefficient for the $k$th-step
		transition kernel. Assumption~\ref{assump: bound derivative} and
		Proposition~3.9(1) of \cite{gland2004} imply for any $k\geq1$ that $\tau
		_{k}\leq\frac{1-\epsilon_{k}^{2}}{1+\epsilon_{k}^{2}}\leq\frac{1-\epsilon^{2}%
		}{1+\epsilon^{2}}.$ Thus, by bringing in the one-step error in
		Lemma~\ref{prop:one step error}, it holds that
		\begin{align*}
			\sup_{f:\mathcal{X}\rightarrow\mathbb{R},\,\left\Vert f\right\Vert
				=1}\left\vert \int_{\mathcal{X}}\hat{r}_{n}^{(L)}(x;\theta)f(x)dx\right\vert
			&  \precsim\left(  \frac{\log L}{L}+\frac{2\log L}{\epsilon^{2}L}+\frac{4\log
				L}{L\epsilon^{4}\log3}\sum_{k=1}^{n-2}\left(  \frac{1-\epsilon^{2}}%
			{1+\epsilon^{2}}\right)  ^{n-k+2}\right) \\
			&  \leq\left(  1+\frac{2}{\epsilon^{2}}+\frac{4}{\epsilon^{4}\log3}\sum
			_{k=1}^{n-2}\left(  \frac{1-\epsilon^{2}}{1+\epsilon^{2}}\right)
			^{n-k+2}\right)  \frac{\log L}{L}.
		\end{align*}
		Here $a\precsim b$ means that $a\leq Cb$ for some positive constant $C.$
	\end{proof}
	
	\subsection{Proof of Theorem \ref{thm:mimle-long-span}}
	
	To prove Theorem \ref{thm:mimle-long-span}, we first need two auxiliary lemmas
	to ensure that the augmented process $\xi_{k}^{+}(\theta)$ is ergodic and the
	sample log-likelihood function $n^{-1}\sum_{k=1}^{n}\ell_{k}(\theta)$
	satisfies several regular properties.
	
	\begin{lemma}
		\label{lem:filter-ergodic}Suppose Assumptions \ref{ass:S6}, \ref{ass:S2},
		\ref{ass:S4} and \ref{ass:S5} hold. Then $\xi_{k}(\theta)$, $\xi_{k}%
		^{+}(\theta)$, and $(X_{k+1},Y_{k+1},\xi_{k}(\theta))$ are ergodic Markov
		processes. Let $(X^{+},Y^{+},X,Y,p(\cdot;\theta),\dot{p}(\cdot;\theta))$
		follow the stationary distribution of $(X_{k+1},Y_{k+1},\xi_{k}(\theta)),$
		$\xi(\theta)=(X,Y,p(\cdot;\theta),\dot{p}(\cdot;\theta)),\quad$and $\xi
		^{+}(\theta)=(Y^{+},\xi(\theta)).$ Then $\xi(\theta)$ and $\xi^{+}(\theta)$
		follow the stationary distributions of $\xi_{k}(\theta)$ and $\xi_{k}%
		^{+}(\theta)$, respectively. Moreover, $(X,Y)$ follows the stationary
		distribution of the ergodic process $(X_{k},Y_{k})$, and for any measurable
		function $f$ on $\mathcal{X}\times\mathcal{Y}$, $(x_{0},y_{0},f_{0},g_{0}%
		)\in\Xi,$ we have
		\begin{align}
			\mathbb{E}[f(X^{+},Y^{+})|X,Y]  &  =\mathbb{E}[f(X_{1},Y_{1})|X_{0}%
			,Y_{0}],\text{ }\label{eq:32}\\
			\mathbb{E}[f(X^{+},Y^{+})|\xi(\theta)]  &  =\mathbb{E}[f(X^{+},Y^{+})|(X,Y)].
			\label{eq:33}%
		\end{align}
		
	\end{lemma}
	
	\begin{proof}
		See Online Appendix \ref{sec:ergodicity}.
	\end{proof}
	
	\begin{lemma}
		\label{lemma:mimle 1}Suppose Assumption \ref{ass:S6}, \ref{ass:S2},
		\ref{ass:S4} and \ref{ass:S5} hold. Let $\ell^{(n)}(\theta)=\sum_{k=1}^{n}%
		\ell_{k}(\theta)$ be the log-likelihood of LMM with $n$ observations,
		$S_{n}\left(  \theta\right)  =\partial\ell^{(n)}(\theta)/\partial\theta$ and
		$H_{n}(\theta)=\partial^{2}\ell^{(n)}(\theta)/\partial\theta\partial
		\theta^{\intercal}.$ Then, there exists a deterministic function $\ell
		(\theta)$, such that: \ (i) $n^{-1}\ell^{(n)}(\theta
		)\overset{p}{\longrightarrow}\ell(\theta)$ uniformly holds for $\theta
		\in\Theta$, (ii) $\ell(\theta)<\ell(\theta_{0})$ for any $\theta\neq\theta
		_{0}$, (iii) $n^{-1/2}S_{n}(\theta_{0})\overset{d}{\longrightarrow}%
		N(0,I^{M}).$ There also exists a matrix-valued function $H_{0}(\theta)$ such
		that: $n^{-1}H_{n}(\theta)\overset{p}{\longrightarrow}H_{0}(\theta)$ uniformly
		holds for $\theta\in\Theta$, where $H_{0}(\theta)$ is Lipschitz on $\Theta,$
		and $-H_{0}(\theta_{0})=I^{M}.$
	\end{lemma}
	
	\begin{proof}
		See Online Appendix \ref{sec:mimle}.
	\end{proof}
	
	Now we are able to prove Theorem \ref{thm:mimle-long-span}.
	
	\begin{proof}
		The proof consists of two steps. First, we prove the consistency of
		$\hat{\theta}^{(n)}$, i.e., $\hat{\theta}^{(n)}\overset{p}{\longrightarrow
		}\theta_{0}.$ Second, we prove the asymptotic normality of $\hat{\theta}%
		^{(n)}$, i.e., $\sqrt{n}(\hat{\theta}^{(n)}-\theta_{0}%
		)\overset{d}{\longrightarrow}N(0,\left(  I^{M}\right)  ^{-1}).$ Based on the
		additive forms of log-likelihood, score, and Hessian constructed via the
		ergodic Markov processes in Lemma \ref{lem:filter-ergodic}, we are able to
		follow the route in \cite{vanderVaart1998} to complete the proof.
		
		Step 1: Proving the consistency of $\hat{\theta}^{(n)}$. As long as we have
		(i) and (ii) of Lemma \ref{lemma:mimle 1}, by Theorem 5.7 of van der Vaart
		(1998), we obtain $\hat{\theta}^{(n)}\overset{p}{\longrightarrow}\theta_{0}.$
		Step 2: Proving the asymptotic normality of $\hat{\theta}^{(n)}$: Following
		the standard route, we obtain by the first-order condition that $S_{n}%
		(\hat{\theta}^{(n)})=0,$ where $S_{n}\left(  \theta\right)  =\partial
		\ell^{(n)}(\theta)/\partial\theta.$ Then, by the Taylor expansion, we obtain
		$S_{n}(\theta_{0})+H_{n}(\theta_{n}^{A})(\hat{\theta}^{(n)}-\theta_{0})=0,$
		with $\theta_{n}^{A}=\lambda_{n}\hat{\theta}^{(n)}+(1-\lambda_{n})\theta
		_{0},\quad\lambda_{n}\in\lbrack0,1].$ Since $\hat{\theta}^{(n)}%
		\overset{p}{\longrightarrow}\theta_{0}$, we have $\theta_{n}^{A}%
		\overset{p}{\longrightarrow}\theta_{0}$. Therefore,
		\[
		\sqrt{n}(\hat{\theta}^{(n)}-\theta_{0})=-[n^{-1}H_{n}(\theta_{n}^{A}%
		)]^{-1}n^{-1/2}S_{n}(\theta_{0}).
		\]
		By Lemma \ref{lemma:mimle 1}, $n^{-1/2}S_{n}(\theta_{0}%
		)\overset{d}{\longrightarrow}N(0,I^{M})$ and $n^{-1}H_{n}(\theta
		)\overset{p}{\longrightarrow}H_{0}(\theta),$ we have
		\[
		\left\Vert n^{-1}H_{n}(\theta_{n}^{A})-H_{0}(\theta_{0})\right\Vert
		\leq\left\Vert n^{-1}H_{n}(\theta_{n}^{A})-H_{0}(\theta_{n}^{A})\right\Vert
		+\left\Vert H_{0}(\theta_{n}^{A})-H_{0}(\theta_{0})\right\Vert ,
		\]
		where $\left\Vert n^{-1}H_{n}(\theta_{n}^{A})-H_{0}(\theta_{n}^{A})\right\Vert
		\leq\sup_{\theta\in\Theta}\left\Vert n^{-1}H_{n}(\theta)-H_{0}(\theta
		)\right\Vert _{1}\rightarrow0$. Since $H_{0}(\theta)$ is Lipschitz, we have
		$\left\Vert H_{0}(\theta_{n}^{A})-H_{0}(\theta_{0})\right\Vert \leq
		C_{H}\left\Vert \theta_{n}^{A}-\theta_{0}\right\Vert ,$ which also goes to
		zero since $\theta_{n}^{A}\overset{p}{\longrightarrow}\theta_{0}.$ Thus, we
		obtain $\left\Vert n^{-1}H_{n}(\theta_{n}^{A})-H_{0}(\theta_{0})\right\Vert
		\overset{p}{\longrightarrow}0.$ Hence, by Slutsky's theorem, the asymptotic
		normality of $\hat{\theta}^{(n)}$ is obtained.
		
		Finally, we derive the upper bounds for $\left\Vert \hat{\theta}%
		_{\text{AMMLE}}^{(n,L)}-\theta_{0}\right\Vert $. Consider the non-negative
		term $\ell(\theta_{0})-\ell(\hat{\theta}_{\text{AMMLE}}^{(n,L)})$. On one
		hand, it admits the upper bound $\ell(\theta_{0})-\ell(\hat{\theta
		}_{\text{AMMLE}}^{(n,L)})\leq J_{1}+J_{2}+J_{3},$ where
		\[
		J_{1}:=\left\vert \ell(\theta_{0})-\frac{1}{n}\sum_{k=1}^{n}\ell_{k}%
		(\theta_{0})\right\vert ,\text{ }J_{2}:=\left\vert \ell(\hat{\theta
		}_{\text{AMMLE}}^{(n,L)})-\frac{1}{n}\sum_{k=1}^{n}\ell_{k}(\hat{\theta
		}_{\text{AMMLE}}^{(n,L)})\right\vert ,\text{ }J_{3}:=\frac{1}{n}\left[
		\sum_{k=1}^{n}\ell_{k}(\theta_{0})-\sum_{k=1}^{n}\ell_{k}(\hat{\theta
		}_{\text{AMMLE}}^{(n,L)})\right]  .
		\]
		It is easy to verify that both $J_{1}$ and $J_{2}$ will be smaller than
		$\sup_{\theta\in\Theta}\left\vert \ell(\theta)-n^{-1}\sum_{k=1}^{n}\ell
		_{k}(\theta)\right\vert .$ Further, $J_{3}$ can be controlled by
		\begin{align*}
			J_{3}  &  \leq\left\vert \frac{1}{n}\sum_{k=1}^{n}\ell_{k}(\theta_{0}%
			)-\frac{1}{n}\sum_{k=1}^{n}\hat{\ell}_{k}^{(L)}(\theta_{0})\right\vert
			+\frac{1}{n}\left(  \sum_{k=1}^{n}\hat{\ell}_{k}^{(L)}(\theta_{0})-\sum
			_{k=1}^{n}\ell_{k}(\hat{\theta}_{\text{AMMLE}}^{(n,L)})\right) \\
			&  \leq\left\vert \frac{1}{n}\sum_{k=1}^{n}\ell_{k}(\theta_{0})-\frac{1}%
			{n}\sum_{k=1}^{n}\hat{\ell}_{k}^{(L)}(\theta_{0})\right\vert +\frac{1}%
			{n}\left(  \sum_{k=1}^{n}\hat{\ell}_{k}^{(L)}(\hat{\theta}_{\text{AMMLE}%
			}^{(n,L)})-\sum_{k=1}^{n}\ell_{k}(\hat{\theta}_{\text{AMMLE}}^{(n,L)})\right)
			\\
			&  \leq2\sup_{\theta\in\Theta}\lvert\frac{1}{n}\sum_{k=1}^{n}\ell_{k}%
			(\theta)-\frac{1}{n}\sum_{k=1}^{n}\hat{\ell}_{k}^{(L)}(\theta)\rvert
			\leq2R_{\ell}(L),
		\end{align*}
		where the second inequality follows from the fact that $\hat{\theta
		}_{\text{AMMLE}}^{(n,L)}=\operatorname*{argmax}_{\theta\in\Theta}\sum
		_{k=1}^{n}\hat{\ell}_{k}^{(L_{k-1})}(\theta),$ and the last inequality follows
		from Corollary~\ref{coro:uniform error for likelihood} with $R_{\ell
		}(L)=O(\log L/L)$. Thus we obtain that
		\begin{equation}
			\ell(\theta_{0})-\ell(\hat{\theta}_{\text{AMMLE}}^{(n,L)})\leq2\sup_{\theta
				\in\Theta}\left\vert \ell(\theta)-\frac{1}{n}\sum_{k=1}^{n}\ell_{k}%
			(\theta)\right\vert +2R_{\ell}(L). \label{eq:pf taylor 0}%
		\end{equation}
		On the other hand, the Taylor expansion yields
		\[
		\ell(\hat{\theta}_{\text{AMMLE}}^{(n,L)})=\ell(\theta_{0})+\frac{1}{2}%
		(\hat{\theta}_{\text{AMMLE}}^{(n,L)}-\theta_{0})^{\top}\left.  \frac
		{\partial^{2}\ell(\theta)}{\partial\theta\partial\theta^{\top}}\right\vert
		_{\theta=\theta^{\prime}}(\hat{\theta}_{\text{AMMLE}}^{(n,L)}-\theta_{0}),
		\]
		where $\theta^{\prime}$ lies on the segment between $\theta_{0}$ and
		$\hat{\theta}_{\text{AMMLE}}^{(n,L)}$, and thus
		\begin{equation}
			\ell(\theta_{0})-\ell(\hat{\theta}_{\text{AMMLE}}^{(n,L)})=\frac{-1}{2}%
			(\hat{\theta}_{AMMLE}^{(n,L)}-\theta_{0})^{\top}\left.  \frac{\partial^{2}%
				\ell(\theta)}{\partial\theta\partial\theta^{\top}}\right\vert _{\theta
				=\theta^{\prime}}(\hat{\theta}_{\text{AMMLE}}^{(n,L)}-\theta_{0})\geq
			\frac{c_{0}}{2}\left\Vert \hat{\theta}_{\text{AMMLE}}^{(n,L)}-\theta
			_{0}\right\Vert ^{2}. \label{eq:pf taylor 1}%
		\end{equation}
		By combining the previous two displays, it holds
		\[
		\left\Vert \hat{\theta}_{\text{AMMLE}}^{(n,L)}-\theta_{0}\right\Vert ^{2}%
		\leq\frac{4}{c_{0}}\sup_{\theta\in\Theta}\left\vert \ell(\theta)-\frac{1}%
		{n}\sum_{k=1}^{n}\ell_{k}(\theta)\right\vert +\frac{4}{c_{0}}R_{\ell}(L).
		\]
		By taking the upper limit, the first term on the right side vanishes since
		$n^{-1}\sum_{k=1}^{n}\ell_{k}(\theta)\overset{p}{\rightarrow}\ell(\theta)$
		proved in (\ref{eq:likelihood convergence}), and it yields
		(\ref{eq:par est bound}).
		
		Moreover, following similar Taylor expansion in (\ref{eq:pf taylor 1}), we
		obtain that%
		\[
		n^{-1}\ell^{(n)}(\hat{\theta}^{(n)})-n^{-1}\ell^{(n)}(\hat{\theta
		}_{\text{AMMLE}}^{(n,L)})=\frac{1}{2}(\hat{\theta}_{AMMLE}^{(n,L)}-\hat
		{\theta}^{(n)})^{\top}\left(  -n^{-1}H_{n}(\theta^{\prime\prime})\right)
		(\hat{\theta}_{\text{AMMLE}}^{(n,L)}-\hat{\theta}^{(n)}),
		\]
		Since $\Theta$ is compact and $n^{-1}H_{n}(\theta)\overset{p}{\longrightarrow
		}H_{0}(\theta)$ uniformly holds for $\theta\in\Theta$ as shown in Step 2.2, we
		have $H_{0}(\theta)=\partial^{2}\ell(\theta)/\partial\theta\partial
		\theta^{\top}.$ Under Assumption \ref{assump:hessian bound}, we have
		$-n^{-1}H_{n}(\theta)\geq c_{0}/2I$ with probability tending to one.
		Therefore, $c_{0}\left\Vert \hat{\theta}_{AMMLE}^{(n,L)}-\hat{\theta}%
		^{(n)}\right\Vert ^{2}/4\leq n^{-1}\left\vert \ell^{(n)}(\hat{\theta}%
		^{(n)})-\ell^{(n)}(\hat{\theta}_{\text{AMMLE}}^{(n,L)})\right\vert ,$ with
		probability tending to one. Also, we notice that
		\begin{align*}
			n^{-1}\left[  \ell^{(n)}(\hat{\theta}^{(n)})-\ell^{(n)}(\hat{\theta
			}_{\text{AMMLE}}^{(n,L)})\right]   &  \leq n^{-1}\lvert\ell^{(n)}(\hat{\theta
			}^{(n)})-\hat{\ell}^{(n,L)}(\hat{\theta}^{(n)})\rvert+n^{-1}[\hat{\ell
			}^{(n,L)}(\hat{\theta}^{(n)})-\ell^{(n)}(\hat{\theta}_{\text{AMMLE}}%
			^{(n,L)})]\\
			&  \leq n^{-1}\lvert\ell^{(n)}(\hat{\theta}^{(n)})-\hat{\ell}^{(n,L)}%
			(\hat{\theta}^{(n)})\rvert+n^{-1}[\hat{\ell}^{(n,L)}(\hat{\theta
			}_{\text{AMMLE}}^{(n,L)})-\ell^{(n)}(\hat{\theta}_{\text{AMMLE}}^{(n,L)})]\\
			&  \leq2\sup_{\theta\in\Theta}n^{-1}\lvert\ell^{(n)}(\theta)-\hat{\ell
			}^{(n,L)}(\theta)\rvert.
		\end{align*}
		Since we have shown in Corollary \ref{coro:uniform error for likelihood} that
		$\sup_{\theta\in\Theta}\lvert n^{-1}\ell^{(n)}(\theta)-n^{-1}\hat{\ell
		}^{(n,L)}(\theta)\rvert\leq R_{\ell}(L),$ we obtain $\limsup_{n\rightarrow
			\infty}\left\Vert \hat{\theta}^{(n)}-\hat{\theta}_{\text{AMMLE}}%
		^{(n,L)}\right\Vert ^{2}\leq8R_{\ell}(L)/c_{0}$ with probability tending to
		one. Since $R_{\ell}(L)=\log L/L$, if we choose the truncation level $L_{n}$
		such that $n\log L_{n}/L_{n}\rightarrow0,$ (\ref{eq:par est clt}) can be
		immediately obtained by using Slutsky's lemma.
	\end{proof}
	\subsection{Proof of Proposition \ref{prop:misspecified_mmle}}
		\begin{proof}
			Under the misspecified counterparts of the conditions of Lemmas \ref{lem:filter-ergodic} and  \ref{lemma:clt for s}, the same arguments yield the uniform convergence of the population criterion and Hessian under the true law. Since $\theta^{\ast}$ is the unique interior maximizer of $\ell^{\ast}(\theta)$, $\hat{\theta}^{(n)} \overset{p}{\rightarrow}\theta^{\ast}$. Moreover, the starred analogue of Lemma \ref{lemma:clt for s} gives
			\[
			n^{-1/2}S_n(\theta^{\ast}) \overset{d}{\rightarrow} \mathcal{N}(0,B_{\ast}),\quad -n^{-1}H_n(\theta)  \overset{p}{\rightarrow} A_{\ast}
			\]
			locally uniformly at $\theta^{\ast}$. Applying the same Taylor-expansion argument as in Theorem \ref{thm:mimle-long-span} yields the stated sandwich asymptotic variance. The arguments for the AMMLE approximation bound and its asymptotic equivalence to the MMLE are identical to those in Theorem \ref{thm:mimle-long-span} after replacing $(\ell,\theta_0)$ with $(\ell^{\ast},\theta^{\ast})$. We omit the detailed derivation for brevity.
		\end{proof}

\clearpage
\section*{Online Supplementary Material}

\setcounter{section}{0}
\renewcommand{\thesection}{S\arabic{section}}
\renewcommand{\thesubsection}{S\arabic{section}.\arabic{subsection}}

\setcounter{figure}{0}
\renewcommand{\thefigure}{S\arabic{figure}}

\setcounter{table}{0}
\renewcommand{\thetable}{S\arabic{table}}

\setcounter{equation}{0}
\renewcommand{\theequation}{S\arabic{equation}}

\section{Tools from Fourier analysis: the convergence of Fourier series
\label{sec:uni convergence fourier}}
We here discuss the convergence of Fourier series. The results in this section
are classic conclusions from Fourier analysis, which will serve as a
bridge in the subsequent proofs. Let $B=b_{1}-b_{0}$. For a function $f$
defined on $[b_{0},b_{1}]^{d}$, we consider its periodic extension with period
$B$ in each coordinate. Its Fourier coefficients are given by $c_{l}%
=B^{-d}\int_{[b_{0},b_{1}]^{d}}f(\xi)\exp\left(  -2\pi il\cdot\xi/B\right)
d\xi,\quad l\in\mathbb{Z}^{d}.$ Formally, the corresponding Fourier-series
representation is $f(x)=\sum_{l\in\mathbb{Z}^{d}}c_{l}\exp\left(  2\pi il\cdot
x/B\right)  .$ For $L\geq1$, let $D_{L}=\left\{  l\in\mathbb{Z}^{d}:\left\Vert
l\right\Vert _{\infty}\leq L\right\}  .$ The rectangularly Fej\'{e}r-adjusted
truncated Fourier series are defined by
\begin{equation}
\sigma_{L}f(x):=\sum_{l\in D_{L}}\left\{  \prod_{j=1}^{d}\left(
1-\frac{|l_{j}|}{L+1}\right)  \right\}  c_{l}\exp\left(  \frac{2\pi il\cdot
x}{B}\right)  . \label{eq:Fejer}%
\end{equation}

Uniform convergence of Fej\'{e}r means for continuous functions is classical;
see, for example, Chapter 1.3 of \cite{jackson1930theory}, Chapter V of
\cite{timan1963theory}, and Chapter XVII of \cite{zygmund2002trigonometric}.
We now record an explicit rate for the tensor-product Fej\'{e}r mean used in \eqref{eq:Fejer}.

Without loss of generality, we first take $B=1$ and work on $[-1/2,1/2]^{d}$.
Suppose that $f\in C^{1}([-1/2,1/2]^{d})$ and that $f$ vanishes on the
boundary of $[-1/2,1/2]^{d}$. Let
\[
M_{f}:=\max_{1\leq j\leq d}\sup_{x\in\lbrack-1/2,1/2]^{d}}\left\vert \partial
f(x)/\partial x_{j}\right\vert .
\]
The one-dimensional Fej\'{e}r kernel is given by $F_{L}(t)=\sum_{|m|\leq
L}\left(  1-|m|/(L+1)\right)  \exp(2\pi imt)$ or equivalently, $F_{L}%
(t)=\left(  L+1\right)  ^{-1}\left(  \sin(\pi(L+1)t)/\sin(\pi t)\right)
^{2}.$ Now \eqref{eq:Fejer} can be written as
\begin{equation}
\sigma_{L}f(x)=\int_{[-1/2,1/2]^{d}}f(x-u)K_{L}^{(d)}(u)\,du=\int%
_{[-1/2,1/2]^{d}}f(u)K_{L}^{(d)}(x-u)\,du, \label{eq:Fejer 1}%
\end{equation}
where $K_{L}^{(d)}(u)=\prod_{j=1}^{d}F_{L}(u_{j}),\quad u=(u_{1},\ldots
,u_{d})$ and is referred to as the tensor-product Fej\'{e}r kernel. Moreover,
it is easy to verify that $K_{L}^{(d)}(u)\geq0,\quad\int_{\lbrack
-1/2,1/2]^{d}}K_{L}^{(d)}(u)\,du=1.$ Therefore, we obtain
\[
\left\vert \sigma_{L}f(x)-f(x)\right\vert \leq\int_{\lbrack-1/2,1/2]^{d}%
}\left\vert f(x-u)-f(x)\right\vert K_{L}^{(d)}(u)\,du.
\]
Since $f$ is Lipschitz, i.e., $|f(x-u)-f(x)|\leq M_{f}\sum_{j=1}^{d}|u_{j}|,$
let $I_{L}:=\int_{-1/2}^{1/2}|t|F_{L}(t)\,dt,$ it follows that
\[
|\sigma_{L}f(x)-f(x)|\leq M_{f}\int_{[-1/2,1/2]^{d}}\left(  \sum_{j=1}%
^{d}|u_{j}|\right)  \prod_{r=1}^{d}F_{L}(u_{r})\,du=dM_{f}I_{L}.
\]
It remains to estimate $I_{L}$. Following Chapter 1.3 of
\cite{jackson1930theory}, we know that $I_{L}\leq C\log(L+1)/(L+1),$ and
consequently, the uniform convergence rate of Fej\'{e}r means can be written
as%
\[
\sup_{x\in\lbrack-1/2,1/2]^{d}}|\sigma_{L}f(x)-f(x)|\leq C_{d}M_{f}\frac{\log
L}{L}.
\]

Therefore, in Section~\ref{sec:truncation}, if the filtered density is
approximated by the Fej\'{e}r-adjusted truncated Fourier series, the
approximation can be interpreted as the Fej\'{e}r mean of the periodic
extension of the filtered density on $\mathcal{X}$. Although the original
filtered density need not be periodic, the resulting finite-dimensional
approximation remains a proper density on one period $\mathcal{X}$: it is
strictly positive on $\mathcal{X}$ and integrates to one for every truncation
level $L$. Hence, regardless of the truncation order, the Fej\'{e}r-adjusted
Fourier approximation produces a valid density function on the state space
$\mathcal{X}$.

\section{Accelerated implementation of the Fourier recursion
\label{sec:multi accelerate}}

It is useful to recast Proposition \ref{prop:implement} into an explicit
recursive algorithm.

\begin{algorithm}
{\small \label{alg:fourier truncation}\textbf{Implementation of Fourier
recursion system} }

{\small \textbf{Input:} Observations $\mathbf{y}_{n}$, parameter $\theta$,
candidate truncation level $L.$ }

{\small \textbf{Output:} Approximate log-likelihood $\hat{\ell}_{n}(\theta)$
and Fourier coefficients of filtered densities $\{\hat{c}_{k,l}^{(L)}
(\theta)\}_{l\in D_{L}}.$ }

{\small \textbf{Step 1. Initialization.} Set $\tilde{\ell}_{n}(\theta)=0.$
Compute the initial Fourier coefficients $\{\hat{c}_{1,l}^{(L)}(\theta
)\}_{l\in D_{L}}.$ }

{\small \textbf{Step 2. Recursive likelihood updating and filtering.} For
$k=1,\ldots,n-1$, repeat the following steps. }

\begin{enumerate}

\item {\small Compute $\mathcal{\hat{L}}_{k+1}^{(L)}(\theta)\text{ by
(\ref{eq:implement 1}). }$ }

\item {\small Update the approximate log-likelihood by $\hat{\ell}_{n}%
(\theta)\leftarrow\hat{\ell}_{n}(\theta)+\log\mathcal{\hat{L}}_{k+1}%
^{(L)}(\theta).$ }

\item {\small For each $l\in D_{L}$, update the next-step Fourier coefficient
$\hat{c}_{k+1,l}^{(L)}(\theta)$ by (\ref{eq:implement 2}). }
\end{enumerate}

{\small \textbf{Step 3. Output.} Return $\hat{\ell}_{n}(\theta)$ and
$\{\hat{c} _{k,l}^{(L)}(\theta)\}_{l\in D_{L}},\quad k=1,\ldots,n.$ }
\end{algorithm}

For multivariate latent states, the truncation set becomes $D_{L}=\left\{
l\in\mathbb{Z}^{d},\text{ }\left\Vert l\right\Vert _{\infty}\leq L\right\}  ,$
where $\left\Vert \cdot\right\Vert _{\infty}$ denotes the maximum norm.
Algorithm \ref{alg:fourier truncation} then extends directly, but the
computational complexity increases to $O\left(  n(2L+1)^{2d_{X}}\right)  ,$
where $d_{X}$ is the dimension of latent process. Because this cost grows
rapidly with the dimension of the latent state, it is useful to consider more
efficient truncation schemes. The main source of this computational burden is
not the Fourier representation itself, but the use of the full rectangular
truncation set $D_{L}$.

A natural remedy is to replace the rectangular set with sparse Fourier
truncation sets. Sparse Fourier expansions commonly use non-rectangular
truncation sets (see, e.g.,
\cite{BungartzGriebel2004SparseGrids,DungTemlyakovUllrich2018HyperbolicCross,Trefethen2017Cubature}%
). If the rectangular truncation set $D_{L}$ is replaced, a similar Fourier
recursion can be applied to multivariate latent processes directly. The first
alternative is the $\ell_{1}$-ball truncation $D_{L}^{\ell_{1}}=\left\{
l\in\mathbb{Z}^{d_{X}}:\Vert l\Vert_{1}\leq L\right\}  .$ This truncation is
naturally associated with Fourier coefficients whose decay is controlled by
the $\ell_{1}$-norm of the frequency index. More precisely, it is suitable
when the filtered density belongs to the weighted Sobolev-type space
$H_{\ell_{1}}^{s}(X)=\left\{  f:\sum_{l\in\mathbb{Z}^{d_{X}}}(1+\Vert
l\Vert_{1})^{2s}|f_{l}|^{2}<\infty\right\}  .$ Its cardinality satisfies
$|D_{L}^{\ell_{1}}|\asymp L^{d_{X}},$ so the asymptotic polynomial order
remains the same as that of the rectangular truncation. However, the constant
is substantially smaller, especially when $d_{X}$ is not small. Therefore the
resulting implementation still has the order $O(n|D_{L}^{\ell_{1}}%
|^{2})=O(nL^{2d_{X}}),$ but with a smaller leading constant.

A second alternative is the spherical $\ell_{2}$-ball truncation $D_{L}%
^{\ell_{2}}=\left\{  l\in\mathbb{Z}^{d_{X}}:\Vert l\Vert_{2}\leq L\right\}  .$
This choice is more closely related to isotropic smoothness. In particular, it
is appropriate when the target function belongs to the periodic isotropic
Sobolev space $H_{\ell_{2}}^{s}(X)=\left\{  f:\sum_{l\in\mathbb{Z}^{d_{X}}%
}(1+\Vert l\Vert_{2}^{2})^{s}|f_{l}|^{2}<\infty\right\}  .$ Since
$|D_{L}^{\ell_{2}}|\asymp L^{d_{X}},$ the computational complexity is again
$O(n\left\vert D_{L}^{\ell_{2}}\right\vert ^{2})=O\left(  nL^{2d_{X}}\right)
,$ but the number of retained coefficients is smaller than that under the
rectangular truncation with the same cutoff level $L$.

When the target function has dominating mixed smoothness, a more efficient
choice is the hyperbolic cross truncation. Define $D_{L}^{\mathrm{HC}%
}=\left\{  l\in\mathbb{Z}^{d_{X}}:\prod_{j=1}^{d_{X}}(1+|l_{j}|)\leq
L\right\}  .$ This truncation is naturally associated with the mixed Sobolev
space $H_{\mathrm{mix}}^{s}(X)=\left\{  f:\sum_{l\in\mathbb{Z}^{d_{X}}}%
\prod_{j=1}^{d_{X}}(1+|l_{j}|^{2})^{s}|f_{l}|^{2}<\infty\right\}  .$ The key
feature of the hyperbolic cross is that it keeps frequency vectors for which
not too many coordinates are simultaneously large. Therefore, it retains the
important low-order and mixed interaction terms. The cardinality of the
hyperbolic cross satisfies $\left\vert D_{L}^{\mathrm{HC}}\right\vert \asymp
L(\log L)^{d_{X}-1}.$ Consequently, the double summation in Algorithm
\ref{alg:fourier truncation} has total cost $O( n\left\vert D_{L}%
^{\mathrm{HC}}\right\vert ^{2}) =O( nL^{2}(\log L)^{2(d_{X}-1)}) .$ This is
substantially smaller than the rectangular cost $O\left(  nL^{2d_{X}}\right)
$.

\section{Specializations and illustrations of the Fourier
recursion\label{sec:illustration}}

Theorem \ref{thm:fourier series update 1} shows the Fourier updating for
likelihood evaluation, filtering, and prediction, which can be carried out
once the transform functions $H$ and $H_{X}$ are available. Therefore, the
implementation of the Fourier recursion reduces to the evaluation of these two
functions. This section details how to obtain $H$ and $H_{X}$ using the
transition law of the latent Markov model for the most widely adopted model classes.

\subsection{Finite-state LMMs}

We begin with the simplest case. Consider a finite latent state space, as in
the regime-switching Markov models in Example \ref{example:regime}, where
$\mathcal{X}=\{0,1,\ldots,M-1\}.$ In this setting, since $\nu_{X}$ is the
counting measure, all integrals with respect to latent state $X$ appearing in
$\Psi,\Psi_{X},H$ and $H_{X}$ reduce to finite sums. Namely, (\ref{eq:def psi}%
) and (\ref{eq:def psi X}) reduce to $\Psi(u,y_{1}|x_{0},y_{0};\theta
)=\sum_{m=0}^{M-1}e^{ium}p_{(X,Y)}(m,y_{1}|x_{0},y_{0};\theta)$ and $\Psi
_{X}(u|x_{0},y_{0};\theta)=\sum_{m=0}^{M-1}e^{ium}p_{X}(m|x_{0},y_{0};\theta
)$, respectively. Substituting the finite-sum representations of $\Psi$ and
$\Psi_{X}$ into the definition of their Fourier-transformed functions $H$ and
$H_{X}$ yields
\[
H(u,y_{1},s,y_{0};\theta)=\sum_{m,r\in\mathcal{X}}e^{ium+isr}p_{(X,Y)}%
(m,y_{1}|r,y_{0};\theta),\text{ }H_{X}(u,s,y_{0};\theta)=\sum_{m,r\in
\mathcal{X}}e^{ium+isr}p_{X}(m|r,y_{0};\theta).
\]

Therefore, when $\mathcal{X}$ is finite, as long as the transition density of
the LMM is explicitly known, both $H$ and $H_{X}$ can be computed directly
from their definitions. However, when $X$ is continuous-valued and the
transition density is difficult to obtain, e.g., the continuous-time models in
Example \ref{example:heston}, \ref{example:bns}, the evaluation of $H$ or
$H_{X}$ is no longer immediate, since it generally requires the computation of
Fourier-type integrals over the latent state space. Hence, it motivates us to
develop efficient evaluation methods for $H$ and $H_{X}$ for LMMs where the
latent process $X$ is continuous-valued.

\subsection{State-space models\label{sec:ssm}}

With these two extra properties (\ref{eq:SSM 1})-(\ref{eq:SSM 2}) of SSMs,
Theorem \ref{thm:fourier series update 1} can be further simplified. First of
all, the transformed functions $H$ and $H_{X}$ in (\ref{eq:def H}%
)-(\ref{eq:def H X}) no longer depend on the previous observed state $y_{0}$.
Therefore, we write $H^{\text{SSM}}(u,y_{1},s;\theta)$ and $H_{X}^{\text{SSM}%
}(u,s;\theta)$ to denote their state-space specialization. Secondly, since the
Fourier expansion for $q_{Y}\left(  y|x;\theta\right)  $ can be written as
$q_{Y}\left(  y|x;\theta\right)  =\sum_{l\in\mathbb{Z}}\varphi_{Y}\left(
y,\frac{-2\pi l}{b_{1}-b_{0}};\theta\right)  \frac{e^{2\pi ilx/(b_{1}-b_{0})}%
}{b_{1}-b_{0}}$, combined with (\ref{eq:def H X}), the transformed function
$H$ in (\ref{eq:def H}) can be represented as%
\[
H^{\text{SSM}}\left(  \frac{2\pi u}{b_{1}-b_{0}},y_{1},\frac{2\pi s}%
{b_{1}-b_{0}};\theta\right)  =\sum_{l\in\mathbb{Z}}\frac{1}{b_{1}-b_{0}%
}\varphi_{Y}\left(  y_{1},\frac{2\pi(u-l)}{b_{1}-b_{0}};\theta\right)
H_{X}^{\text{SSM}}\left(  \frac{2\pi l}{b_{1}-b_{0}},\frac{2\pi s}{b_{1}%
-b_{0}};\theta\right)  .
\]
And by using the conditional independence restriction in (\ref{eq:SSM 2}) and
the Fourier series representation of $p(x_{k}|\mathbf{y}_{k-1};\theta)$,
$\mathcal{L}_{k}\left(  \theta\right)  $ can be represented as
\begin{equation}
\mathcal{L}_{k}\left(  \theta\right)  =\int_{\mathcal{X}}q_{Y}\left(
y_{k}|x_{k};\theta\right)  p(x_{k}|\mathbf{y}_{k-1};\theta)\nu_{X}%
(dx_{k})=\sum_{l\in\mathbb{Z}}c_{k,l}^{-}\left(  \theta\right)  \varphi
_{Y}\left(  y_{k},\frac{2\pi l}{b_{1}-b_{0}};\theta\right)  ,
\label{eq:likelihood prediction update}%
\end{equation}
Similarly, by combining (\ref{eq:fourier filter update}%
)-(\ref{eq:fourier prediction update}), the Fourier coefficients $\left\{
c_{k,l}\left(  \theta\right)  \right\}  _{l\in\mathbb{Z}}$ now can be updated
from $\{c_{k,l}^{-}\left(  \theta\right)  \}_{l\in\mathbb{Z}}$ as
\[
c_{k,l}(\theta)=\sum_{r\in\mathbb{Z}}\frac{c_{k,r}^{-}(\theta)}{\mathcal{L}%
_{k}(\theta)(b_{1}-b_{0})}\varphi_{Y}\left(  y_{k},\frac{2\pi(r-l)}%
{b_{1}-b_{0}};\theta\right)  .
\]
Finally, the Fourier coefficients of the next one-step predictive density,
$\{c_{k+1,l}^{-}\left(  \theta\right)  \}_{l\in\mathbb{Z}}$, according to
(\ref{eq:fourier prediction update}), can be represented with $\left\{
c_{k,l}\left(  \theta\right)  \right\}  _{l\in\mathbb{Z}}$. Therefore,
(\ref{eq:likelihood prediction update}) and
(\ref{eq:fourier prediction update}) together constitute Corollary
\ref{thm:fourier series update 3}.

Corollary \ref{thm:fourier series update 3} is a more operational version of
Theorem \ref{thm:fourier series update 1} for SSMs by exploiting the special
structure of SSMs to reduce the evaluation of the Fourier-transformed $H$ to
the computation of the more explicit quantity $\varphi_{Y}$, thereby
substantially reducing the computational complexity of the update procedure.
Panel A of Table \ref{tab:representative-fourier-recursion} reports the
model-specific expressions of $\varphi_{Y}$ and $H_{X}^{\text{SSM}}$ for
several representative state-space models. These expressions are available
either in closed form or as Fourier transform of deterministic functions.
Although closed-form formulae are not always available, the Fourier transform
integrals can be evaluated efficiently through fast Fourier transform (FFT).
Therefore, the quantities required in Corollary
\ref{thm:fourier series update 3} are computationally tractable for a broad
class of SSMs.

\subsection{LMMs with tractable conditional characteristic
functions\label{sec:affine}}

It suffices for us to show how $H$ and $H_{X}$ can be obtained from
$\phi_{(X,Y)}$ when the latent state is continuous-valued. Recall that the
definition of $\Psi$ involves a Dirac delta function $\delta(Y_{1}-y_{1}).$ In
fact, within the distribution theory of functional analysis, the Dirac delta
function admits a rigorous series expansion representation. Namely, if we
assume that the observed process $Y_{k}$ is also compactly supported, namely,
$Y_{k}\in\lbrack a_{0},a_{1}],$ for some $a_{0}<a_{1},$ $\delta(Y_{1}-y_{1})$
can be represented by an expansion as\footnote{This identity is understood in
the distributional sense. More precisely, the series $(a_{1}-a_{0})^{-1}%
\sum_{l\in\mathbb{Z}}\exp\left\{  \frac{2\pi il(Y_{1}-y_{1})}{a_{1}-a_{0}%
}\right\}  $ is the Fourier-series representation of the Dirac delta on
$[a_{0},a_{1}]$. That is, for any sufficiently regular test function $f$,
$\mathbb{E}\left[  f\left(  Y_{1}\right)  \delta(Y_{1}-y_{1})\right]
=(a_{1}-a_{0})^{-1}\mathbb{E}\left[  f\left(  Y_{1}\right)  \sum_{l\in D}%
\exp\left(  \frac{2\pi il\left(  Y_{1}-y_{1}\right)  }{a_{1}-a_{0}}\right)
\right]  .$ For a more detailed treatment, see, e.g., Chapters 6--7 of
\cite{rudin1991functional}.}
\begin{equation}
\delta(Y_{1}-y_{1})=\frac{1}{a_{1}-a_{0}}\sum_{l\in D}\exp\left(  \frac{2\pi
il\left(  Y_{1}-y_{1}\right)  }{a_{1}-a_{0}}\right)  ,
\label{eq:expansion of dirac}%
\end{equation}
where $D$ is a set of integers that depends on $Y_{k}$. When $Y_{k}$ is
continuous-valued, $D=\mathbb{Z}$ and the above expansion becomes the standard
Fourier series of the Dirac delta function. When $Y_{k}$ is discrete-valued,
say $Y_{k}\in\left\{  0,1,2,...,N-1\right\}  ,$ then $D=\left\{
0,1,2,...,N-1\right\}  ,$ and (\ref{eq:expansion of dirac}) turns into
$\delta(Y_{1}-y_{1})=N^{-1}\sum_{l=0}^{N-1}\exp\left(  2\pi il\left(
Y_{1}-y_{1}\right)  /N\right)  ,$ Hence, by the Fourier representation of the
Dirac delta function in (\ref{eq:expansion of dirac}), $\Psi$ defined in
(\ref{eq:def psi}) can be written as%
\[
\Psi(u,y_{1}|x_{0},y_{0};\theta)=\sum_{l\in D}\frac{\exp\left(  -2\pi
ily_{1}/\left(  a_{1}-a_{0}\right)  \right)  }{a_{1}-a_{0}}\phi_{(X,Y)}\left(
\left.  u,\frac{2\pi l}{a_{1}-a_{0}}\right\vert x_{0},y_{0};\theta\right)  ,
\]
and by the definition in (\ref{eq:def psi X}), $\Psi_{X}$ can be represented
by $\phi_{(X,Y)}$ as
\begin{equation}
\Psi_{X}(u|x_{0},y_{0};\theta)=\phi_{(X,Y)}\left(  u,0|x_{0},y_{0}%
;\theta\right)  . \label{eq:Psi X repre}%
\end{equation}
To express $H$ and $H_{X}$ in a unified way, let $\Phi\left(  u,v,y_{0}%
,s;\theta\right)  :=\int_{\mathcal{X}}\exp(isx_{0})\phi_{(X,Y)}\left(
u,v|x_{0},y_{0};\theta\right)  dx_{0}$ denote the Fourier transform of
$\phi_{(X,Y)}$ w.r.t. previous latent state $x_{0}.$ Combining
(\ref{eq:Psi X repre}), transform functions $H$ and $H_{X}$ can be represented
as follows:%
\begin{align}
H(u,y_{1},s,y_{0};\theta)  &  =\sum_{l\in D}\frac{\exp\left(  -2\pi
ily_{1}/\left(  a_{1}-a_{0}\right)  \right)  }{a_{1}-a_{0}}\Phi\left(
u,\frac{2\pi l}{a_{1}-a_{0}},y_{0},s;\theta\right)  ,\label{eq:H repre}\\
H_{X}(u,s,y_{0};\theta)  &  =\Phi\left(  u,0,y_{0},s;\theta\right)  .\nonumber
\end{align}
Thus, the evaluation of $H$ and $H_{X}$ is reduced to the evaluation of the
single auxiliary function $\Phi,$ which directly leads to Corollary
\ref{thm:fourier series update 2}. Even outside the affine class, $\Phi\left(
u,v,y_{0},s;\theta\right)  $ can also be computed either analytically or
accurately approximated since it is the Fourier transform of $\phi
_{(X,Y)}\left(  u,v|x_{0},y_{0};\theta\right)  $ w.r.t. $x_{0}$.\footnote{In
numerical implementation, after discretizing $\mathcal{X}$ with an appropriate
grid, this Fourier transform can be efficiently evaluated using the fast
Fourier transform (FFT), thereby substantially reducing the computational
burden. Similarly, when $H$ or $H_{X}$ can not be obtained analytically, FFT
can also be employed to compute these two functions efficiently, since both
can be represented as Fourier transforms of certain functions.}%
\begin{table}[ptbh]
{\normalsize \centering
\setlength{\tabcolsep}{3pt} \renewcommand{\arraystretch}{1.7}
\begin{threeparttable}
			\caption{Representative Fourier recursions for LMMs}
			\label{tab:representative-fourier-recursion}
			\vspace{0.4em}
			\begin{tabular}{@{} m{0.15\textwidth} m{0.22\textwidth} m{0.36\textwidth} m{0.21\textwidth} @{}}
				\toprule
				\multicolumn{4}{@{}l}{\text{Panel A: Fourier recursion for SSMs in Section \ref{sec:ssm}}} \\
				\midrule
				Model class & Representatives & $\phi_Y(y,u;\theta)$ & $H_X^{\rm SSM}(u,s;\theta)$ \\
				\midrule
				\multirow{2}{=}[-0em]{Linear structure}&
				DSV model in (\ref{eq:discrete sv}) &
				$\displaystyle \frac{e^{iuy}2^{-iu}}{\Gamma(1/2)}\Gamma\left(1/2-iu\right)$ &
				$\displaystyle e^{i\alpha u-\frac{\sigma^2u^2}{2}}{\cal J}(s+\beta u)$ \\
				\cmidrule{2-4}
				&
				Factor ARCH in (\ref{eq:factor arch}) &
				$\displaystyle
				|\Lambda|^{-1}
				e^{\frac{iuy}{\Lambda}
					-\frac{\sigma_\varepsilon^2u^2}{2\Lambda^2}}$ &
				$\displaystyle
				e^{-\frac{\omega u^2}{2}-\frac{s^2}{2a u^2}}
				\sqrt{\frac{2\pi}{a u^2}}$ \\
				\midrule
				\multirow{2}{=}[-0.2em]{Nonlinear structure} &
				QTSM in (\ref{eq:qtsm}) &
				$\displaystyle \int_{b_0}^{b_1}p_\epsilon(y-A-Bx-Cx^2)e^{iux}\,dx$ &
				$\displaystyle \phi_\eta(\Sigma u)e^{i\mu u}{\cal J}(s+\Phi u)$ \\
				\cmidrule{2-4}
				&
				SRNN in (\ref{eq:SRNN}) &
				$\displaystyle \int_{b_0}^{b_1}p_\epsilon(y-g_\theta(x))e^{iux}\,dx$ &
				$\displaystyle \phi_\eta(u)\int_{b_0}^{b_1}e^{isx+iuf_\theta(x)}\,dx$ \\
				\midrule
				\multirow{2}{=}[-0em]{Dynamic GLMs} &
				Binary SSM in (\ref{eq:logit}) &
				$\displaystyle  e^{-iu\alpha}\Delta_Y^B(y+iu,1-y-iu) $&
				$\displaystyle \phi_\eta(u){\cal J}(s+Au)$ \\
				\cmidrule{2-4}
				&
				Poisson SSM in (\ref{eq:poisson model}) &
				$\displaystyle \frac{c^{-iu}}{y!} \Delta_Y^G\!\left(y+iu;ce^{b_1},ce^{b_0}\right)$ &
				$\displaystyle e^{-\frac{1}{2}\sigma^2u^2}{\cal J}(s+\gamma_1u)$ \\
			\end{tabular}
			\begin{tabular}{@{}
					m{0.18\textwidth}
					m{0.28\textwidth}
					>{\centering\arraybackslash}m{0.48\textwidth}
					@{}}
				\midrule
				\multicolumn{3}{@{}l}{\text{Panel B: Fourier recursion for models with tractable $\phi_{(X,Y)}$ in Section \ref{sec:affine}}} \\
				\midrule
				Model class & Representatives & \hspace{1.2em}$\Phi(u,v,y_0,s;\theta)$ \\
				\midrule
				\multirow{3}{=}[-0em]{Affine models}
				& DPS models in Example \ref{example:heston}
				& \multirow{3}{*}{%
					\parbox{0.48\textwidth}{\centering
						$\displaystyle
						e^{C_0(u,v,y_0)}
						\frac{
							e^{A_s(u,v,y_0)b_1}
							-
							e^{A_s(u,v,y_0)b_0}
						}{
							A_s(u,v,y_0)
						}$}} \\
				\cmidrule{2-2}
				& BNS models in Example \ref{example:bns} & \\
				\cmidrule{2-2}
				& Hawkes process in Example \ref{example:hawkes} & \\
				\midrule
			\end{tabular}
			\begin{tablenotes}
				\small
				\item \textit{Notes:} For Panel A, $X_k\in[b_0,b_1]$, $c=e^\alpha$, $\ell_\alpha(x)=ce^x/(1+ce^x)$, and ${\cal J}(v)=\int_{b_0}^{b_1}e^{ivx}\,dx$. The expressions are written on compact support $[b_0,b_1]$. The endpoint-difference notations are $\Delta_Y^G(a;r_1,r_0)=\gamma(a,r_1)-\gamma(a,r_0)$ and $\Delta_Y^B(a,b)=B_{\ell_\alpha(b_1)}(a,b)-B_{\ell_\alpha(b_0)}(a,b)$, where $\gamma(a,z)=\int_0^z t^{a-1}e^{-t}dt$ is the lower incomplete Gamma function and $B_z(a,b)=\int_0^z t^{a-1}(1-t)^{b-1}dt$ is the incomplete Beta function. For Panel B, $A_s(u,v,y_0)=is+C_1(u,v,y_0)$. The numerical integrals above can be efficiently computed by FFT.
			\end{tablenotes}
		\end{threeparttable}
}\end{table}\newcolumntype{Y}{>{\RaggedRight\arraybackslash}X}

\section{Auxiliary assumptions, lemmas and proofs\label{sec:auxiliary proof}}

\subsection{Proof of Proposition \ref{prop:smoothing}\label{sec:smoothing}}

\begin{proof}
To do smoothing, a general strategy is to consider a backward recursion. Given
filtered density $p\left(  x_{k}|\mathbf{y}_{k};\theta\right)  $ and the
previous smoothed density $p\left(  x_{k+1}|\mathbf{y}_{n};\theta\right)  ,$
by applying the Bayes' rule, the smoothed density $p\left(  x_{k}%
|\mathbf{y}_{n};\theta\right)  $ can be written as%
\begin{equation}
p\left(  x_{k}|\mathbf{y}_{n};\theta\right)  =\int_{\mathcal{X}}p\left(
x_{k}|x_{k+1},\mathbf{y}_{n};\theta\right)  p\left(  x_{k+1}|\mathbf{y}%
_{n};\theta\right)  \nu_{X}(dx_{k+1}), \label{eq:smooth 1}%
\end{equation}
where $p\left(  x_{k}|x_{k+1},\mathbf{y}_{n};\theta\right)  :=\mathbb{E}%
\left[  \delta\left(  X_{k}-x_{k}\right)  |X_{k+1}=x_{k+1},\mathbf{Y}%
_{n}=\mathbf{y}_{n};\theta\right]  .$ Here the backward conditional density
$p\left(  x_{k}|x_{k+1},\mathbf{y}_{n};\theta\right)  $ can be further
simplified. First, applying Bayes' rule to $p\left(  x_{k}|x_{k+1}%
,\mathbf{y}_{n};\theta\right)  $ gives
\begin{equation}
p\left(  x_{k}|x_{k+1},\mathbf{y}_{n};\theta\right)  =\frac{p\left(
\mathbf{y}_{k+1:n}|x_{k+1},x_{k},\mathbf{y}_{k};\theta\right)  p\left(
x_{k}|x_{k+1},\mathbf{y}_{k};\theta\right)  }{p\left(  \mathbf{y}%
_{k+1:n}|x_{k+1},\mathbf{y}_{k};\theta\right)  }. \label{eq:smooth 1.1}%
\end{equation}
Under SSM properties (\ref{eq:SSM 1}) and (\ref{eq:SSM 2}), once $x_{k+1}$ is
given, the future observations $\mathbf{y}_{k+1:n}$ are independent of both
$x_{k}$ and the past observations $\mathbf{y}_{k}$, i.e., $p\left(
\mathbf{y}_{k+1:n}|x_{k+1},x_{k},\mathbf{y}_{k};\theta\right)  =p\left(
\mathbf{y}_{k+1:n}|x_{k+1};\theta\right)  ,$ and $p\left(  \mathbf{y}%
_{k+1:n}|x_{k+1},\mathbf{y}_{k};\theta\right)  =p\left(  \mathbf{y}%
_{k+1:n}|x_{k+1};\theta\right)  .$ According to these simplified terms, we
notice that (\ref{eq:smooth 1.1}) can be written as
\[
p\left(  x_{k}|x_{k+1},\mathbf{y}_{n};\theta\right)  =\frac{p\left(
\mathbf{y}_{k+1:n}|x_{k+1};\theta\right)  p\left(  x_{k}|x_{k+1}%
,\mathbf{y}_{k};\theta\right)  }{p\left(  \mathbf{y}_{k+1:n}|x_{k+1}%
;\theta\right)  }=p\left(  x_{k}|x_{k+1},\mathbf{y}_{k};\theta\right)  ,
\]
and substituting the above equation into (\ref{eq:smooth 1}), we have%
\begin{equation}
p\left(  x_{k}|\mathbf{y}_{n};\theta\right)  =\int_{\mathcal{X}}q_{X}\left(
x_{k+1}|x_{k};\theta\right)  p\left(  x_{k}|\mathbf{y}_{k};\theta\right)
\frac{p\left(  x_{k+1}|\mathbf{y}_{n};\theta\right)  }{p\left(  x_{k+1}%
|\mathbf{y}_{k};\theta\right)  }\nu_{X}(dx_{k+1}), \label{eq:smooth 2}%
\end{equation}
by using (\ref{eq:SSM 1}). Let $r_{k}(x_{k};\theta):=p\left(  x_{k}%
|\mathbf{y}_{n};\theta\right)  /p\left(  x_{k}|\mathbf{y}_{k};\theta\right)  $
with $r_{n}(x;\theta)=1,$ then by using (\ref{eq:SSM 2}), (\ref{eq:smooth 2})
turns into
\begin{equation}
r_{k}(x_{k})=\int_{\mathcal{X}}\frac{q_{X}\left(  x_{k+1}|x_{k};\theta\right)
q_{Y}\left(  y_{k+1}|x_{k+1};\theta\right)  }{\mathcal{L}_{k+1}(\theta
)}r_{k+1}(x_{k+1};\theta)\nu_{X}(dx_{k+1}). \label{eq:smooth ratio update}%
\end{equation}
This implies that the density ratio $r_{k}(x_{k};\theta)$ can be recursively
evaluated backwards, i.e., recovering $r_{k}$ from $r_{k+1}$. To handle
$r_{k},$ following the similar Fourier-series representation of filtered
density, $r_{k}(x_{k};\theta)$ can also be expanded by Fourier series as
$r_{k}(x_{k};\theta)=\sum_{l\in\mathbb{Z}}\bar{b}_{k,l}\left(  \theta\right)
\exp\left(  \frac{2\pi ilx_{k}}{b_{1}-b_{0}}\right)  .$ By plugging the
Fourier series of $r_{k}(x_{k})$ into (\ref{eq:smooth ratio update}), the
associated Fourier coefficients $\left\{  \bar{b}_{k,l}\left(  \theta\right)
\right\}  _{l\in\mathbb{Z}}$ can be obtained as:
\begin{equation}
\bar{b}_{k,l}\left(  \theta\right)  =\sum_{s\in\mathbb{Z}}\frac{\bar
{b}_{k+1,s}\left(  \theta\right)  }{(b_{1}-b_{0})\mathcal{L}_{k+1}(\theta
)}H\left(  \frac{2\pi s}{b_{1}-b_{0}},y_{k+1},\frac{-2\pi l}{b_{1}-b_{0}%
},y_{k};\theta\right)  . \label{eq:smooth 3}%
\end{equation}
Finally, we express the smoothed density $p\left(  x_{k}|\mathbf{y}_{n}%
;\theta\right)  $ by a Fourier-series expansion $p\left(  x_{k}|\mathbf{y}%
_{n};\theta\right)  =\sum_{l\in\mathbb{Z}}b_{k,l}\left(  \theta\right)
\exp\left(  \frac{2\pi ilx_{k}}{b_{1}-b_{0}}\right)  .$ Since the smoothed
density is given by the product of the backward smoothing density ratio
$r_{k}(x_{k})$ and the filtered density $p\left(  x_{k}|\mathbf{y}_{k}%
;\theta\right)  $, namely, $p\left(  x_{k}|\mathbf{y}_{n};\theta\right)
=r_{k}(x_{k})p\left(  x_{k}|\mathbf{y}_{k};\theta\right)  ,$ the Fourier
coefficients of the smoothed density are obtained through the convolution of
the corresponding Fourier coefficients $b_{k,l}\left(  \theta\right)
=\sum_{j+s=l}\bar{b}_{k,s}(\theta)c_{k,j}(\theta).$
\end{proof}

\subsection{Proof of Lemma \ref{prop:one step error}\label{sec:one-step error}%
}

This subsection collects the auxiliary assumptions and technical lemmas used
to establish the uniform error control result for the truncated Fourier
recursion. The assumptions below impose regularity conditions on the
transition, the likelihood update, and the truncation scheme, ensuring that
the approximation error generated at each filtering step remains uniformly
controlled over the relevant parameter space. The subsequent lemmas provide
the main intermediate bounds, including positivity and stability of the
likelihood update, one-step truncation error control, and recursive
propagation of the approximation error. Together, these results deliver the
uniform approximation bound stated in Section \ref{sec:uniform bound} of the
main text.

Before moving on, we need to review some notation in \cite{gland2004}. For any
probability measure $\mu$ on $\mathcal{X}$, denote the transition operator
$R_{k}$ and its normalized version $\bar{R}_{k}$ as
\begin{align*}
R_{k}\mu(dx_{1})  &  =\int_{\mathcal{X}}p_{(X,Y)}(x_{1},y_{k}|x_{0}%
,y_{k-1};\theta)\mu(dx_{0})\,dx_{1},\\
\bar{R}_{k}\mu(dx_{1})  &  =\frac{\int_{\mathcal{X}}p_{(X,Y)}(x_{1}%
,y_{k}|x_{0},y_{k-1};\theta)\mu(dx_{0})}{\int_{\mathcal{X}}p_{X}(x_{k}%
|x_{0},y_{k-1};\theta)\mu(dx_{0})}\,dx_{1}.
\end{align*}
In fact, under Assumption~\ref{assump: bound derivative}, the operator $R_{k}$
is contraction under the Hilbert metric with the Birkhoff contraction
coefficient satisfying
\[
\tau(R_{k}):=\sup_{0<h(\mu,\mu^{\prime})<\infty}\frac{h(R_{k}\mu,R_{k}%
\mu^{\prime})}{h(\mu,\mu^{\prime})}\leq\frac{1-\epsilon_{k}^{2}}%
{1+\epsilon_{k}^{2}},
\]
(see Lemma~3.8 and Proposition~3.9 in \cite{gland2004}). For convenience,
assume $\{\epsilon_{k}\}_{k\geq1}$ is bounded away from zero; then $\tau
(R_{k})$ is uniformly smaller than 1. Based on the assumptions above, before
proving the positivity of the filtered density, we first provide an auxiliary lemma.

\begin{lemma}
\label{lemma: likelihood bound}Under Assumption~\ref{assump: bound derivative}%
, the approximated likelihood $\mathcal{\hat{L}}_{k}^{(L_{k-1})}(\theta)$ is
uniformly bounded as $\epsilon\leq\mathcal{\hat{L}}_{k}^{(L)}(\theta
)\leq\epsilon^{-1},$ where $\epsilon$ is the mixing coefficient in
Assumption~\ref{assump: bound derivative}.
\end{lemma}

\begin{proof}
Recall that
\[
\mathcal{\hat{L}}_{k}^{(L)}(\theta)=\int_{\mathcal{X}}p_{Y}(y_{k}%
|x_{k-1},y_{k-1};\theta)\hat{p}^{(L)}(x_{k-1}|\mathbf{y}_{k-1};\theta
)dx_{k-1},
\]
where $p_{Y}(y_{k}|x_{k-1},y_{k-1};\theta)$ is continuous with respect to
$y_{k}$ and satisfies
\[
p_{Y}(y_{k}|x_{k-1},y_{k-1};\theta)=\int_{\mathcal{X}}p_{(X,Y)}(x_{k}%
,y_{k}|x_{k-1},y_{k-1};\theta)dx_{k}\in\left[  \epsilon_{k},\frac{1}%
{\epsilon_{k}}\right]  .
\]
Combined with the fact that $\int_{\mathcal{X}}\hat{p}^{(L)}(x_{k-1}%
|\mathbf{y}_{k-1};\theta)dx_{k-1}=1,$ the result can be immediately obtained.
\end{proof}

We denote $\bar{p}^{(L)}(x|\mathbf{y}_{k};\theta)$ as the untruncated filtered
density at time $k$ updated from the truncated filtered density $\hat{p}%
^{(L)}(x_{k-1}|\mathbf{y}_{k-1};\theta),$ i.e.,
\[
\bar{p}^{(L)}(x|\mathbf{y}_{k};\theta)=\left(  \mathcal{\hat{L}}_{k}%
^{(L)}(\theta)\right)  ^{-1}\int_{\mathcal{X}}p_{(X,Y)}(x_{k},y_{k}%
|x_{k-1},y_{k-1};\theta)\hat{p}^{(L)}(x_{k-1}|\mathbf{y}_{k-1};\theta
)\,dx_{k-1}.
\]
We have shown that both $\mathcal{\hat{L}}_{k}^{(L)}(\theta)$ and $\hat
{p}^{(L)}(x_{k+1}|\mathbf{y}_{k+1};\theta)$ keep positive, and this gives us
validate a definition of the approximated log-likelihood $\hat{\ell}_{k}%
^{(L)}(\theta)=\log\mathcal{\hat{L}}_{k}^{(L)}(\theta)$ and the prerequisite
to use the forgetting property of filtering (see, e.g., \cite{gland2004}).
First, we present the one-step error bounds. Now we are able to give the proof
of Lemma \ref{prop:one step error}.

\noindent\textit{Proof of Lemma \ref{prop:one step error}. }Since $\hat{p}%
^{(L)}(x|\mathbf{y}_{k};\theta)$ is the Fej\'{e}r-adjusted Fourier series for
$\bar{p}^{(L)}(x|\mathbf{y}_{k};\theta),$ by using the uniform convergence
result in \ref{sec:uni convergence fourier}, it holds
\begin{align}
\sup_{x\in\mathcal{X}}\lvert\bar{p}^{(L)}(x|\mathbf{y}_{k};\theta)-\hat
{p}^{(L)}(x|\mathbf{y}_{k};\theta)\rvert &  \leq\sup_{x\in\mathcal{X}%
}\left\vert \frac{\partial\bar{p}^{(L)}(x|\mathbf{y}_{k};\theta)}{\partial
x}\right\vert \frac{(b_{1}-b_{0})^{d_{X}}C_{d_{X}}\log L}{L}\nonumber\\
&  \leq\frac{C_{1}}{\epsilon}\frac{(b_{1}-b_{0})^{d_{X}}C_{d_{X}}\log L}{L}.
\label{eq:pf one step 1}%
\end{align}
By Assumption~\ref{assump: bound derivative} and Lemma
\ref{lemma: likelihood bound}, we have
\begin{align*}
\left\vert \frac{\partial\bar{p}^{(L_{1})}(x|\mathbf{y}_{2};\theta)}{\partial
x}\right\vert  &  =\left\vert \frac{\int_{\mathcal{X}}\partial p_{(X,Y)}%
(x,y_{2}|x_{1},y_{1};\theta)/\partial x\times\hat{p}^{(L)}(x_{1}|y_{1}%
;\theta)dx_{1}}{\mathcal{\hat{L}}_{2}^{(L)}(\theta)}\right\vert \\
&  \leq\frac{1}{\mathcal{\hat{L}}_{2}^{(L)}(\theta)}\int_{\mathcal{X}%
}\left\vert \frac{\partial p_{(X,Y)}(x,y_{2}|x_{1},y_{1};\theta)}{\partial
x}\right\vert \hat{p}^{(L)}(x_{1}|y_{1};\theta)dx_{1}%
\end{align*}
Then, combining the above inequality with (\ref{eq:pf one step 1}), we obtain
$\sup_{x\in\mathcal{X}}\lvert\bar{p}^{(L)}(x|\mathbf{y}_{k};\theta)-\hat
{p}^{(L)}(x|\mathbf{y}_{k};\theta)\rvert\leq\tilde{C}_{d_{X}}\log L /\left(
\epsilon L\right)  ,$ for some constant $\tilde{C}_{d_{X}}.$ Thus, by
integrating with respect to $x\in\mathcal{X}$, for any $f$ with $\left\Vert
f\right\Vert _{\infty}=1$, it holds
\[
\left\vert \int_{\mathcal{X}}\left\vert \bar{p}^{(L_{k-1})}(x|\mathbf{y}%
_{k};\theta)-\hat{p}^{(L_{k})}(x|\mathbf{y}_{k};\theta)\right\vert
f(x)dx\right\vert \leq(b_{1}-b_{0})^{d_{X}}\sup_{x\in\mathcal{X}}\lvert\bar
{p}^{(L_{k-1})}(x|\mathbf{y}_{k};\theta)-\hat{p}^{(L_{k})}(x|\mathbf{y}%
_{k};\theta)\rvert.
\]

\subsection{Proof of Lemma \ref{lem:filter-ergodic}\label{sec:ergodicity}}

This subsection provides the auxiliary assumptions and lemmas used to prove
Lemma \ref{lem:filter-ergodic}. These conditions are formulated as regularity
requirements on the latent Markov process, the observation density, and the
derivatives of the filtered likelihood, and are intended to support the
consistency and asymptotic normality results stated in the main text. The
lemmas below establish the required stability, ergodicity, score central limit
theorem, and Hessian convergence results. These ingredients are then combined
to justify the limiting behavior of the estimator under the Fourier-based
filtering approximation. Let \begin{align}
	\mathcal{S}
	&=
	\overline{\left\{
		p(\cdot\mid\mathbf{y}_{k};\theta):
		k\geq0,\ 
		\mathbf{y}_{k}=(y_{0},\ldots,y_{k})\in\mathcal{Y}^{k+1},\
		\theta\in\Theta
		\right\}},
	\label{eq:S_state}
	\\
	\dot{\mathcal{S}}
	&=
	\dot{\mathcal{S}}_{1}\times\cdots\times\dot{\mathcal{S}}_{d},
	\quad
	\dot{\mathcal{S}}_{j}
	=
	\overline{\left\{
		\frac{\partial p(\cdot\mid\mathbf{y}_{k};\theta)}
		{\partial\theta_{j}}:
		k\geq0,\
		\mathbf{y}_{k}=(y_{0},\ldots,y_{k})\in\mathcal{Y}^{k+1},\
		\theta\in\Theta
		\right\}},
	\label{eq:Sdot_state}
\end{align} 
where the closures in (\ref{eq:S_state})--(\ref{eq:Sdot_state}) are taken with respect to the uniform norm
$\|f\|_{\infty}:=\sup_{x\in\mathcal{X}}|f(x)|$. The preceding assumptions play two roles. Assumptions \ref{ass:S2}%
--\ref{ass:S5} ensure stability, regularity, and identification of the latent
Markov model and the filtered likelihood, while Assumption
\ref{assump:hessian bound} provides the local concavity needed for
likelihood-based inference to ensure the.

\begin{assumption}
\label{ass:S2}There exist constants $C_{(X,Y)}>0$ and $r_{(X,Y)}\in(0,1)$ such
that, for any $(x_{0},y_{0})\in\mathcal{X}\times\mathcal{Y}$,
\[
\int_{\mathcal{X}\times\mathcal{Y}}\left\vert p_{(X,Y)}(k,\tilde{x},\tilde
{y}|x_{0},y_{0};\theta_{0})-\pi_{(X,Y)}(\tilde{x},\tilde{y};\theta
_{0})\right\vert \,d\tilde{x}\,d\tilde{y}\leq C_{(X,Y)}r_{(X,Y)}^{k},
\]
where $p_{(X,Y)}(k,\tilde{x},\tilde{y}|x_{0},y_{0};\theta_{0})$ denotes the
$k$ steps transition density, i.e.,
\[
p_{(X,Y)}(k,\tilde{x},\tilde{y}|x_{0},y_{0};\theta_{0})=\mathbb{E}\left[
\delta\left(  X_{k}-\tilde{x}\right)  \delta\left(  Y_{k}-\tilde{y}\right)
|X_{0}=x_{0},Y_{0}=y_{0},\theta_{0}\right]  .
\]

\end{assumption}

\begin{assumption}
\label{ass:S4}The transition density $p_{(X,Y)}$ satisfies the following
conditions. There exist constants $\epsilon_{(X,Y)}\in(0,1)$ and $K_{(X,Y)}%
>0$, and a function $v_{(X,Y)}:\mathcal{X}\times\mathcal{Y}\times
\Theta\rightarrow(0,\infty),$ such that, for any $(x,y),(x_{0},y_{0}%
)\in\mathcal{X}\times\mathcal{Y}$, $\theta\in\Theta$, and $f_{0}\in
\mathcal{S}$,
\[
\epsilon_{(X,Y)}v_{(X,Y)}(x,y;\theta)\leq p_{(X,Y)}(x,y|x_{0},y_{0}%
;\theta)\leq\epsilon_{(X,Y)}^{-1}v_{(X,Y)}(x,y;\theta),
\]
and, for any observed next value $y$ and observed current value $y_{0}$,
\[
\int_{\mathcal{X}}\left\Vert \int_{\mathcal{X}}\frac{\partial p_{(X,Y)}%
(\tilde{x},y|\tilde{x}_{0},y_{0};\theta)}{\partial\theta}f_{0}(\tilde{x}%
_{0})\,d\tilde{x}_{0}\right\Vert d\tilde{x}\leq K_{(X,Y)}\int_{\mathcal{X}%
^{2}}p_{(X,Y)}(\tilde{x},y|\tilde{x}_{0},y_{0};\theta)f_{0}(\tilde{x}%
_{0})\,d\tilde{x}_{0}d\tilde{x}.
\]

\end{assumption}

\begin{assumption}
\label{ass:S5}The density of the initial state $(X_{0},Y_{0})$ is continuous
in $(x_{0},y_{0},\theta)$. The equality $p(y_{k+1}|\mathbf{y}_{k}%
;\theta)=p(y_{k+1}|\mathbf{y}_{k};\theta_{0})$ holds almost everywhere in
$y_{k+1}\in\mathcal{Y}$ if and only if $\theta=\theta_{0}$. Moreover,
derivatives of $p(y_{k+1}|\mathbf{y}_{k};\theta),$ $p_{Y}(y|x_{0},y_{0}%
;\theta)$ with respect to $\theta$ up to the third order, are uniformly
bounded over $k\geq0$, $\mathbf{y}_{k}\in\mathcal{Y}^{k+1}$, $y_{k+1}%
,y_{0},y\in\mathcal{Y}$, $x\in\mathcal{X}$ and $\theta\in\Theta$. For any $k\geq0$, $\mathbf{y}_{k}\in\mathcal{Y}^{k+1}$, and
$\theta\in\Theta$, the first- and second-order filter derivatives
$\partial p(\cdot|\mathbf{y}_{k};\theta)/\partial\theta$ and
$\partial^{2}p(\cdot|\mathbf{y}_{k};\theta)/
\partial\theta\partial\theta^{\intercal}$ are uniformly bounded in
$L^{1}(\mathcal{X})$ and uniformly Lipschitz in $x$.
\end{assumption}

\begin{assumption}
\label{assump:hessian bound}The Hessian matrix of $\ell(\theta)$ satisfies
$-\partial^{2}\ell(\theta)/\partial\theta\partial\theta^{\top}\geq c_{0}I,$
$\forall\theta\in B\left(  \theta_{0},\delta\right)  $ a.s., where $B\left(
\theta_{0},\delta\right)  =\left\{  \theta\text{ }|\text{ }\left\Vert
\theta-\theta_{0}\right\Vert \leq\delta\right\}  .$
\end{assumption}

The first step is to augment the latent state with the filtered density and
its derivative. This converts the recursive filtering system into a Markov
process on an enlarged state space, which allows standard ergodic arguments to
be applied to functions of the filtered likelihood. We now show the proof of
Lemma \ref{lem:filter-ergodic}.

\noindent\textit{Proof of Lemma \ref{lem:filter-ergodic}. }The proof consists of three
steps. In the first step, we prove the Markov property of the processes
$\xi_{k}(\theta)=(X_{k},Y_{k},p(\cdot|\mathbf{Y}_{k};\theta),\partial
p(\cdot|\mathbf{Y}_{k};\theta)/\partial\theta{)},$ $\xi_{k}^{+}(\theta
)=\left(  Y_{k+1},\xi_{k}(\theta)\right)  $ and $\left(  X_{k+1},Y_{k+1}%
,\xi_{k}(\theta)\right)  .$ In the second step, we prove the ergodicity of
these processes. In the third step, we prove the distributional properties
(i), (ii) and (iii). Since the proofs of the ergodic Markov properties of
$\xi_{k}^{+}(\theta)$ and $(X_{k+1},Y_{k+1},\xi_{k}(\theta))$ follow the same
procedure as that of $\xi_{k}(\theta)$, in the first two steps we concentrate
on the case of $\xi_{k}(\theta)$ and omit the repetitive proofs for $\xi
_{k}^{+}(\theta)$ and $(X_{k+1},Y_{k+1},\xi_{k}(\theta))$. As a preparation,
we introduce two operators to describe the one-step transition of the filter
and the filter derivative, respectively. Indeed, it follows from the Bayes
updating system that the filter satisfies $p(x_{k+1}|\mathbf{Y}_{k+1}%
;\theta)=\Gamma\left[  Y_{k+1},Y_{k},p(\cdot|\mathbf{Y}_{k};\theta)\right]
(x_{k+1}),$ and the operator $\Gamma$ is defined as follows. For any
$(y_{1},y_{0},f_{0})\in\mathcal{Y}\times\mathcal{Y}\times\mathcal{S}$,
\[
\Gamma\lbrack y_{1},y_{0},f_{0}](x):={\frac{\int_{\mathcal{X}}%
p_{(X,Y)}(x,y_{1}|\widetilde{x}_{0},y_{0};\theta)f_{0}(\widetilde{x}%
_{0})d\widetilde{x}_{0}}{\int_{\mathcal{X}}p_{Y}(y_{1}|\widetilde{x}_{0}%
,y_{0};\theta)f_{0}(\widetilde{x}_{0})d\widetilde{x}_{0}}}.
\]
In particular, $\Gamma$ can be interpreted as a mapping from $\mathcal{Y}%
\times\mathcal{Y}\times\mathcal{S}$ to $\mathcal{S}$, and $\Gamma$ describes
how the filter $p(\cdot|\mathbf{Y}_{k+1};\theta)$ depends on the observations
$(Y_{k+1},Y_{k})$ and the filter $p(\cdot|\mathbf{Y}_{k};\theta)$ at time $k$,
that is,
\begin{equation}
p(\cdot|\mathbf{Y}_{k+1};\theta)=\Gamma\left[  Y_{k+1},Y_{k},p(\cdot
|\mathbf{Y}_{k};\theta)\right]  . \label{eq:S.36}%
\end{equation}
Then, differentiating both sides with respect to $\theta$ leads to
\begin{equation}
{\frac{\partial p(x_{k+1}|\mathbf{Y}_{k+1};\theta)}{\partial\theta}}%
=\Gamma^{\prime}\left[  Y_{k+1},Y_{k},p(\cdot|\mathbf{Y}_{k};\theta
),{\frac{\partial p(\cdot|\mathbf{Y}_{k};\theta)}{\partial\theta}}\right]
(x_{k+1}), \label{eq:S.38}%
\end{equation}
where the operator $\Gamma^{\prime}$ is defined as follows. For any
$(y_{1},y_{0},f_{0},g_{0})\in\mathcal{Y}\times\mathcal{Y}\times\mathcal{S}%
\times\mathcal{\dot{S}},$ we define%
\begin{align*}
\Gamma^{\prime}[y_{1},y_{0},f_{0},g_{0}](x)  &  :={}{\frac{\int_{\mathcal{X}%
}{\frac{\partial p_{(X,Y)}}{\partial\theta}}(x,y_{1}|\widetilde{x}_{0}%
,y_{0};\theta)f_{0}(\widetilde{x}_{0})d\widetilde{x}_{0}+\int_{\mathcal{X}%
}p_{(X,Y)}(x,y_{1}|\widetilde{x}_{0},y_{0};\theta)g_{0}(\widetilde{x}%
_{0})d\widetilde{x}_{0}}{\int_{\mathcal{X}}p_{Y}(y_{1}|\widetilde{x}_{0}%
,y_{0};\theta)f_{0}(\widetilde{x}_{0})d\widetilde{x}_{0}}}\\
&  -\Gamma\lbrack y_{1},y_{0},f_{0}](x){\frac{\int_{\mathcal{X}}%
{\frac{\partial p_{Y}}{\partial\theta}}(y_{1}|\widetilde{x}_{0},y_{0}%
;\theta)f_{0}(\widetilde{x}_{0})d\widetilde{x}_{0}+\int_{\mathcal{X}}%
p_{Y}(y_{1}|\widetilde{x}_{0},y_{0};\theta)g_{0}(\widetilde{x}_{0}%
)d\widetilde{x}_{0}}{\int_{\mathcal{X}}p_{Y}(y_{1}|\widetilde{x}_{0}%
,y_{0};\theta)f_{0}(\widetilde{x}_{0})d\widetilde{x}_{0}}.}%
\end{align*}
Similar to $\Gamma$, the operator $\Gamma^{\prime}$ can be interpreted as a
mapping from $\mathcal{Y}\times\mathcal{Y}\times\mathcal{S}\times
\mathcal{\dot{S}}$ to $\mathcal{\dot{S}}$, and $\Gamma^{\prime}$ describes how
the filter derivative $\partial p(\cdot|\mathbf{Y}_{k+1};\theta)/\partial
\theta$ depends on the observations $(Y_{k+1},Y_{k})$, the filter
$p(\cdot|\mathbf{Y}_{k};\theta)$ and the filter derivative $\partial
p(\cdot|\mathbf{Y}_{k};\theta)/\partial\theta$ at time $k$, that is,
\[
{\frac{\partial p(\cdot|\mathbf{Y}_{k+1};\theta)}{\partial\theta}}%
=\Gamma^{\prime}\left[  Y_{k+1},Y_{k},p(\cdot|\mathbf{Y}_{k};\theta
),{\frac{\partial p(\cdot|\mathbf{Y}_{k};\theta)}{\partial\theta}}\right]  .
\]

\textit{Step 1. Proving the Markov property.} To prove the Markov property of
$\xi_{k}(\theta)$, it suffices to show that
\[
\mathbb{E}\left[  F(\xi_{k+1}(\theta))|\xi_{k}(\theta)=\xi_{k},\ldots,\xi
_{0}(\theta)=\xi_{0}\right]  =\mathbb{E}\left[  F(\xi_{k+1}(\theta))|\xi
_{k}(\theta)=\xi_{k}\right]  ,
\]
\ for any $F\in C_{b}(\Xi)$, where $\xi_{k}=(x_{k},y_{k},f_{k},g_{k})\in\Xi,$
and $\Xi=\mathcal{X}\times\mathcal{Y}\times\mathcal{S}\times\mathcal{\dot{S}%
},$ for $k=0,1,\ldots$. Indeed, by the updating formulae for the filter and
the filter derivative, we have%
\begin{align*}
&  F(\xi_{k+1}(\theta))\\
&  ={}F\left(  X_{k+1},Y_{k+1},p(\cdot|\mathbf{Y}_{k+1};\theta),{\frac
{\partial p(\cdot|\mathbf{Y}_{k+1};\theta)}{\partial\theta}}\right) \\
&  ={}F\left(  X_{k+1},Y_{k+1},\Gamma\left[  Y_{k+1},Y_{k},p(\cdot
|\mathbf{Y}_{k};\theta)\right]  ,\Gamma^{\prime}\left[  Y_{k+1},Y_{k}%
,p(\cdot|\mathbf{Y}_{k};\theta),{\frac{\partial p(\cdot|\mathbf{Y}_{k}%
;\theta)}{\partial\theta}}\right]  \right)
\end{align*}
We define $G(x_{k+1},y_{k+1},x,y,f,g):=F\left(  x_{k+1},y_{k+1},\Gamma\lbrack
y_{k+1},y,f],\Gamma^{\prime}[y_{k+1},y,f,g]\right)  $. Then we have
$\mathbb{E}\left[  F(\xi_{k+1}(\theta))|\xi_{k}(\theta)=\xi_{k}\right]
=\mathbb{E}\left[  G(X_{k+1},Y_{k+1},\xi_{k})|\xi_{k}(\theta)=\xi_{k}\right]
,$ and%
\begin{align*}
&  \mathbb{E}\left[  F(\xi_{k+1}(\theta))|\xi_{k}(\theta)=\xi_{k},\xi
_{k-1}(\theta)=\xi_{k-1},\ldots,\xi_{0}(\theta)=\xi_{0}\right] \\
&  =\mathbb{E}\left[  G(X_{k+1},Y_{k+1},\xi_{k})|\xi_{k}(\theta)=\xi_{k}%
,\xi_{k-1}(\theta)=\xi_{k-1},\ldots,\xi_{0}(\theta)=\xi_{0}\right]
\end{align*}
On the one hand, note that each element in $\xi_{k}(\theta)=\left(
X_{k},Y_{k},p(\cdot|\mathbf{Y}_{k};\theta),\partial p(\cdot|\mathbf{Y}%
_{k};\theta){/}\partial\theta\right)  $ is adapted to the filtration generated
by $\{(X_{j},Y_{j})\}_{j=1}^{k}.$ Since $\xi_{k}(\theta)$ contains $\left(
X_{k},Y_{k}\right)  ,$ by the Markov property of $(X_{k},Y_{k})$, we obtain
\begin{equation}
\mathbb{E}\left[  G(X_{k+1},Y_{k+1},\xi_{k})|\xi_{k}(\theta)=\xi_{k}\right]
=\mathbb{E}\left[  G(X_{k+1},Y_{k+1},\xi_{k})|(X_{k},Y_{k})=(x_{k}%
,y_{k})\right]  . \label{eq:S.41}%
\end{equation}
Similarly, as $\xi_{k}(\theta),\xi_{k-1}(\theta),\ldots,\xi_{0}(\theta)$ are
also adapted, by the Markov property of $(X_{k},Y_{k})$, we obtain%
\[
\mathbb{E}\left[  G(X_{k+1},Y_{k+1},\xi_{k})|\xi_{k}(\theta)=\xi_{k}%
,\ldots,\xi_{0}(\theta)=\xi_{0}\right]  =\mathbb{E}\left[  G(X_{k+1}%
,Y_{k+1},\xi_{k})|(X_{k},Y_{k})=(x_{k},y_{k})\right]
\]
Therefore, the Markov property of $\xi_{k}(\theta)$ can be immediately
obtained by%
\[
\mathbb{E}\left[  F(\xi_{k+1}(\theta))|\xi_{k}(\theta)=\xi_{k},\ldots,\xi
_{0}(\theta)=\xi_{0}\right]  =\mathbb{E}\left[  F(\xi_{k+1}(\theta))|\xi
_{k}(\theta)=\xi_{k}\right]
\]
\noindent\ 

\textit{Step 2. Proving the ergodicity.} As a preparation, we introduce the
function class $G_{Lip}$, which collects all functions satisfying that $G\in
C_{b}(\Xi)$ and there exists a constant $C_{G}>0$ such that, for any
$(x_{0},y_{0})\in\mathcal{X}\times\mathcal{Y},\quad f_{0},f_{1}\in
\mathcal{S},\quad g_{0},g_{1}\in\mathcal{\dot{S}},$ it holds
\[
\left\vert G(x_{0},y_{0},f_{0},g_{0})-G(x_{0},y_{0},f_{1},g_{1})\right\vert
\leq C_{G}\left(  \Vert f_{0}-f_{1}\Vert_{L_{1}}+\Vert g_{0}-g_{1}\Vert
_{L_{1}}\right)  .
\]
We name this property as the $\mathcal{S}\times\mathcal{\dot{S}}$-Lipschitz
property. Then, we use three substeps to prove the ergodicity of $\xi
_{k}(\theta)$. First, we show that proving the ergodicity of $\xi_{k}(\theta)$
boils down to proving that there exist a random variable $\xi(\theta)$ and a
constant $r\in(0,1)$ such that, for any $G\in G_{Lip}$, $\xi_{0}\in\Xi$, and
$k>0$, it holds
\begin{equation}
\left\vert \mathbb{E}\left[  G(\xi_{k}(\theta))|\xi_{0}(\theta)=\xi
_{0}\right]  -\mathbb{E}\left[  G(\xi(\theta))\right]  \right\vert \leq
K_{G}r^{k}, \label{eq:S.43}%
\end{equation}
where $K_{G}$ is a constant depending only on $G$. Second, we show that
proving the inequality (\ref{eq:S.43}) boils down to proving that, for any
$\xi_{0},\overline{\xi}_{0}\in\Xi$, it holds
\begin{equation}
\left\vert \mathbb{E}\left[  G(\xi_{k}(\theta))|\xi_{0}(\theta)=\xi
_{0}\right]  -\mathbb{E}\left[  G(\xi_{k}(\theta))|\xi_{0}(\theta
)=\overline{\xi}_{0}\right]  \right\vert \leq K_{G}r^{k}, \label{eq:S.44}%
\end{equation}
where $K_{G}$ and $r$ are the constants in the inequality (\ref{eq:S.43}).
Third, we prove this last inequality.

\textit{Step 2.1. Reducing the proof of ergodicity to proving (\ref{eq:S.43}%
).} As long as we have $\left\vert \mathbb{E}\left[  G(\xi_{k}(\theta
))|\xi_{0}(\theta)=\xi_{0}\right]  -\mathbb{E}\left[  G(\xi(\theta))\right]
\right\vert \leq K_{G}r^{k},$ we can prove the ergodicity of $\xi_{k}(\theta)$
in the following way. First, we prove that the law of $\xi(\theta)$ is the
unique stationary distribution of $\xi_{k}(\theta)$. To prove this, as
discussed in the definition of ergodicity, it suffices to show that
$\lim_{k\rightarrow\infty}\mathbb{E}\left[  G(\xi_{k}(\theta))\right]
=\mathbb{E}\left[  G(\xi(\theta))\right]  $ for any $G\in C_{b}(\Xi)$ and for
any initial distribution of $\xi_{0}(\theta)$. Note that $\mathbb{E}\left[
G(\xi_{k}(\theta))\right]  =\mathbb{E}\left\{  \mathbb{E}\left[  G(\xi
_{k}(\theta))|\xi_{0}(\theta)\right]  \right\}  $, it follows from
(\ref{eq:S.43}) that
\begin{equation}
\lim_{k\rightarrow\infty}\mathbb{E}\left[  G(\xi_{k}(\theta))\right]
=\mathbb{E}\left[  G(\xi(\theta))\right]  \label{eq:S.45}%
\end{equation}
for any $G\in G_{Lip}$ and for any initial distribution of $\xi_{0}(\theta)$.
Thus, to prove the same convergence for any $G\in C_{b}(\Xi)$, it suffices to
show that $G_{Lip}$ is dense in $C_{b}(\Xi)$ as we have already proved in
Lemma \ref{lemma:lip condi xi_k}. Second, we prove that $n^{-1}\sum
_{k=0}^{n-1}G(\xi_{k}(\theta))\rightarrow\mathbb{E}[G(\xi(\theta))],$ a.s.,
for any $G\in C_{b}(\Xi)$. Indeed, since $G_{Lip}$ is dense in $C_{b}(\Xi)$
under the supremum norm, it suffices for us to show that $n^{-1}\sum
_{k=0}^{n-1}G(\xi_{k}(\theta))\rightarrow\mathbb{E}[G(\xi(\theta))]$ for any
$G\in G_{Lip}$. We follow the standard route based on the following Poisson
equation for $V$ (see, e.g., Section 17.4 of \cite{meyn2009markov}) in order
to use the martingale functional CLT:
\begin{equation}
V(\xi)-TV(\xi)=G(\xi)-\mathbb{E}[G(\xi(\theta))], \label{eq:S.47}%
\end{equation}
for any $\xi\in\Xi$, where the operator $T$ is defined by $TV(\xi
_{0})=\mathbb{E}[V(\xi_{1}(\theta))|\xi_{0}(\theta)=\xi_{0}].$ By
calculations, we obtain that the function
\begin{equation}
V_{0}(\xi_{0})=\sum_{k=0}^{\infty}\left(  \mathbb{E}\left[  G(\xi_{k}%
(\theta))|\xi_{0}(\theta)=\xi_{0}\right]  -\mathbb{E}[G(\xi(\theta))]\right)
\label{eq:S.48}%
\end{equation}
is a solution for the Poisson equation in (\ref{eq:S.47}). The function
$V_{0}$ is well-defined because the above series converges as $\left\vert
\mathbb{E}\left[  G(\xi_{k}(\theta))|\xi_{0}(\theta)=\xi_{0}\right]
-\mathbb{E}[G(\xi(\theta))]\right\vert \leq K_{G}r^{k}$ according to
(\ref{eq:S.43}). In addition, (\ref{eq:S.43}) also implies that $V_{0}$ is
bounded. Using $V_{0}$ in (\ref{eq:S.48}), we rewrite the summation as
\[
\frac{1}{n}\sum_{k=0}^{n-1}G(\xi_{k}(\theta))-\mathbb{E}[G(\xi(\theta
))]=\frac{1}{n}\sum_{k=0}^{n-1}\left(  V_{0}(\xi_{k}(\theta))-TV_{0}(\xi
_{k-1}(\theta))\right)  +\frac{1}{n}\left(  V_{0}(\xi_{0}(\theta))-V_{0}%
(\xi_{n}(\theta))\right)  .
\]
Note that $n^{-1}\left(  V_{0}(\xi_{0}(\theta))-V_{0}(\xi_{n}(\theta))\right)
\rightarrow0$ a.s., since $V_{0}$ is bounded. Thus, to prove convergence, it
suffices to show that
\begin{equation}
\frac{1}{n}\sum_{k=0}^{n-1}\left(  V_{0}(\xi_{k}(\theta))-TV_{0}(\xi
_{k-1}(\theta))\right)  \rightarrow0,\text{ a.s..} \label{eq:pf 1.3}%
\end{equation}
Denote by $\mathcal{F}_{k}^{\xi}$ the filtration generated by the process
$\{\xi_{k}(\theta)\}$. Then, by definition of $T$ and the Markov property, we
have $TV_{0}(\xi_{k-1}(\theta))=\mathbb{E}[V_{0}(\xi_{k}(\theta))|\mathcal{F}%
_{k-1}^{\xi}].$ Thus, the summation is a martingale difference sequence. By
standard results (e.g., Theorem 2.19 of \cite{hall1980martingale}), we obtain
the convergence in (\ref{eq:pf 1.3}).

\textit{Step 2.2. Reducing the proof of (\ref{eq:S.43}) to prove
(\ref{eq:S.44}).} Let $T_{k}G\left(  \xi_{0}\right)  =\mathbb{E}[G(\xi
_{k}(\theta))|\xi_{0}(\theta)=\xi_{0}]$ for any $G\in G_{Lip}$. By the tower
property of conditional expectation and the time-homogeneous Markov property,
we have $\mathbb{E}[G(\xi_{k+\ell}(\theta))|\xi_{0}(\theta)=\xi_{0}%
]=\mathbb{E}\left[  \mathbb{E}[G(\xi_{k+i}(\theta))|\xi_{k}(\theta)]|\xi
_{0}(\theta)=\xi_{0}\right]  =\mathbb{E}\left[  T_{i}G\left(  \xi_{k}\right)
|\xi_{0}(\theta)=\xi_{0}\right]  $ \noindent for any $k>0$ and $i\geq0$. Then,
as long as we have (\ref{eq:S.44}), we can prove (\ref{eq:S.43}) in the
following way.

First, (\ref{eq:S.44}) implies that
\begin{equation}
\left\vert \mathbb{E}[G(\xi_{k+i}(\theta))|\xi_{0}(\theta)]-\mathbb{E}%
[G(\xi_{i}(\theta))]\right\vert =\left\vert \mathbb{E}\left[  T_{i}G(\xi
_{k}(\theta))-T_{i}G(\xi)|\xi_{0}(\theta)\right]  \right\vert \leq K_{G}r^{k}.
\label{eq:S.49}%
\end{equation}
So, the sequence $\{\mathbb{E}[G(\xi_{k}(\theta))|\xi_{0}(\theta)=\xi
_{0}]\}_{k\geq0}$ is a Cauchy sequence for any $G\in G_{Lip}$, implying the
existence of $\lim_{k\rightarrow\infty}\mathbb{E}[G(\xi_{k}(\theta))|\xi
_{0}(\theta)=\xi_{0}].$ Moreover, (\ref{eq:S.44}) implies
\[
\lim_{k\rightarrow\infty}\mathbb{E}[G(\xi_{k}(\theta))|\xi_{0}(\theta)=\xi
_{0}]=\lim_{k\rightarrow\infty}\mathbb{E}[G(\xi_{k}(\theta))|\xi_{0}%
(\theta)=\xi_{0}^{\prime}]
\]
for any $\xi_{0},\xi_{0}^{\prime}\in\Xi$. Hence the limit does not depend on
the initial state. We denote this limit as $\Pi(G)$, which only depends on
$G\in G_{Lip}$. Then, letting $i\rightarrow\infty$ in (\ref{eq:S.49}), we
obtain
\begin{equation}
\left\vert \Pi(G)-\mathbb{E}[G(\xi_{k}(\theta))|\xi_{0}(\theta)=\xi
_{0}]\right\vert \leq K_{G}r^{k}. \label{eq:S.50}%
\end{equation}
Since $\mathbb{E}[G(\xi_{k}(\theta))]=\mathbb{E}[\mathbb{E}[G(\xi_{k}%
(\theta))|\xi_{0}(\theta)]],$ we obtain from (\ref{eq:S.50}) that $\left\vert
\Pi(G)-\mathbb{E}[G(\xi_{k}(\theta))]\right\vert \leq K_{G}r^{k},$ which
implies that $\lim_{k\rightarrow\infty}\mathbb{E}[G(\xi_{k}(\theta))]=\Pi(G)$
for any $G\in G_{Lip}$. As discussed previously, since the state space $\Xi$
is compact, by the Prohorov theorem (e.g., see \cite{billingsley1995probability}), there
exists a subsequence $\{\xi_{k_{j}}(\theta)\}$ and a random variable
$\xi(\theta)$ such that $\xi_{k_{j}}(\theta)$ converges weakly to $\xi
(\theta)$. Thus, $\lim_{k\rightarrow\infty}\mathbb{E}[G(\xi_{k}(\theta
))]=\mathbb{E}[G(\xi(\theta))]$ for any $G\in G_{Lip}$. Since this limit is
equal to $\Pi(G)$, we conclude that $\Pi(G)=\mathbb{E}[G(\xi(\theta))]\quad
$for all $G\in G_{Lip}.$ And (\ref{eq:S.44}) can be shown according to Lemma
\ref{lemma:lip condi xi_k}.

Step 3: Proving the distributional properties (i), (ii), and (iii): We will
use the fact that, for any Markov processes $R_{i}$ and $(R_{i},S_{i})$, if
the law of $(R,S)$ is a stationary distribution of $(R_{i},S_{i})$, then by
definition of stationarity, the law of $R$ must be a stationary distribution
of $R_{i}$. By definition, we have that $(X^{+},Y^{+},\xi(\theta))$ follows
the stationary distribution of $(X_{(k+1)},Y_{(k+1)},\xi_{k}(\theta)).$ Then,
using the above-mentioned fact and the uniqueness of the stationary
distributions of the processes $(X_{k},Y_{k}),\quad\xi_{k}(\theta),\quad
$and$\quad\xi_{k}^{+}(\theta)=(X_{k+1},Y_{k+1},\xi_{k}(\theta)),$ we obtain
that $(X,Y,\xi(\theta))$ and $\xi^{+}(\theta)=(X^{+},Y^{+},\xi(\theta))$
follow the stationary distributions of $(X_{k},Y_{k}),\quad\xi_{k}%
(\theta),\quad$and$\quad\xi_{k}^{+}(\theta),$ respectively. Then, to prove
(\ref{eq:32})-(\ref{eq:33}), it suffices to show that $\mathbb{E}%
[F(X_{k+1},Y_{k+1})|\xi_{k}(\theta)=\xi_{0}]=\mathbb{E}[F(X_{k+1}%
,Y_{k+1})|(X_{k},Y_{k})=(x_{0},y_{0})],$ for any $F\in C_{b}(\mathcal{X}%
\times\mathcal{Y})$. This follows from equality (\ref{eq:S.41}) in Step 1,
which is used to prove the Markov property of $\xi_{k}(\theta)$.

\subsection{Proof of Lemma \ref{lemma:mimle 1}\label{sec:mimle}}

To prove Lemma \ref{lemma:mimle 1}, we first give two auxiliary lemmas that
characterize exponential forgetting property, Lipschitz properties and
asymptotic normality of augmented Markov process. We will use the same
assumptions as in \ref{sec:ergodicity}.

\begin{lemma}
\label{lemma:lip condi xi_k}Under the same conditions as in Lemma
\ref{lem:filter-ergodic}, let $\Xi=\mathcal{X}\times\mathcal{Y}\times
\mathcal{S}\times\mathcal{\dot{S}}$, $C_{b}(\Xi)$ as the collection of
continuous bounded function over $\Xi,$ and $G_{Lip}$ collects every $G\in
C_{b}(\Xi)$ such that there exists a constant $C_{G}>0$ , for any
$(x_{0},y_{0})\in\mathcal{X}\times\mathcal{Y},\qquad f_{0},f_{1}\in
\mathcal{S},\qquad g_{0},g_{1}\in\mathcal{\dot{S}},$ it holds
\[
\left\vert G(x_{0},y_{0},f_{0},g_{0})-G(x_{0},y_{0},f_{1},g_{1})\right\vert
\leq C_{G}\left(  \Vert f_{0}-f_{1}\Vert_{L_{1}}+\Vert g_{0}-g_{1}\Vert
_{L_{1}}\right)  .
\]
Then, for any $\xi_{0},\overline{\xi}_{0}\in\Xi$, we have
\begin{equation}
\left\vert \mathbb{E}\left[  G(\xi_{k}(\theta))|\xi_{0}(\theta)=\xi
_{0}\right]  -\mathbb{E}\left[  G(\xi_{k}(\theta))|\xi_{0}(\theta
)=\overline{\xi}_{0}\right]  \right\vert \leq K_{G}r^{k},
\end{equation}
for any $G\in G_{Lip}.$ Moreover, we claim that $G_{Lip}$ is dense in
$C_{b}(\Xi).$
\end{lemma}

\begin{proof}
See \ref{sec:aux lemma1}.
\end{proof}

\begin{lemma}
\label{lemma:clt for s}Under the same conditions in Theorem
\ref{thm:mimle-long-span}, there exists a matrix $I_{S},$ such that
\[
n^{-1/2}\left(  \sum_{k=0}^{n-1}S(\xi_{k}^{+};\theta_{0})\right)
\overset{d}{\longrightarrow}N(0,I_{S}),
\]
where $\xi_{k}^{+}(\theta)$ is ergodic Markov as shown in Lemma
\ref{lem:filter-ergodic}. Let $H_{n}(\theta)=\partial^{2}\ell_{n}%
(\theta)/\partial\theta\partial\theta^{\intercal},$ there also exists a
matrix-valued function $H_{0}(\theta)$ such that $n^{-1}H_{n}(\theta
)\overset{p}{\longrightarrow}H_{0}(\theta)$ uniformly holds for $\theta
\in\Theta$, and $H_{0}(\theta)$ is Lipschitz in $\theta$, i.e., there exists a
constant $C_{H}>0$ such that $\left\Vert H_{0}(\theta)-H_{0}(\theta^{\prime
})\right\Vert \leq C_{H}\left\Vert \theta-\theta^{\prime}\right\Vert .$ \ 
\end{lemma}

\begin{proof}
See \ref{sec:aux lemma2}.
\end{proof}

\noindent\textit{Proof of Lemma \ref{lemma:mimle 1}. }We begin by expressing the
log-likelihood as $\ell^{(n)}(\theta)=\sum_{k=0}^{n-1}\ell(\xi_{k}^{+}%
(\theta);\theta),$ with
\[
\ell(\xi_{0}^{+};\theta)=\log\int_{\mathcal{X}}p_{Y}(y_{1}|y_{0},x_{0}%
;\theta)f_{0}(x_{0})dx_{0},
\]
for any $\xi_{0}^{+}=(y_{1},y_{0},x_{0},f_{0}(\cdot),g_{0}(\cdot
))\in\mathcal{Y}\times\mathcal{Y}\times\mathcal{X}\times\mathcal{X}%
\times\mathcal{S}\times\mathcal{\dot{S}}$, where
\[
\xi_{k}^{+}(\theta)=(Y_{k+1},Y_{k},X_{k},p(\cdot|\mathbf{Y}_{k};\theta
),\partial p(\cdot|\mathbf{Y}_{k};\theta)/\partial\theta)
\]
is ergodic Markov as we proved in Lemma \ref{lem:filter-ergodic}. By
Assumption \ref{ass:S5}, we have that $\ell(\xi_{0}^{+};\theta)$ is a
continuous and bounded function of $\xi_{0}^{+}$. Thus, by the property of
ergodic process, we obtain
\begin{equation}
n^{-1}\ell^{(n)}(\theta)\overset{a.s.}{\longrightarrow}\ell(\theta),
\label{eq:likelihood convergence}%
\end{equation}
with $\ell(\theta):=\mathbb{E}[\ell(\xi^{+}(\theta);\theta)],$ where $\xi
^{+}(\theta)$ follows the stationary distribution of $\xi_{i}^{+}(\theta)$,
and thus $n^{-1}\ell^{(n)}(\theta)\overset{p}{\longrightarrow}\ell(\theta).$

We now prove (i). Since we already have the pointwise convergence $n^{-1}%
\ell^{(n)}(\theta)\overset{p}{\longrightarrow}\ell(\theta)$ for any $\theta
\in\Theta$ and the parameter space $\Theta$ is compact, it suffices to show
the Lipschitz property of $n^{-1}\ell^{(n)}(\theta)$ in $\theta$, i.e., there
exists $C_{\ell}>0$ such that $\left\vert n^{-1}\ell^{(n)}(\theta)-n^{-1}%
\ell^{(n)}(\theta^{\prime})\right\vert \leq C_{\ell}\left\Vert \theta
-\theta^{\prime}\right\Vert .$ Then, the uniform convergence follows from the
standard finite covering argument in analysis, i.e., the compactness of
$\Theta$ guarantees that $\Theta$ can always be covered by a finite number of
balls with arbitrarily small radius and thus the uniform convergence can be
proven by combining the results in each ball, using the Lipschitz property of
$n^{-1}\ell^{(n)}(\theta)$. We now prove the Lipschitz property of $n^{-1}%
\ell^{(n)}(\theta)$. By calculations we have $\ell(\xi_{k}^{+}(\theta
);\theta)=\log p(Y_{k+1}|\mathbf{Y}_{k};\theta)$. By Assumption \ref{ass:S5},
we obtain that there exists $C_{\ell}>0$ such that $\left\vert \log
p(Y_{k+1}|\mathbf{Y}_{k};\theta)-\log p(Y_{k+1}|\mathbf{Y}_{k};\theta^{\prime
})\right\vert \leq C_{\ell}\left\Vert \theta-\theta^{\prime}\right\Vert $
holds almost surely for any $k>0$. Thus, we obtain
\begin{equation}
\left\vert n^{-1}\ell^{(n)}(\theta)-n^{-1}\ell^{(n)}(\theta^{\prime
})\right\vert \leq n^{-1}\cdot nC_{\ell}\left\Vert \theta-\theta^{\prime
}\right\Vert _{1}=C_{\ell}\left\Vert \theta-\theta^{\prime}\right\Vert .
\label{eq:pf lip l_n}%
\end{equation}

To prove (ii), i.e., $\ell(\theta)<\ell(\theta_{0})$ for any $\theta\in\Theta$
and $\theta\neq\theta_{0}$, we use the standard argument based on
Kullback-Leibler divergence. To start, we choose the initial state $\xi
_{0}^{+}(\theta)$ to follow the stationary distribution of $\xi_{i}^{+}%
(\theta)$. Then, we have $\ell(\theta)=\mathbb{E}[\ell(\xi^{+}(\theta
);\theta)]=\mathbb{E}[\ell(\xi_{k}^{+}(\theta);\theta)],$ for any $k\geq0.$ By
the tower property of conditional expectation, since $p(Y_{k+1}|\mathbf{Y}%
_{k};\theta_{0})$ is the conditional density of $Y_{k+1}$ given $\mathbf{Y}%
_{k}$, we obtain
\[
\ell(\theta)-\ell(\theta_{0})=\mathbb{E}\left[  \int_{\mathcal{Y}}\left(
\log\frac{p(y_{k+1}|\mathbf{Y}_{k};\theta)}{p(y_{k+1}|\mathbf{Y}_{k}%
;\theta_{0})}\right)  p(y_{k+1}|\mathbf{Y}_{k};\theta_{0})dy_{k+1}\right]  .
\]
Here, the right-hand side is the negative KL divergence between $p(y_{k+1}%
|\mathbf{Y}_{k};\theta)$ and $p(y_{k+1}|\mathbf{Y}_{k};\theta_{0})$ and under
Assumption \ref{ass:S5}, by Lemma 5.35 in Section 5.5 of
\cite{vanderVaart1998}, it is strictly negative for any $\theta\in\Theta$ and
$\theta\neq\theta_{0}$. Thus, we conclude that $\ell(\theta)<\ell(\theta_{0})$
for any $\theta\in\Theta$ and $\theta\neq\theta_{0}$.

To prove (iii), we will use the additive form of the score function, i.e.,
$S_{n}(\theta_{0})=\sum_{k=0}^{n-1}S(\xi_{k}^{+};\theta_{0}),$ where $\xi
_{k}^{+}(\theta)$ is ergodic Markov as we proved in Lemma
\ref{lem:filter-ergodic}. The detailed verification of the central limit
theorem for this additive score process is deferred to Lemma
\ref{lemma:clt for s}. Second, Lemma \ref{lemma:clt for s} $n^{-1}H_{n}%
(\theta)\overset{p}{\longrightarrow}H_{0}(\theta)$ uniformly holds for
$\theta\in\Theta$ and the Lipschitz property of $H_{0}(\theta)$.

Last, we use the Bartlett identities (see, e.g.,
\cite{bartlett1953a,bartlett1953b}), to prove $I_{S}=-I_{H}=I^{M}$. We first
prove $-H_{0}(\theta_{0})=I^{M},$ and then prove $I_{S}=-H_{0}(\theta_{0}).$
Since $H_{0}(\theta_{0})=\mathbb{E}[\mathbb{H}(\tilde{\xi}^{+}(\theta
_{0});\theta_{0})],$ where $\tilde{\xi}^{+}(\theta_{0})$ follows the
stationary distribution of $\tilde{\xi}_{k}^{+}(\theta_{0})$ and $\mathbb{H}$
is defined in (\ref{eq:S.62}), we obtain $H_{0}(\theta_{0})=\lim
_{k\rightarrow\infty}\mathbb{E}[\mathbb{H}(\tilde{\xi}_{k}^{+}(\theta
_{0});\theta_{0})].$ As discussed after (\ref{eq:S.63}), we have
$\mathbb{H}(\tilde{\xi}_{k}^{+}(\theta_{0});\theta_{0})=\frac{\partial^{2}%
}{\partial\theta\partial\theta^{\intercal}}\log p(Y_{k+1}|\mathbf{Y}%
_{k};\theta_{0}).$ Next, by the conditional version of the Bartlett
identities, we obtain
\begin{equation}
\mathbb{E}\left[  -\frac{\partial^{2}}{\partial\theta\partial\theta
^{\intercal}}\log p(Y_{k+1}|\mathbf{Y}_{k};\theta_{0})\right]  =\mathbb{E}%
\left[  \frac{\partial}{\partial\theta}\log p(Y_{k+1}|\mathbf{Y}_{k}%
;\theta_{0})\frac{\partial}{\partial\theta^{\intercal}}\log p(Y_{k+1}%
|\mathbf{Y}_{k};\theta_{0})\right]  . \label{eq:S.65}%
\end{equation}
As discussed after (\ref{eq:S.60}), we have
\begin{equation}
S(\xi_{k}^{+}(\theta_{0});\theta_{0})=\frac{\partial}{\partial\theta}\log
p(Y_{k+1}|\mathbf{Y}_{k};\theta_{0}), \label{eq:fisher information}%
\end{equation}
and using (\ref{eq:S.65}) we obtain $-H_{0}(\theta_{0})=\lim_{k\rightarrow
\infty}\mathbb{E}[S(\xi_{k}^{+}(\theta_{0});\theta_{0})S^{\intercal}(\xi
_{k}^{+}(\theta_{0});\theta_{0})].$ By the same route as that for proving
(\ref{eq:S.45}), we have $\lim_{k\rightarrow\infty}\mathbb{E}[G(\xi_{k}%
^{+}(\theta))]=\mathbb{E}[G(\xi^{+}(\theta))]$ for any $G\in C_{b}%
(\mathcal{Y}\times\Xi)$. Then, by Assumption \ref{ass:S5}, we have that
$S(\xi_{k}^{+}(\theta_{0});\theta_{0})S^{\intercal}(\xi_{k}^{+}(\theta
_{0});\theta_{0})$ is a continuous bounded function of $\xi_{k}^{+}$. So, we
obtain
\[
-H_{0}(\theta_{0})=\lim_{k\rightarrow\infty}\mathbb{E}[S(\xi_{k}^{+}%
(\theta_{0});\theta_{0})S^{\intercal}(\xi_{k}^{+}(\theta_{0});\theta
_{0})]=\mathbb{E}[S(\xi^{+}(\theta_{0});\theta_{0})S^{\intercal}(\xi
^{+}(\theta_{0});\theta_{0})]=I^{M}.
\]
We now prove $I_{S}=-H_{0}(\theta_{0}).$ Since $n^{-1/2}S_{n}(\theta
_{0})\overset{d}{\longrightarrow}N(0,I_{S})$ as proved previously, we obtain
by the definition of weak convergence that $I_{S}=\lim_{n\rightarrow\infty
}n^{-1}\mathbb{E}\left[  S_{n}(\theta_{0})S_{n}^{\intercal}(\theta
_{0})\right]  .$ By the definition of score $S_{n}(\theta_{0})$ and Hessian
$H_{n}(\theta_{0})$, and the conditional version of the Bartlett identities,
we have $\mathbb{E}\left[  S_{n}(\theta_{0})S_{n}^{\intercal}(\theta
_{0})\right]  =-\mathbb{E}[H_{n}(\theta_{0})],$ and thus we obtain
\[
I_{S}=-\lim_{n\rightarrow\infty}n^{-1}\mathbb{E}[H_{n}(\theta_{0}%
)]=-\lim_{n\rightarrow\infty}\frac{1}{n}\sum_{k=0}^{n-1}\mathbb{E}\left[
\mathbb{H}(\tilde{\xi}_{k}^{+}(\theta_{0});\theta_{0})\right]  .
\]
Then, as discussed prior to (\ref{eq:S.63}), we have $\lim_{k\rightarrow
\infty}\mathbb{E}\left[  \mathbb{H}(\tilde{\xi}_{k}^{+}(\theta_{0});\theta
_{0})\right]  =\mathbb{E}\left[  \mathbb{H}(\tilde{\xi}^{+}(\theta_{0}%
);\theta_{0})\right]  =H_{0}(\theta_{0}).$ Thus, the average of the sequence
$\left\{  \mathbb{E}\left[  \mathbb{H}(\tilde{\xi}_{k}^{+}(\theta_{0}%
);\theta_{0})\right]  \right\}  _{k=0}^{+\infty}$ also converges to
$H_{0}(\theta_{0})$, leading to $I_{S}=-H_{0}(\theta_{0}).$

\subsection{Proof of auxiliary Lemma \ref{lemma:lip condi xi_k}%
\label{sec:aux lemma1}}

\begin{proof}
To \textit{prove the inequality (\ref{eq:S.44}),} for ease of exposition, we
define the following operators iteratively based on $\Gamma$ and
$\Gamma^{\prime}$. Let $\Gamma_{1}=\Gamma,\quad\Gamma_{1}^{\prime}%
=\Gamma^{\prime}.$ Then we define
\begin{align}
\Gamma_{k+1}[y_{k+1:1},y_{0},f_{0}](x)  &  =\frac{\int_{\mathcal{X}}%
p_{(X,Y)}(x,y_{k+1}|\tilde{x}_{k},y_{k};\theta)\Gamma_{k}[y_{k:1},y_{0}%
,f_{0}](\tilde{x}_{k})d\tilde{x}_{k}}{\int_{\mathcal{X}}p_{Y}(y_{k+1}%
|\tilde{x}_{k},y_{k};\theta)\Gamma_{k}[y_{k:1},y_{0},f_{0}](\tilde{x}%
_{k})d\tilde{x}_{k}},\label{eq:S.51}\\
\Gamma_{k+1}^{\prime}[y_{k+1:1},y_{0},f_{0},g_{0}](x)  &  =A_{1}+A_{2}%
-\Gamma_{k+1}[y_{k+1:1},y_{0},f_{0}](x)\left(  B_{1}+B_{2}\right)
\label{eq:S.52}%
\end{align}
with $y_{k:1}=(y_{1},\ldots,y_{k})$ and $\left(  A_{1},A_{2},B_{1}%
,B_{2}\right)  $ defined as follows:%
\begin{align*}
A_{1}  &  =\frac{\int_{\mathcal{X}}\frac{\partial p_{(X,Y)}}{\partial\theta
}(x,y_{k+1}|\tilde{x}_{k},y_{k};\theta)\Gamma_{k}[y_{k:1},y_{0},f_{0}%
](\tilde{x}_{k})d\tilde{x}_{k},}{\int_{\mathcal{X}}p_{Y}(y_{k+1}|\tilde{x}%
_{k},y_{k};\theta)\Gamma_{k}[y_{k:1},y_{0},f_{0}](\tilde{x}_{k})d\tilde{x}%
_{k}},\text{ }\\
A_{2}  &  =\frac{\int_{\mathcal{X}}p_{(X,Y)}(x,y_{k+1}|\tilde{x}_{k}%
,y_{k};\theta)\Gamma_{k}^{\prime}[y_{k:1},y_{0},f_{0},g_{0}](\tilde{x}%
_{k})d\tilde{x}_{k}}{\int_{\mathcal{X}}p_{Y}(y_{k+1}|\tilde{x}_{k}%
,y_{k};\theta)\Gamma_{k}[y_{k:1},y_{0},f_{0}](\tilde{x}_{k})d\tilde{x}_{k}},\\
B_{1}  &  =\frac{\int_{\mathcal{X}}\frac{\partial p_{Y}}{\partial\theta
}(y_{k+1}|\tilde{x}_{k},y_{k};\theta)\Gamma_{k}[y_{k:1},y_{0},f_{0}](\tilde
{x}_{k})d\tilde{x}_{k}}{\int_{\mathcal{X}}p_{Y}(y_{k+1}|\tilde{x}_{k}%
,y_{k};\theta)\Gamma_{k}[y_{k:1},y_{0},f_{0}](\tilde{x}_{k})d\tilde{x}_{k}%
},\text{ }\\
B_{2}  &  =\frac{\int_{\mathcal{X}}p_{Y}(y_{k+1}|\tilde{x}_{k},y_{k}%
;\theta)\Gamma_{k}^{\prime}[y_{k:1},y_{0},f_{0},g_{0}](\tilde{x}_{k}%
)d\tilde{x}_{k}}{\int_{\mathcal{X}}p_{Y}(y_{k+1}|\tilde{x}_{k},y_{k}%
;\theta)\Gamma_{k}[y_{k:1},y_{0},f_{0}](\tilde{x}_{k})d\tilde{x}_{k}}.
\end{align*}
Then, by iteratively applying the updating formulae (\ref{eq:S.36}) and
(\ref{eq:S.38}), we obtain
\begin{align*}
p(x_{k}|\mathbf{Y}_{k};\theta)  &  =\Gamma_{k}[Y_{k:1},Y_{0},p_{0}(\cdot
|Y_{0};\theta)](x_{k}),\\
\frac{\partial p(x_{k}|\mathbf{Y}_{k};\theta)}{\partial\theta}  &  =\Gamma
_{k}^{\prime}\left[  Y_{k:1},Y_{0},p_{0}(\cdot|Y_{0};\theta),\frac{\partial
p_{0}(\cdot|Y_{0};\theta)}{\partial\theta}\right]  (x_{k}).
\end{align*}
Thus, by the Markov property of $(X_{k},Y_{k})$ and the same argument as in
Step 1, we have, for any $\xi_{0}=(x_{0},y_{0},f_{0},g_{0})\in\Xi,$%
\begin{align*}
&  \mathbb{E}[G(\xi_{k}(\theta))|\xi_{0}(\theta)=\xi_{0}]\\
&  =\mathbb{E}\left[  G\left(  X_{k},Y_{k},\Gamma_{k}[Y_{k:1},Y_{0}%
,p_{0}(\cdot|Y_{0};\theta)],\Gamma_{k}^{\prime}\left[  Y_{k:1},Y_{0}%
,p_{0}(\cdot|Y_{0};\theta),\frac{\partial p_{0}(\cdot|Y_{0};\theta)}%
{\partial\theta}\right]  \right)  |\xi_{0}(\theta)=\xi_{0}\right] \\
&  =\mathbb{E}\left[  G\left(  X_{k},Y_{k},\Gamma_{k}[Y_{k:1},y_{0}%
,f_{0}(\cdot)],\Gamma_{k}^{\prime}[Y_{k:1},y_{0},f_{0}(\cdot),g_{0}%
(\cdot)]\right)  |(X_{0},Y_{0})=(x_{0},y_{0})\right]  .
\end{align*}
To prove (\ref{eq:S.44}), we rewrite the above expression using the transition
density:
\begin{align*}
&  \mathbb{E}[G(\xi_{k}(\theta))|\xi_{0}(\theta)=\xi_{0}]\\
&  =\int_{\mathcal{X}^{k}\times\mathcal{Y}^{k}}G\left(  x_{k},y_{k},\Gamma
_{k}[Y_{k:1},y_{0},f_{0}(\cdot)],\Gamma_{k}^{\prime}[Y_{k:1},y_{0},f_{0}%
(\cdot),g_{0}(\cdot)]\right) \\
&  \quad\times\prod_{j=1}^{k}p_{(X,Y)}(x_{j},y_{j}|x_{j-1},y_{j-1}%
;\theta)\,dx_{1:k}\,dy_{1:k}.
\end{align*}
Let $q(x_{j+1:k},y_{j+1:k}|x_{j},y_{j};\theta)=\prod_{j=1}^{k}p_{(X,Y)}%
(x_{j},y_{j}|x_{j-1},y_{j-1};\theta)$, then, for any two initial values
$\xi_{0}=(x_{0},y_{0},f_{0},g_{0}),\quad\bar{\xi}_{0}=(\bar{x}_{0},\bar{y}%
_{0},\bar{f}_{0},\bar{g}_{0}),$ we decompose $\mathbb{E}[G(\xi_{k}%
(\theta))|\xi_{0}(\theta)=\xi_{0}]-\mathbb{E}[G(\xi_{k}(\theta))|\xi
_{0}(\theta)=\bar{\xi}_{0}]=D_{1}+D_{2},$ where%
\begin{align}
D_{1}  &  =\int_{\mathcal{X}^{k}\times\mathcal{Y}^{k}}\left[  G(x_{k}%
,y_{k},\Gamma_{k}[Y_{k:1},y_{0},f_{0}],\Gamma_{k}^{\prime}[Y_{k:1},y_{0}%
,f_{0},g_{0}])\right. \label{eq:S.53a}\\
&  \quad\left.  -G(x_{k},y_{k},\Gamma_{k}[Y_{k:1},\bar{y}_{0},\bar{f}%
_{0}],\Gamma_{k}^{\prime}[Y_{k:1},\bar{y}_{0},\bar{f}_{0},\bar{g}%
_{0}])\right]  q(x_{2:k},y_{2:k}|x_{1},y_{1};\theta)\,dx_{1:k}\,dy_{1:k}%
,\nonumber\\
D_{2}  &  =\int_{\mathcal{X}^{k}\times\mathcal{Y}^{k}}G(x_{k},y_{k},\Gamma
_{k}[Y_{k:1},\bar{y}_{0},\bar{f}_{0}],\Gamma_{k}^{\prime}[Y_{k:1},\bar{y}%
_{0},\bar{f}_{0},\bar{g}_{0}])q(x_{2:k},y_{2:k}|x_{1},y_{1};\theta)\nonumber\\
&  \times\left(  p_{(X,Y)}(x_{1},y_{1}|x_{0},y_{0};\theta)-p_{(X,Y)}%
(x_{1},y_{1}|\bar{x}_{0},\bar{y}_{0};\theta)\right)  dx_{1:k}\,dy_{1:k}
\label{eq:S.53b}%
\end{align}
Thus, to prove (\ref{eq:S.44}), it suffices to show that there exist constants
$C_{j}>0$ and $r\in(0,1)$ such that $|D_{j}|\leq C_{j}r_{j}^{k},\quad j=1,2.$
As a preparation, we show that there exist constants $C_{\Gamma}%
,C_{\Gamma^{\prime}}>0$ and $r_{\Gamma},r_{\Gamma^{\prime}}\in(0,1)$ such
that, for any $y_{k:1}\in\mathcal{Y}^{k},\quad x_{0}\in\mathcal{X},\quad
f_{0}(\cdot)\in\mathcal{S},\quad g_{0}(\cdot)\in\mathcal{\dot{S}},$
\begin{align}
\Vert\Gamma_{k}[y_{k:1},y_{0},f_{0}]-\Gamma_{k}[y_{k:1},\bar{y}_{0},\bar
{f}_{0}]\Vert_{L_{1}}  &  \leq C_{\Gamma}(r_{\Gamma})^{k},\label{eq:S.54a}\\
\Vert\Gamma_{k}^{\prime}[y_{k:1},y_{0},f_{0},g_{0}]-\Gamma_{k}^{\prime
}[y_{k:1},\bar{y}_{0},\bar{f}_{0},\bar{g}_{0}]\Vert_{L_{1}}  &  \leq
C_{\Gamma^{\prime}}(r_{\Gamma^{\prime}})^{k}. \label{eq:S.54b}%
\end{align}
We will apply Theorems 3.1-3.2 in \cite{TadicDoucet2005} to prove these
bounds. The transition probability for $(X_{k},Y_{k})$ in our case is more
complex than that in the references, the techniques in \cite{TadicDoucet2005}
still work. So, for the sake of brevity, we omit the modifications of their
proof and directly state the results as follows. Under the regularity
assumptions, following Theorems 3.1-3.2 in \cite{TadicDoucet2005}, we have
\begin{equation}
\Vert\Gamma_{k}[y_{k:1},y_{0},f_{0}]-\Gamma_{k}[y_{k:1},\bar{y}_{0},\bar
{f}_{0}]\Vert_{L_{1}}\leq C_{\Gamma}^{(1)}r_{\Gamma}^{k}\Vert f_{0}-\bar
{f}_{0}\Vert_{L_{1}}\leq C_{\Gamma}r_{\Gamma}^{k} \label{eq:S.55a}%
\end{equation}
and
\begin{align}
&  \Vert\Gamma_{k}^{\prime}[y_{k:1},y_{0},f_{0},g_{0}]-\Gamma_{k}^{\prime
}[y_{k:1},\bar{y}_{0},\bar{f}_{0},\bar{g}_{0}]\Vert_{L_{1}}\nonumber\\
&  \leq C_{\Gamma^{\prime}}^{(1)}r_{\Gamma^{\prime}}^{k}\Vert f_{0}-\bar
{f}_{0}\Vert_{L_{1}}\Vert g_{0}\Vert_{L_{1}}+C_{\Gamma^{\prime}}%
^{(2)}r_{\Gamma^{\prime}}^{k}\Vert g_{0}-\bar{g}_{0}\Vert_{L_{1}}.
\label{eq:S.55b}%
\end{align}
Since $f_{0},\bar{f}_{0}\in S$, we have $\Vert f_{0}-\bar{f}_{0}\Vert_{L_{1}%
}\leq2.$ From (\ref{eq:S.46}), there exists a constant $C_{\mathcal{\dot{S}}%
},$ such that $\Vert g_{0}\Vert_{L_{1}},\Vert\bar{g}_{0}\Vert_{L_{1}}\leq
C_{\mathcal{\dot{S}}}$. Then, we can find $C_{\Gamma^{\prime}}>0$ and
$r_{\Gamma^{\prime}}\in(0,1)$ such that
\[
\Vert\Gamma_{k}^{\prime}[y_{k:1},y_{0},f_{0},g_{0}]-\Gamma_{k}^{\prime
}[y_{k:1},\bar{y}_{0},\bar{f}_{0},\bar{g}_{0}]\Vert_{L_{1}}\leq C_{\Gamma
^{\prime}}r_{\Gamma^{\prime}}^{k}.
\]
Now we estimate $D_{1}$ and $D_{2}$. For $D_{1}$ in (\ref{eq:S.53a}), by the
Lipschitz property of $G\in G_{Lip}$, we have
\begin{align*}
&  \left\vert G(x_{k},y_{k},\Gamma_{k}[y_{k:1},y_{0},f_{0}],\Gamma_{k}%
^{\prime}[y_{k:1},y_{0},f_{0},g_{0}])-G(x_{k},y_{k},\Gamma_{k}[y_{k:1},\bar
{y}_{0},\bar{f}_{0}],\Gamma_{k}^{\prime}[y_{k:1},\bar{y}_{0},\bar{f}_{0}%
,\bar{g}_{0}])\right\vert \\
&  \leq C_{G}\Vert\Gamma_{k}[y_{k:1},y_{0},f_{0}]-\Gamma_{k}[y_{k:1},\bar
{y}_{0},\bar{f}_{0}]\Vert_{L_{1}}+C_{G}\Vert\Gamma_{k}^{\prime}[y_{k:1}%
,y_{0},f_{0},g_{0}]-\Gamma_{k}^{\prime}[y_{k:1},\bar{y}_{0},\bar{f}_{0}%
,\bar{g}_{0}]\Vert_{L_{1}}.
\end{align*}
Next, by the iterative definitions of $\Gamma_{k}$ and $\Gamma_{k}^{\prime}$
and the inequalities (\ref{eq:S.54a})--(\ref{eq:S.54b}), we obtain
\begin{align*}
&  \left\Vert \Gamma_{k}[y_{k:1},y_{0},f_{0}]-\Gamma_{k}[y_{k:1},\bar{y}%
_{0},\bar{f}_{0}]\right\Vert _{L_{1}}\leq C_{\Gamma}(r_{\Gamma})^{i-1},\\
&  \left\Vert \Gamma_{k}^{\prime}[y_{k:1},y_{0},f_{0},g_{0}]-\Gamma
_{k}^{\prime}[y_{k:1},\bar{y}_{0},\bar{f}_{0},\bar{g}_{0}]\right\Vert _{L_{1}%
}\leq C_{\Gamma^{\prime}}(r_{\Gamma^{\prime}})^{i-1},
\end{align*}
where $f_{0}^{(1)}(x)=\Gamma_{1}[y_{1},y_{0},f_{0}](x),$ $\bar{f}_{0}%
^{(1)}(x)=\Gamma_{1}[y_{1},\bar{y}_{0},\bar{f}_{0}](x),$ $g_{0}^{(1)}%
(x)=\Gamma_{1}^{\prime}[y_{1},y_{0},f_{0},g_{0}](x),$ and $\bar{g}_{0}%
^{(1)}(x)=\Gamma_{1}^{\prime}[y_{1},\bar{y}_{0},\bar{f}_{0},\bar{g}_{0}](x).$
Thus, for $D_{1}$ in (\ref{eq:S.53a}), we have
\begin{align*}
|D_{1}|  &  \leq\int_{\mathcal{X}^{k}\times\mathcal{Y}^{k}}C_{G}\left[
C_{\Gamma}(r_{\Gamma})^{k-1}+C_{\Gamma^{\prime}}(r_{\Gamma^{\prime}}%
)^{k-1}\right]  q(x_{1:k},y_{1:k}|x_{0},y_{0};\theta)\,dx_{1:k}dy_{1:k}\\
&  =C_{G}\left[  C_{\Gamma}(r_{\Gamma})^{k-1}+C_{\Gamma^{\prime}}%
(r_{\Gamma^{\prime}})^{k-1}\right]  ,
\end{align*}
where the last equality is owing to the density nature of $p_{(X,Y)}$. Then,
letting $C_{1}=C_{G}\left(  C_{\Gamma}/r_{\Gamma}+C_{\Gamma^{\prime}%
}/r_{\Gamma^{\prime}}\right)  ,$ $r_{1}=\max\{r_{\Gamma},r_{\Gamma^{\prime}%
}\},$ we obtain $|D_{1}|\leq C_{1}r_{1}^{k}.$ For $D_{2}$ in (\ref{eq:S.53b}),
we further decompose $D_{2}$ into $D_{2}=D_{2,1}+D_{2,2}$, where
\begin{align*}
D_{2,1}  &  =\int_{\mathcal{X}^{k}\times\mathcal{Y}^{k}}(G(x_{k},y_{k}%
,\Gamma_{k}[y_{k:1},\bar{y}_{0},\bar{f}_{0}](x_{k}),\Gamma_{k}^{\prime
}[y_{k:1},\bar{y}_{0},\bar{f}_{0},\bar{g}_{0}(\cdot)](x_{k}))\\
&  \qquad\qquad-G(x_{k},y_{k},\Gamma_{k-m}[y_{k:m+1},y_{m},f^{\prime}%
](x_{k}),\Gamma_{k-m}^{\prime}[y_{k:m+1},y_{m},f^{\prime},g^{\prime}%
](x_{k})))\\
&  \qquad\times q(x_{2:k},y_{2:k}|x_{1},y_{1};\theta)\left(  p_{(X,Y)}%
(x_{1},y_{1}|x_{0},y_{0};\theta_{0})-p_{(X,Y)}(x_{1},y_{1}|\bar{x}_{0},\bar
{y}_{0};\theta_{0})\right)  \,dx_{1:k}dy_{1:k},\\
D_{2,2}  &  =\int_{\mathcal{X}^{k}\times\mathcal{Y}^{k}}G(x_{k},y_{k}%
,\Gamma_{k-m}[y_{k:m+1},y_{m},f^{\prime}](x_{k}),\Gamma_{k-m}^{\prime
}[y_{k:m+1},y_{m},f^{\prime},g^{\prime}](x_{k}))\\
&  \times q(x_{2:k},y_{2:k}|x_{1},y_{1};\theta)\left(  p_{(X,Y)}(x_{1}%
,y_{1}|x_{0},y_{0};\theta_{0})-p_{(X,Y)}(x_{1},y_{1}|\bar{x}_{0},\bar{y}%
_{0};\theta_{0})\right)  \,dx_{1:k}dy_{1:k}.
\end{align*}
Here, we have $m=\lfloor k/2\rfloor$, where we denote by $\lfloor x\rfloor$
the largest integer that is no larger than $x$, and $(f^{\prime}%
,g^{\prime})\in\mathcal{S}\times \dot{\mathcal{S}}$ is fixed. Then, it
suffices to show that there exist $C_{2,j}>0$ and $r_{2,j}\in(0,1)$ such that
$|D_{2,j}|\leq C_{2,j}r_{2,j}^{k}$, for $j=1,2$. For $D_{2,1}$, similar to our
methods for $D_{1}$, by the iterative definition of $\Gamma_{k}$ and
$\Gamma_{k}^{\prime}$ and the inequalities (\ref{eq:S.54a})-(\ref{eq:S.54b}),
we obtain
\begin{align*}
&  \left\Vert \Gamma_{k}[y_{k:1},\bar{y}_{0},\bar{f}_{0}]-\Gamma
_{k-m}[y_{k:m+1},y_{m},f^{\prime}]\right\Vert _{L_{1}}\\
&  =\left\Vert \Gamma_{k-m}[y_{k:m+1},y_{m},\bar{f}_{0}^{(1)}]-\Gamma
_{k-m}[y_{k:m+1},y_{m},f^{\prime}]\right\Vert _{L_{1}}\\
&  \quad\leq C_{\Gamma}(r_{\Gamma})^{k-m},\\
&  \left\Vert \Gamma_{k}^{\prime}[y_{k:1},\bar{y}_{0},\bar{f}_{0},\bar{g}%
_{0}]-\Gamma_{k-m}^{\prime}[y_{k:m+1},y_{m},f^{\prime},g^{\prime}]\right\Vert
_{L_{1}}\\
&  =\left\Vert \Gamma_{k-m}^{\prime}[y_{k:m+1},y_{m},\bar{f}_{0}^{(1)},\bar
{g}_{0}^{(1)}]-\Gamma_{k-m}^{\prime}[y_{k:m+1},y_{m},f^{\prime},g^{\prime
}]\right\Vert _{L_{1}}\\
&  \leq C_{\Gamma^{\prime}}(r_{\Gamma^{\prime}})^{k-m},
\end{align*}
where $\bar{f}_{0}^{(1)}=\Gamma_{m}[y_{m:1},\bar{y}_{0},\bar{f}_{0}]$ and
$\bar{g}_{0}^{(1)}=\Gamma_{m}^{\prime}[y_{m:1},\bar{y}_{0},\bar{f}_{0},\bar
{g}_{0}].$ Then, by the Lipschitz property of $G\in G_{Lip}$, we have
\begin{align*}
|D_{2,1}|  &  \leq\int_{\mathcal{X}^{k}\times\mathcal{Y}^{k}}C_{G}\left[
C_{\Gamma}(r_{\Gamma})^{k-m}+C_{\Gamma^{\prime}}(r_{\Gamma^{\prime}}%
)^{k-m}\right]  q(x_{2:k},y_{2:k}|x_{1},y_{1};\theta)\\
&  \qquad\times\left\vert p_{(X,Y)}(x_{1},y_{1}|x_{0},y_{0};\theta
_{0})-p_{(X,Y)}(x_{1},y_{1}|\bar{x}_{0},\bar{y}_{0};\theta_{0})\right\vert
\,dx_{1:k}dy_{1:k}\\
&  \leq2C_{G}\left[  C_{\Gamma}(r_{\Gamma})^{k-m}+C_{\Gamma^{\prime}%
}(r_{\Gamma^{\prime}})^{k-m}\right]  ,
\end{align*}
where the last equality is owing to the density nature of $p_{(X,Y)}$. Then,
since $m=\lfloor k/2\rfloor$, we let $C_{2,1}=2C_{G}\left(  \frac{C_{\Gamma}%
}{r_{\Gamma}}+\frac{C_{\Gamma^{\prime}}}{r_{\Gamma^{\prime}}}\right)  ,\quad
r_{2,1}=\max\{r_{\Gamma},r_{\Gamma^{\prime}}\}^{1/2}$. Then, we obtain
$|D_{2,1}|\leq C_{2,1}r_{2,1}^{k}.$ For $D_{2,2}$, since the function
$G(x_{k},y_{k},\Gamma_{k-m}[y_{k:m+1},y_{m},f^{\prime}(\cdot)](x_{k}%
),\Gamma_{k-m}^{\prime}[y_{k:m+1},y_{m},f^{\prime},g^{\prime}](x_{k}))$ is
free of the variables $x_{1:(m-1)}$ and $y_{1:(m-1)}$, we have
\begin{align*}
D_{2,2}  &  =\int_{\mathcal{X}^{m}\times\mathcal{Y}^{m}}G_{m}(x_{m}%
,y_{m})q(x_{2:m},y_{2:m}|x_{1},y_{1};\theta)\\
&  \qquad\times\left\vert p_{(X,Y)}(x_{1},y_{1}|x_{0},y_{0};\theta
_{0})-p_{(X,Y)}(x_{1},y_{1}|\bar{x}_{0},\bar{y}_{0};\theta_{0})\right\vert
\,dx_{1:m}dy_{1:m},
\end{align*}
where $G_{m}(x_{m},y_{m})$ is defined as
\begin{align*}
G_{m}(x_{m},y_{m})  &  =\int_{\mathcal{X}^{k-m}\times\mathcal{Y}^{k-m}}%
G(x_{k},y_{k},\Gamma_{k-m}[y_{k:m+1},y_{m},f^{\prime}](x_{k}),\Gamma
_{k-m}^{\prime}[y_{k:m+1},y_{m},f^{\prime},g^{\prime}](x_{k}))\\
&  \qquad\times q(x_{m+1:k},y_{m+1:k}|x_{m},y_{m};\theta)\,dx_{(m+1):k}%
dy_{(m+1):k}.
\end{align*}
Then, using the Chapman-Kolmogorov equation for the transition density
$p_{(X,Y)}$, we have
\[
D_{2,2}=\int_{\mathcal{X}\times\mathcal{Y}}G_{m}(x_{m},y_{m})|p_{(X,Y)}%
(m,x_{m},y_{m}|x_{0},y_{0};\theta_{0})-p_{(X,Y)}(m,x_{m},y_{m}|\bar{x}%
_{0},\bar{y}_{0};\theta_{0})|\,dx_{m}dy_{m}.
\]
where $p_{(X,Y)}(m,x_{m},y_{m}|x_{0},y_{0};\theta_{0})$ denotes the $m$ steps
transition density. Since the function $G$ is bounded, there exists a constant
$B_{G}>0$ such that $|G(\xi_{0})|<B_{G}$ for any $\xi_{0}\in\Xi$. Thus, we
obtain
\[
|G_{m}(x_{m},y_{m})|\leq\int_{\mathcal{X}^{k-m}\times\mathcal{Y}^{k-m}}%
B_{G}q(x_{m+1:k},y_{m+1:k}|x_{m},y_{m};\theta)dx_{(m+1):k}dy_{(m+1):k}=B_{G},
\]
where the last equality is owing to the density nature of $p_{(X,Y)}$. Then,
by Assumption \ref{ass:S2} we obtain%
\begin{align*}
\left\vert D_{2,2}\right\vert  &  \leq\int_{\mathcal{X}\times\mathcal{Y}}%
B_{G}\left\vert p_{(X,Y)}(m,x_{m},y_{m}|x_{0},y_{0};\theta_{0})-p_{(X,Y)}%
(m,x_{m},y_{m}|\bar{x}_{0},\bar{y}_{0};\theta_{0})\right\vert \,dx_{m}dy_{m}\\
&  \leq2B_{G}C_{(X,Y)}r_{(X,Y)}^{m}.
\end{align*}
Then, since $m=\lfloor k/2\rfloor$, we let $C_{2,2}=2B_{G}C_{(X,Y)}%
r_{(X,Y)}^{-1},\quad r_{2,2}=r_{(X,Y)}^{1/2}.$ Then, we obtain $|D_{2,2}|\leq
C_{2,2}r_{2,2}^{k}.$

To show that $G_{Lip}$ is dense in $C_{b}(\Xi)$, we begin by proving that
$\mathcal{S}\times\mathcal{\dot{S}}$ is compact. By definition, $\mathcal{S}$ defined in (\ref{eq:S_state}) is the uniform-norm closure of the set of reachable filtered densities. Under Assumption \ref{ass:S5}, these densities are uniformly Lipschitz on the compact set $\mathcal{X}$, together with $||f||_{L^1}=1$, this implies uniform boundedness and equicontinuity. Hence, $\mathcal{S}$ is compact under the uniform norm by the Arzel\`a-Ascoli theorem. Similarly, the reachable filter derivatives are uniformly bounded in $L^1$ and uniformly Lipschitz on $\mathcal{X}$, so each $\dot{\mathcal{S}_j}$, and therefore $\dot{\mathcal{S}}$ is compact under the uniform norm. Since uniform convergence implies $L^1$ convergence on compact $\mathcal{X}$, both $\mathcal{S}$ and $\mathcal{\dot{S}}$ are compact under the $L^1$ topology. Therefore, $\Xi$ is
compact. So we have that $\Xi$ is compact, where the norm on $\Xi$ is defined
as
\begin{equation}
\Vert\xi_{0}\Vert_{\Xi,1}=|x_{0}|+|y_{0}|+\Vert f_{0}\Vert_{L_{1}}+\Vert
g_{0}\Vert_{L_{1}}, \label{eq:S.46}%
\end{equation}
for $\xi_{0}=(x_{0},y_{0},f_{0},g_{0})\in\Xi.$ The distance on $\Xi$ is
defined as $d(\xi_{0},\xi_{1})=\Vert\xi_{0}-\xi_{1}\Vert_{\Xi,1},$ for any
$\xi_{0},\xi_{1}\in\Xi$. Then, by the Stone-Weierstrass Theorem, proving that
$G_{Lip}$ is dense in $C_{b}(\Xi)$ boils down to verifying the following two
properties for $G_{Lip}$: (i) The constant function 1 belongs to $G_{Lip}$,
and for any $G_{1}$, $G_{2}\in G_{Lip}$, and any $c_{1}$, $c_{2}\in$
$\mathbb{R}$, it holds that $c_{1}G_{1}+c_{2}G_{2}\in G_{Lip}.$ (ii) For any
$\xi_{1}\neq\xi_{2},$ there exists a $G\in G_{Lip}$, such that $G\left(
\xi_{1}\right)  \neq G\left(  \xi_{2}\right)  .$

\noindent Indeed, (i) follows from the $\mathcal{S}\times\mathcal{\dot{S}}%
$-Lipschitz property and the boundedness of the functions in $G_{Lip}$.
Suppose that $\xi_{1}=(x_{1},y_{1},f_{1}(\cdot),g_{1}(\cdot)),$ $\xi
_{2}=(x_{2},y_{2},f_{2},g_{2})\in\Xi,$ and satisfy $f_{1}\neq f_{2}$. First,
we focus on the case when $f_{1}(\cdot)\neq f_{2}(\cdot)$. Recall that the
definition of $\mathcal{S}$, $f_{1}(\cdot)\neq f_{2}(\cdot)$ implies that
there exists $\phi(\cdot)\in C_{b}(\mathcal{X})$ such that $\int_{\mathcal{X}%
}\phi(y)f_{1}(y)dy\neq\int_{\mathcal{X}}\phi(y)f_{2}(y)dy.$ Using $\phi$, we
define $G(\xi_{0})=\int_{\mathcal{X}}\phi(y)f(y)dy,$ for any $\xi_{0}%
=(x_{0},y_{0},f_{0}(\cdot),g_{0}(\cdot))\in\Xi.$ Then, we have that $G(\xi
_{1})\neq G(\xi_{2})$, and we can verify that $G\in G_{Lip}$ because $G$
satisfies the $\mathcal{S}\times\mathcal{\dot{S}}$-Lipschitz property and $G$
is bounded. The cases where $g_{1}(\cdot)\neq g_{2}(\cdot)$ can be treated
similarly, and the cases where the $x$ dimension or $y$ dimension of $\xi_{k}$
and $\xi_{\ell}$ are different can be handled in an easier way, e.g., if
$x_{1}\neq x_{2}$, we define $G(\xi_{i})=x_{i},$ for any $i=1,2.$
\end{proof}

\subsection{Proof of auxiliary Lemma \ref{lemma:clt for s}%
\label{sec:aux lemma2}}

\begin{proof}
By Assumption \ref{ass:S5}, it can be verified that $S(\xi_{0}^{+};\theta
_{0})$ satisfies the following $\mathcal{S}\times\mathcal{\dot{S}}$-Lipschitz
property, i.e., there exists a constant $C_{S}>0$ such that for any
$(y,y_{0},x_{0})\in\mathcal{Y}\times\mathcal{Y}\times\mathcal{X},\qquad
f_{0},f_{1}\in\mathcal{S},\qquad g_{0},g_{1}\in\mathcal{\dot{S}},$ it holds
\[
\left\vert S(y,y_{0},x_{0},f_{0},g_{0};\theta_{0})-S(y,y_{0},x_{0},f_{1}%
,g_{1};\theta_{0})\right\vert \leq C_{S}\left(  \left\Vert f_{0}%
-f_{1}\right\Vert _{L_{1}}+\left\Vert g_{0}-g_{1}\right\Vert _{L_{1}}\right)
.
\]
Next, following the same route for proving (\ref{eq:S.43}), we have that there
exist constants $K_{S}>0$ and $r\in(0,1)$ such that, for any $\xi_{0}^{+}%
\in\mathcal{Y}\times\Xi$, it holds
\begin{equation}
\left\vert \mathbb{E}[S(\xi^{+}(\theta_{0});\theta_{0})|\xi_{0}^{+}(\theta
_{0})=\xi_{0}^{+}]-\mathbb{E}[S(\xi^{+}(\theta_{0});\theta_{0})]\right\vert
\leq K_{S}r^{k}. \label{eq:S.59}%
\end{equation}

Then, to apply the martingale CLT, we follow the standard route as we did in
(\ref{eq:S.47}) to consider the Poisson equation for $V$, i.e., $V(\xi_{0}%
^{+})-T^{+}V(\xi_{0}^{+})=S(\xi_{0}^{+};\theta_{0})-\mathbb{E}[S(\xi
^{+}(\theta_{0});\theta_{0})]$ for any $\xi_{0}^{+}\in\mathcal{Y}\times\Xi$,
where the operator $T^{+}$ is defined as $T^{+}G(\xi_{0}^{+})=\mathbb{E}%
[G(\xi_{k+1}^{+}(\theta_{0}))|\xi_{k}^{+}(\theta_{0})=\xi_{0}^{+}]$ for any
$G\in C_{b}(\mathcal{Y}\times\Xi)$. Similar to (\ref{eq:S.48}), we obtain by
calculations that a solution for the aforementioned Poisson equation is the
function $V_{S}$ defined as follows:
\begin{equation}
V_{S}(\xi_{0}^{+})=\sum_{k=0}^{+\infty}\left\{  \mathbb{E}[S(\xi^{+}%
(\theta_{0});\theta_{0})|\xi_{0}^{+}(\theta_{0})=\xi_{0}^{+}]-E[S(\xi
^{+}(\theta_{0});\theta_{0})]\right\}  . \label{eq:S.60}%
\end{equation}
The function $V_{S}$ is well-defined, because the series in (\ref{eq:S.48})
converges owing to (\ref{eq:S.59}). In addition, (\ref{eq:S.59}) also implies
that $V_{S}$ is bounded since $V_{S}(\xi_{0}^{+})|\leq\sum_{k=0}^{+\infty
}K_{S}r^{k}=K_{S}/\left(  1-r\right)  .$ In particular, as (\ref{eq:S.59})
implies $\mathbb{E}[S(\xi^{+}(\theta_{0});\theta_{0})]=\lim_{k\rightarrow
\infty}\mathbb{E}[S(\xi_{k}^{+}(\theta_{0});\theta_{0})],$ and $S(\xi_{k}%
^{+}(\theta_{0});\theta_{0})=\partial\log p(Y_{k+1}|\mathbf{Y}_{k};\theta
_{0})/\partial\theta,$ by the conditional version of the Bartlett identities,
we obtain $\mathbb{E}\left[  \frac{\log p(Y_{k+1}|\mathbf{Y}_{k};\theta_{0}%
)}{\partial\theta}\right]  =\mathbb{E}[S(\xi_{k}^{+}(\theta_{0});\theta
_{0})]=0,$ which immediately leads to $\mathbb{E}[S(\xi^{+}(\theta_{0}%
);\theta_{0})]=0$.

To prove the CLT for the score function $S_{n}(\theta_{0})$, using $V_{S}$ in
(\ref{eq:S.60}), which is the solution of the Poisson equation, we begin by
expressing $S_{n}(\theta_{0})$ as
\begin{equation}
S_{n}(\theta_{0})=\left(  V_{S}(\xi_{0}^{+}(\theta_{0}))-V_{S}(\xi_{n}%
^{+}(\theta_{0}))\right)  +\left(  \sum_{k=1}^{n}V_{S}(\xi_{k}^{+}(\theta
_{0}))-TV_{S}(\xi_{k-1}^{+}(\theta_{0}))\right)  . \label{eq:S.61}%
\end{equation}
Since $S_{n}(\theta_{0})$ is a vector, it follows from the Cramer-Wold device
(see, e.g., Section 2.3 of \cite{vanderVaart1998}) that proving $n^{-1/2}%
S_{n}(\theta_{0})\overset{d}{\longrightarrow}N(0,I_{S})$ boils down to proving
$n^{-1/2}v^{\intercal}S_{n}(\theta_{0})\overset{d}{\longrightarrow
}N(0,v^{\intercal}I_{S}v)$ for any $v\in\mathbb{R}^{d}$. In the right-hand
side of (\ref{eq:S.61}), for the first term, by the boundedness of $V_{S}$, we
obtain $n^{-1/2}v^{\intercal}\left[  V_{S}(\xi_{0}^{+}(\theta_{0}))-V_{S}%
(\xi_{n}^{+}(\theta_{0}))\right]  \overset{p}{\longrightarrow}0.$ For the
second term, since the process $(\xi_{k-1}^{+}(\theta_{0}),\xi_{k}^{+}%
(\theta_{0}))$ is ergodic, and since $v^{\intercal}\left[  V_{S}(\xi_{1}%
^{+})-TV_{S}(\xi_{0}^{+})\right]  \left[  V_{S}(\xi_{1}^{+})-TV_{S}(\xi
_{0}^{+})\right]  ^{\intercal}v$ is a continuous bounded function of $(\xi
_{1}^{+},\xi_{0}^{+})$ owing to $V_{S}(\xi_{0}^{+})\in C_{b}(\mathcal{Y}%
\times\Xi)$, we obtain by the property of ergodic process that
\[
\frac{1}{n}\sum_{k=1}^{n}\left(  v^{\intercal}\left[  V_{S}(\xi_{k}^{+}%
(\theta_{0}))-TV_{S}(\xi_{k-1}^{+}(\theta_{0}))\right]  \left[  V_{S}(\xi
_{k}^{+}(\theta_{0}))-TV_{S}(\xi_{k-1}^{+}(\theta_{0}))\right]  ^{\intercal
}v\right)  \overset{p}{\longrightarrow}v^{\intercal}I_{S}v.
\]
where $I_{S}:=\mathbb{E}\left[  \left(  \left[  V_{S}(\xi^{++}(\theta
_{0}))-TV_{S}(\xi^{+}(\theta_{0}))\right]  \left[  V_{S}(\xi^{++}(\theta
_{0}))-TV_{S}(\xi^{+}(\theta_{0}))\right]  ^{\intercal}\right)  \right]  ,$
with $\xi^{+}(\theta_{0})$ following the stationary distribution of $\xi
_{k}^{+}(\theta_{0})$ and the conditional distribution of $\xi^{++}(\theta
_{0})$ given $\xi^{+}(\theta_{0})$ is equal to that of $\xi_{k}^{+}(\theta
_{0})$ given $\xi_{k-1}^{+}(\theta_{0})$. Then, by Theorem 3.2 of
\cite{hall1980martingale}, we obtain
\[
n^{-1/2}\sum_{k=1}^{n}v^{\intercal}\left[  V_{S}(\xi_{k}^{+}(\theta
_{0}))-TV_{S}(\xi_{k-1}^{+}(\theta_{0}))\right]  \overset{d}{\longrightarrow
}N(0,v^{\intercal}I_{S}v).
\]
Thus, we conclude that $n^{-1/2}S_{n}(\theta_{0})\overset{d}{\longrightarrow
}N(0,I_{S}).$

To prove $n^{-1}H_{n}(\theta)\overset{p}{\longrightarrow}H_{0}(\theta)$, we
begin by expressing the Hessian $H_{n}(\theta)=\partial^{2}\ell_{n}%
(\theta)/\partial\theta\partial\theta^{\intercal}$ as
\begin{equation}
H_{n}(\theta)=\sum_{k=0}^{n-1}\mathbb{H}(\tilde{\xi}_{k}^{+}(\theta
);\theta),\text{ with }\tilde{\xi}_{k}^{+}(\theta):=\left(  \xi_{k}^{+}%
(\theta),\frac{\partial^{2}p(\cdot|\mathbf{Y}_{k};\theta)}{\partial
\theta\partial\theta^{\intercal}}\right)  , \label{eq:S.62}%
\end{equation}
where $\mathbb{H}$ is a $d\times d$ matrix-valued function defined by
\begin{align*}
&  \mathbb{H}(y,y_{0},x_{0},f_{0},g_{0},h_{0};\theta):=\frac{1}{\int%
_{\mathcal{X}}p_{Y}(y|y_{0},\tilde{x}_{0};\theta)f_{0}(\tilde{x}_{0}%
)d\tilde{x}_{0}}\times\left(  \int_{\mathcal{X}}\frac{\partial^{2}%
p_{Y}(y|y_{0},\tilde{x}_{0};\theta)}{\partial\theta\partial\theta^{\intercal}%
}f_{0}(\tilde{x}_{0})d\tilde{x}_{0}\right. \\
\qquad\quad &  \left.  +\int_{\mathcal{X}}\frac{\partial p_{Y}(y|y_{0}%
,\tilde{x}_{0};\theta)}{\partial\theta}(g_{0}(\tilde{x}_{0}))^{\intercal
}d\tilde{x}_{0}+\int_{\mathcal{X}}\frac{\partial p_{Y}(y|y_{0},\tilde{x}%
_{0};\theta)}{\partial\theta^{\intercal}}g_{0}(\tilde{x}_{0})d\tilde{x}%
_{0}+\int_{\mathcal{X}}p_{Y}(y|y_{0},\tilde{x}_{0};\theta)h_{0}(\tilde{x}%
_{0})d\tilde{x}_{0}\right) \\
\quad &  -\frac{1}{\left(  \int_{\mathcal{X}}p_{Y}(y|y_{0},\tilde{x}%
_{0};\theta)f_{0}(\tilde{x}_{0})d\tilde{x}_{0}\right)  ^{2}}\left[
\int_{\mathcal{X}}\frac{\partial p_{Y}(y|y_{0},\tilde{x}_{0};\theta)}%
{\partial\theta}f_{0}(\tilde{x}_{0})d\tilde{x}_{0}+\int_{\mathcal{X}}%
p_{Y}(y|y_{0},\tilde{x}_{0};\theta)g_{0}(\tilde{x}_{0})d\tilde{x}_{0}\right]
\\
\quad &  \times\left[  \int_{\mathcal{X}}\frac{\partial p_{Y}(y|y_{0}%
,\tilde{x}_{0};\theta)}{\partial\theta}f_{0}(\tilde{x}_{0})d\tilde{x}_{0}%
+\int_{\mathcal{X}}p_{Y}(y|y_{0},\tilde{x}_{0};\theta)g_{0}(\tilde{x}%
_{0})d\tilde{x}_{0}\right]  ^{\intercal}.
\end{align*}
We denote by $\mathcal{\dot{S}}_{2}$ the state space of $\partial^{2}%
p(\cdot|\mathbf{Y}_{k};\theta)/\partial\theta\partial\theta^{\intercal}$. Then
under Assumptions \ref{ass:S5}, $\mathcal{\dot{S}}_{2}$ is compact with the
norm and distance given by $\Vert\cdot\Vert_{L_{1}}$, which can be similarly
shown as (\ref{eq:S.46}). Next, following the same route as that for proving
Lemma \ref{lem:filter-ergodic}, we have that $\tilde{\xi}_{k}^{+}(\theta)$ is
ergodic Markov and $\lim_{k\rightarrow\infty}\mathbb{E}[G(\tilde{\xi}_{k}%
^{+}(\theta))]=\mathbb{E}[G(\tilde{\xi}^{+}(\theta))]$ for any $G\in
C_{b}(\mathcal{Y}\times\Xi\times\mathcal{\dot{S}}_{2})$, where $\tilde{\xi
}^{+}(\theta)$ follows the stationary distribution of $\tilde{\xi}_{k}%
^{+}(\theta)$. Then, by Assumption \ref{ass:S5}, we have $\mathbb{H}\in
C_{b}(\mathcal{Y}\times\Xi\times\mathcal{\dot{S}}_{2}).$ Thus, by property of
ergodic process $\tilde{\xi}_{k}^{+}(\theta)$, we obtain
\begin{equation}
\frac{1}{n}H_{n}(\theta)=\frac{1}{n}\sum_{k=0}^{n-1}\mathbb{H}(\tilde{\xi}%
_{k}^{+}(\theta);\theta)\overset{p}{\longrightarrow}H_{0}(\theta),
\label{eq:S.63}%
\end{equation}
with $H_{0}(\theta):=\mathbb{E}[\mathbb{H}(\tilde{\xi}^{+}(\theta);\theta)].$
By calculations, we have that the summand $\mathbb{H}(\tilde{\xi}_{k}%
^{+}(\theta);\theta)$ in (\ref{eq:S.62}) equals $\partial^{2}\log
p(Y_{k+1}|\mathbf{Y}_{k};\theta)/\partial\theta\partial\theta^{\intercal}$.
Since $\partial^{2}\log p(Y_{k+1}|\mathbf{Y}_{k};\theta)/\partial
\theta\partial\theta^{\intercal}$ is Lipschitz in $\theta$ by Assumption
\ref{ass:S5}, we obtain by the same reasoning as for $n^{-1}\ell_{n}(\theta)$
in Step 1 that there exists a constant $C_{H}>0$ such that $\left\Vert
n^{-1}H_{n}(\theta)-n^{-1}H_{n}(\theta^{\prime})\right\Vert \leq C_{H}%
\Vert\theta-\theta^{\prime}\Vert,$ and the convergence (\ref{eq:S.63}) of
$n^{-1}H_{n}(\theta)$ holds uniformly for $\theta\in\Theta$.

To prove that $H_{0}(\theta)$ is Lipschitz, i.e., $\Vert H_{0}(\theta
)-H_{0}(\theta^{\prime})\Vert_{1}\leq C_{H}\Vert\theta-\theta^{\prime}%
\Vert_{1},$ we begin by using the triangle inequality to obtain
\[
\Vert H_{0}(\theta)-H_{0}(\theta^{\prime})\Vert\leq\left\Vert \frac
{H_{n}(\theta)-H_{n}(\theta^{\prime})}{n}\right\Vert +\left\Vert \frac
{H_{n}(\theta)}{n}-H_{0}(\theta)\right\Vert +\left\Vert \frac{H_{n}%
(\theta^{\prime})}{n}-H_{0}(\theta^{\prime})\right\Vert .
\]
which almost surely holds for any $n$. In the right-hand side of the above
inequality, for the first term, as discussed after (\ref{eq:S.63}), we have
$\left\Vert n^{-1}H_{n}(\theta)-n^{-1}H_{n}(\theta^{\prime})\right\Vert \leq
C_{H}\Vert\theta-\theta^{\prime}\Vert;$ for the second and third terms, both
of them converge to zero in probability. So, we get $\Vert H_{0}(\theta
)-H_{0}(\theta^{\prime})\Vert\leq C_{H}\Vert\theta-\theta^{\prime}\Vert
+o_{p}(1).$Since both $\Vert H_{0}(\theta)-H_{0}(\theta^{\prime})\Vert$ and
$C_{H}\Vert\theta-\theta^{\prime}\Vert$ are deterministic, we conclude that
$\Vert H_{0}(\theta)-H_{0}(\theta^{\prime})\Vert\leq C_{H}\Vert\theta
-\theta^{\prime}\Vert.$
\end{proof}

	{\small
		\bibliographystyle{elsevier}
		\bibliography{newadded}

@article{HansenSargent2007,
  author  = {Hansen, Lars Peter and Sargent, Thomas J.},
  title   = {Recursive Robust Estimation and Control without Commitment},
  journal = {Journal of Economic Theory},
  year    = {2007},
  volume  = {136},
  number  = {1},
  pages   = {1--27},
  month   = {September},
  doi     = {10.1016/j.jet.2006.06.010}
}

@article{DaiSingleton2003,
  author  = {Dai, Qiang and Singleton, Kenneth J.},
  title   = {Term Structure Dynamics in Theory and Reality},
  journal = {The Review of Financial Studies},
  year    = {2003},
  volume  = {16},
  number  = {3},
  pages   = {631--678},
  doi     = {10.1093/rfs/hhg010}
}

@book{kallenberg2021foundations,
  title     = {Foundations of Modern Probability},
  author    = {Kallenberg, Olav},
  edition   = {3},
  year      = {2021},
  publisher = {Springer}
}

@inproceedings{kingma2015adam,
  title={Adam: A Method for Stochastic Optimization},
  author={Kingma, Diederik P. and Ba, Jimmy},
  booktitle={International Conference on Learning Representations},
  year={2015}
}

@book{schmudgen2017moment,
  author    = {Konrad Schm{\"u}dgen},
  title     = {The Moment Problem},
  series    = {Graduate Texts in Mathematics},
  volume    = {277},
  publisher = {Springer},
  address   = {Cham},
  year      = {2017},
  doi       = {10.1007/978-3-319-64546-9},
  isbn      = {978-3-319-64546-9}
}

@article{Andrews1991,
  author  = {Andrews, Donald W. K.},
  title   = {Heteroskedasticity and Autocorrelation Consistent Covariance Matrix Estimation},
  journal = {Econometrica},
  year    = {1991},
  volume  = {59},
  number  = {3},
  pages   = {817--858}
}

@article{White1982,
  author  = {White, Halbert},
  title   = {Maximum Likelihood Estimation of Misspecified Models},
  journal = {Econometrica},
  year    = {1982},
  volume  = {50},
  number  = {1},
  pages   = {1--25},
  doi     = {10.2307/1912526}
}

@article{TadicDoucet2005,
  author  = {Tadi{\'c}, Vladislav B. and Doucet, Arnaud},
  title   = {Exponential Forgetting and Geometric Ergodicity for Optimal Filtering in General State-Space Models},
  journal = {Stochastic Processes and their Applications},
  volume  = {115},
  number  = {8},
  pages   = {1408--1436},
  year    = {2005},
  doi     = {10.1016/j.spa.2005.03.005}
}

@incollection{bollerslev1994arch,
  author    = {Bollerslev, Tim and Engle, Robert F. and Nelson, Daniel B.},
  title     = {{ARCH Models}},
  editor    = {Engle, Robert F. and McFadden, Daniel L.},
  booktitle = {Handbook of Econometrics},
  volume    = {4},
  chapter   = {49},
  pages     = {2959--3038},
  publisher = {Elsevier},
  address   = {Amsterdam},
  year      = {1994},
  doi       = {10.1016/S1573-4412(05)80018-2}
}

@article{GustHerbstLopezSalidoSmith2017,
  author  = {Gust, Christopher and Herbst, Edward and L{\'o}pez-Salido, David and Smith, Matthew E.},
  title   = {The Empirical Implications of the Interest-Rate Lower Bound},
  journal = {American Economic Review},
  year    = {2017},
  volume  = {107},
  number  = {7},
  pages   = {1971--2006},
  doi     = {10.1257/aer.20121437}
}

@article{DanglHalling2012,
  author  = {Dangl, Thomas and Halling, Michael},
  title   = {Predictive Regressions with Time-Varying Coefficients},
  journal = {Journal of Financial Economics},
  year    = {2012},
  volume  = {106},
  number  = {1},
  pages   = {157--181},
  doi     = {10.1016/j.jfineco.2012.04.003}
}

@article{SewellChen2015,
  title={Latent Space Models for Dynamic Networks},
  author={Sewell, Daniel K. and Chen, Yuguo},
  journal={Journal of the American Statistical Association},
  volume={110},
  number={512},
  pages={1646--1657},
  year={2015},
  doi={10.1080/01621459.2014.988214}
}

@article{ahn2002qtsm,
  title={Quadratic Term Structure Models: Theory and Evidence},
  author={Ahn, Dong-Hyun and Dittmar, Robert F. and Gallant, A. Ronald},
  journal={Review of Financial Studies},
  volume={15},
  number={1},
  pages={243--288},
  year={2002},
  publisher={Oxford University Press},
  doi={10.1093/rfs/15.1.243}
}

@article{MalikPitt2011ParticleFiltersContinuous,
author  = {Malik, Sheheryar and Pitt, Michael K.},
title   = {Particle Filters for Continuous Likelihood Evaluation and Maximisation},
journal = {Journal of Econometrics},
year    = {2011},
volume  = {165},
number  = {2},
pages   = {190--209},
doi     = {10.1016/j.jeconom.2011.07.006}
}

@article{PoyiadjisDoucetSingh2011,
  title={Particle Approximations of the Score and Observed Information Matrix in State Space Models with Application to Parameter Estimation},
  author={Poyiadjis, George and Doucet, Arnaud and Singh, Sumeetpal S.},
  journal={Biometrika},
  volume={98},
  number={1},
  pages={65--80},
  year={2011},
  doi={10.1093/biomet/asq062}
}

@article{AndrieuDoucetHolenstein2010,
  title={Particle Markov Chain Monte Carlo Methods},
  author={Andrieu, Christophe and Doucet, Arnaud and Holenstein, Roman},
  journal={Journal of the Royal Statistical Society: Series B},
  volume={72},
  number={3},
  pages={269--342},
  year={2010}
}

@book{TsayChen2019NonlinearTimeSeries,
  author    = {Tsay, Ruey S. and Chen, Rong},
  title     = {Nonlinear Time Series Analysis},
  series    = {Wiley Series in Probability and Statistics},
  publisher = {John Wiley \& Sons},
  address   = {Hoboken, NJ},
  year      = {2019},
  isbn      = {9781119264057}
}

@article{HansenMayerSargent2010,
  author  = {Hansen, Lars Peter and Mayer, Ricardo and Sargent, Thomas J.},
  title   = {Robust Hidden Markov LQG Problems},
  journal = {Journal of Economic Dynamics and Control},
  year    = {2010},
  volume  = {34},
  number  = {10},
  pages   = {1951--1966},
  doi     = {10.1016/j.jedc.2010.05.004}
}

@article{LiPhillipsShiYu2025WeakIdentification,
  title   = {Weak Identification of Long Memory with Implications for Volatility Modeling},
  author  = {Li, Jia and Phillips, Peter C. B. and Shi, Shuping and Yu, Jun},
  journal = {The Review of Financial Studies},
  volume  = {38},
  number  = {10},
  pages   = {3117--3148},
  year    = {2025},
  doi     = {10.1093/rfs/hhaf022}
}

@article{JacksonPernoud2021,
  author  = {Jackson, Matthew O. and Pernoud, Agathe},
  title   = {Systemic Risk in Financial Networks: A Survey},
  journal = {Annual Review of Economics},
  year    = {2021},
  volume  = {13},
  pages   = {171--202}
}

@article{AcemogluOzdaglarTahbazSalehi2015,
  author  = {Acemoglu, Daron and Ozdaglar, Asuman and Tahbaz-Salehi, Alireza},
  title   = {Systemic Risk and Stability in Financial Networks},
  journal = {American Economic Review},
  year    = {2015},
  volume  = {105},
  number  = {2},
  pages   = {564--608}
}

@article{chen2014sieve,
  author  = {Chen, Xiaohong and Liao, Zhipeng},
  title   = {Sieve M Inference on Irregular Parameters},
  journal = {Journal of Econometrics},
  year    = {2014},
  volume  = {182},
  number  = {1},
  pages   = {70--86},
  doi     = {10.1016/j.jeconom.2014.04.009}
}

@article{carr2004time,
  title={Time-Changed L{\'e}vy Processes and Option Pricing},
  author={Carr, Peter and Wu, Liuren},
  journal={Journal of Financial Economics},
  volume={71},
  number={1},
  pages={113--141},
  year={2004},
  publisher={Elsevier}
}

@article{Reiter2009,
  author  = {Reiter, Michael},
  title   = {Solving Heterogeneous-Agent Models by Projection and Perturbation},
  journal = {Journal of Economic Dynamics and Control},
  year    = {2009},
  volume  = {33},
  number  = {3},
  pages   = {649--665},
  doi     = {10.1016/j.jedc.2008.08.010}
}

@book{DungTemlyakovUllrich2018HyperbolicCross,
  author    = {D{\~u}ng, Dinh and Temlyakov, Vladimir N. and Ullrich, Tino},
  title     = {Hyperbolic Cross Approximation},
  series    = {Advanced Courses in Mathematics -- CRM Barcelona},
  publisher = {Birkh{\"a}user},
  address   = {Cham},
  year      = {2018},
  doi       = {10.1007/978-3-319-92240-9},
  isbn      = {978-3-319-92239-3}
}

@article{BungartzGriebel2004SparseGrids,
  author  = {Bungartz, Hans-Joachim and Griebel, Michael},
  title   = {Sparse Grids},
  journal = {Acta Numerica},
  volume  = {13},
  pages   = {147--269},
  year    = {2004},
  doi     = {10.1017/S0962492904000182}
}

@article{Trefethen2017Cubature,
  author  = {Trefethen, Lloyd N.},
  title   = {Cubature, Approximation, and Isotropy in the Hypercube},
  journal = {SIAM Review},
  volume  = {59},
  number  = {3},
  pages   = {469--491},
  year    = {2017},
  doi     = {10.1137/16M1066312}
}

@article{OlleyPakes1996,
  title   = {The Dynamics of Productivity in the Telecommunications Equipment Industry},
  author  = {Olley, G. Steven and Pakes, Ariel},
  journal = {Econometrica},
  volume  = {64},
  number  = {6},
  pages   = {1263--1297},
  year    = {1996}
}

@article{DoraszelskiJaumandreu2013,
  title   = {R{\&}D and Productivity: Estimating Endogenous Productivity},
  author  = {Doraszelski, Ulrich and Jaumandreu, Jordi},
  journal = {The Review of Economic Studies},
  volume  = {80},
  number  = {4},
  pages   = {1338--1383},
  year    = {2013},
  doi     = {10.1093/restud/rdt011}
}

@article{HallYao2003,
  author  = {Hall, Peter and Yao, Qiwei},
  title   = {Inference in {ARCH} and {GARCH} Models with Heavy-Tailed Errors},
  journal = {Econometrica},
  year    = {2003},
  volume  = {71},
  number  = {1},
  pages   = {285--317},
  doi     = {10.1111/1468-0262.00396}
}

@article{FernandezVillaverdeRubioRamirez2007,
  author  = {Fern{\'a}ndez-Villaverde, Jes{\'u}s and Rubio-Ram{\'i}rez, Juan F.},
  title   = {Estimating Macroeconomic Models: A Likelihood Approach},
  journal = {The Review of Economic Studies},
  year    = {2007},
  volume  = {74},
  number  = {4},
  pages   = {1059--1087},
  doi     = {10.1111/j.1467-937X.2007.00437.x}
}

@article{KimShephardChib1998,
  author  = {Kim, Sangjoon and Shephard, Neil and Chib, Siddhartha},
  title   = {Stochastic Volatility: Likelihood Inference and Comparison with {ARCH} Models},
  journal = {The Review of Economic Studies},
  year    = {1998},
  volume  = {65},
  number  = {3},
  pages   = {361--393},
  doi     = {10.1111/1467-937X.00050}
}

@article{KaplanMollViolante2018,
  title   = {Monetary Policy According to HANK},
  author  = {Kaplan, Greg and Moll, Benjamin and Violante, Giovanni L.},
  journal = {American Economic Review},
  volume  = {108},
  number  = {3},
  pages   = {697--743},
  year    = {2018},
  doi     = {10.1257/aer.20160042}
}

@article{Sargent2026,
author = {Sargent, Thomas J. and Selvakumar, Yatheesan J. and Yang, Ziyue},
title = {Dynamic Mode Decompositions and Vector Autoregressions},
journal = {International Economic Review},
 year      = {2026},
doi = {https://doi.org/10.1111/iere.70068},
eprint = {https://onlinelibrary.wiley.com/doi/pdf/10.1111/iere.70068}
}

@book{DurbinKoopman2012,
  author    = {Durbin, James and Koopman, Siem Jan},
  title     = {Time Series Analysis by State Space Methods},
  edition   = {2},
  publisher = {Oxford University Press},
  address   = {Oxford},
  year      = {2012},
  doi       = {10.1093/acprof:oso/9780199641178.001.0001}
}

@article{AitSahaliaCachoDiazLaeven2015,
  author  = {Ait-Sahalia, Yacine and Cacho-Diaz, Julio and Laeven, Roger J. A.},
  title   = {Modeling Financial Contagion Using Mutually Exciting Jump Processes},
  journal = {Journal of Financial Economics},
  year    = {2015},
  volume  = {117},
  number  = {3},
  pages   = {585--606},
  doi     = {10.1016/j.jfineco.2015.03.002}
}

@inproceedings{KatharopoulosVyasPappasFleuret2020,
  author    = {Katharopoulos, Angelos and Vyas, Apoorv and Pappas, Nikolaos and Fleuret, Fran{\c{c}}ois},
  title     = {Transformers are {RNN}s: Fast Autoregressive Transformers with Linear Attention},
  booktitle = {Proceedings of the 37th International Conference on Machine Learning},
  year      = {2020},
  volume    = {119},
  series    = {Proceedings of Machine Learning Research},
  pages     = {5156--5165},
  publisher = {PMLR},
}

@article{Bates2019JF,
  author  = {Bates, David S.},
  title   = {How Crashes Develop: Intradaily Volatility and Crash Evolution},
  journal = {The Journal of Finance},
  year    = {2019},
  volume  = {74},
  number  = {1},
  pages   = {193--238},
  month   = feb,
  doi     = {10.1111/jofi.12732}
}

@book{cappe2005inference,
  title     = {Inference in Hidden Markov Models},
  author    = {Capp{\'e}, Olivier and Moulines, Eric and Ryd{\'e}n, Tobias},
  year      = {2005},
  publisher = {Springer},
  series    = {Springer Series in Statistics},
  address   = {New York}
}

@article{bates2006maximum,
  author  = {Bates, David S.},
  title   = {Maximum Likelihood Estimation of Latent Affine Processes},
  journal = {The Review of Financial Studies},
  volume  = {19},
  number  = {3},
  pages   = {909--965},
  year    = {2006},
  doi     = {10.1093/rfs/hhj022}
}

@article{chen2012estimation,
  author  = {Chen, Xiaohong and Pouzo, Demian},
  title   = {Estimation of Nonparametric Conditional Moment Models with Possibly Nonsmooth Generalized Residuals},
  journal = {Econometrica},
  volume  = {80},
  number  = {1},
  pages   = {277--321},
  year    = {2012},
  doi     = {10.3982/ECTA7888}
}

@article{barndorff2001non,
  author  = {Barndorff-Nielsen, Ole E. and Shephard, Neil},
  title   = {Non-Gaussian Ornstein--Uhlenbeck-Based Models and Some of Their Uses in Financial Economics},
  journal = {Journal of the Royal Statistical Society: Series B},
  year    = {2001},
  volume  = {63},
  number  = {2},
  pages   = {167--241},
  doi     = {10.1111/1467-9868.00282}
}

@article{chen2022factor,
  author  = {Chen, Rong and Yang, Dan and Zhang, Cun-Hui},
  title   = {Factor Models for High-Dimensional Tensor Time Series},
  journal = {Journal of the American Statistical Association},
  year    = {2022},
  volume  = {117},
  number  = {537},
  pages   = {94--116},
  doi     = {10.1080/01621459.2021.1912757}
}

@article{hamilton1989new,
  author  = {Hamilton, James D.},
  title   = {A New Approach to the Economic Analysis of Nonstationary Time Series and the Business Cycle},
  journal = {Econometrica},
  year    = {1989},
  volume  = {57},
  number  = {2},
  pages   = {357--384},
  doi     = {10.2307/1912559}
}

@article{Pouzo2022,
author = {Pouzo, Demian and Psaradakis, Zacharias and Sola, Martin},
title = {Maximum Likelihood Estimation in Markov Regime-Switching Models With Covariate-Dependent Transition Probabilities},
journal = {Econometrica},
volume = {90},
number = {4},
pages = {1681-1710},
doi = {https://doi.org/10.3982/ECTA17249},
eprint = {https://onlinelibrary.wiley.com/doi/pdf/10.3982/ECTA17249},
year = {2022}
}

@book{rudin1991functional,
  author    = {Walter Rudin},
  title     = {Functional Analysis},
  edition   = {2nd},
  year      = {1991},
  publisher = {McGraw-Hill},
  address   = {New York},
  isbn      = {0-07-100944-2}
}

@article{gland2004,
author = {Fran{\c{c}}ois Le Gland and Nadia Oudjane},
title = {{Stability and uniform approximation of nonlinear filters using the Hilbert metric and application to particle filters}},
volume = {14},
journal = {The Annals of Applied Probability},
number = {1},
publisher = {Institute of Mathematical Statistics},
pages = {144 -- 187},
year = {2004},
doi = {10.1214/aoap/1075828050},
}

@book{zygmund2002trigonometric,
  author    = {Zygmund, Antoni},
  title     = {Trigonometric Series},
  publisher = {Cambridge University Press},
  address   = {Cambridge},
  year      = {2002},
  series    = {Cambridge Mathematical Library},
  edition   = {3}
}

@book{timan1963theory,
  author    = {Timan, A. F.},
  title     = {Theory of Approximation of Functions of a Real Variable},
  publisher = {Pergamon Press},
  address   = {Oxford},
  year      = {1963},
  series    = {International Series of Monographs in Pure and Applied Mathematics},
  volume    = {34},
  note      = {Translated by J. Berry; English translation edited by J. Cossar}
}

@book{jackson1930theory,
  title={The theory of approximation},
  author={Jackson, Dunham},
  volume={11},
  year={1930},
  publisher={American Mathematical Soc.}
}

@book{duoandikoetxea2024fourier,
  title={Fourier analysis},
  author={Duoandikoetxea, Javier},
  volume={29},
  year={2024},
  publisher={American Mathematical Society}
}

@article{duffie2000transform,
  title={Transform analysis and asset pricing for affine jump-diffusions},
  author={Duffie, Darrell and Pan, Jun and Singleton, Kenneth},
  journal={Econometrica},
  volume={68},
  number={6},
  pages={1343--1376},
  year={2000},
  publisher={Wiley Online Library}
}

@book{billingsley1995probability,
  title={Probability and Measure},
  author={Billingsley, P.},
  isbn={9780471007104},
  lccn={gb95051456},
  series={Wiley Series in Probability and Statistics},
  year={1995},
  publisher={Wiley}
}

@book{vanderVaart1998,
  author    = {van der Vaart, A. W.},
  title     = {Asymptotic Statistics},
  year      = {1998},
  publisher = {Cambridge University Press},
  address   = {Cambridge, UK},
  isbn      = {978-0521496032},
  note      = {see Chapters 5, 8, 12 for likelihood asymptotics and CLT},
}

@book{hall1980martingale,
  author    = {Hall, Peter and Heyde, C. C.},
  title     = {Martingale Limit Theory and Its Application},
  year      = {1980},
  publisher = {Academic Press},
  address   = {New York, NY},
  isbn      = {978-0123865508},
  series    = {Probability and Mathematical Statistics},
  note      = {Classic reference for martingale central limit theorems}
}

@article{bartlett1953a,
  author  = {Bartlett, M. S.},
  title   = {Approximate Confidence Intervals: I},
  journal = {Biometrika},
  year    = {1953},
  volume  = {40},
  pages   = {12--19},
  doi     = {10.1093/biomet/40.12},
}

@article{bartlett1953b,
  author  = {Bartlett, M. S.},
  title   = {Approximate Confidence Intervals: II. More than One Unknown Parameter},
  journal = {Biometrika},
  year    = {1953},
  volume  = {40},
  pages   = {306--317},
  doi     = {10.1093/biomet/40.3-4.306},
}

@book{meyn2009markov,
  title     = {Markov Chains and Stochastic Stability},
  author    = {Meyn, S. P. and Tweedie, R. L.},
  year      = {2009},
  publisher = {Cambridge University Press},
  address   = {Cambridge, UK},
  isbn      = {978-0521849783},
  note      = {Second edition. A comprehensive reference for ergodic Markov chains, geometric ergodicity, Poisson equations and stability conditions},
}

@article{Hawkes_1971, title={Spectra of some self-exciting and mutually exciting point processes}, volume={58}, DOI={10.2307/2334319}, number={1}, journal={Biometrika}, author={Hawkes, Alan G.}, year={1971}, month={Apr}, pages={83}}
	}

\end{document}